\documentclass[qe,nameyear,draft]{econsocart} 

\usepackage{amsmath,blkarray,bm}
\usepackage{amsfonts} 
\usepackage{amsthm} 
\usepackage{amssymb}
\usepackage{array} %
\usepackage{bbm} 
\usepackage{booktabs}
\usepackage{caption}
\usepackage{centernot}
\usepackage{colortbl}
\usepackage{etoolbox}
\usepackage{enumitem}
\usepackage[mathscr]{euscript}
\usepackage{amssymb}
\usepackage{subcaption}
\usepackage{hhline}
\usepackage{pgfplots} 
\usepackage{pdfsync} 
\usepackage{longtable}
\usepackage{tabularx}
\usepackage{rotating}
\usepackage{threeparttable}
\usepackage{mathtools}
\usepackage{multirow}
\usepackage{setspace} 
\usepackage{scalerel}
\usepackage{stmaryrd}
\usetikzlibrary{shapes.geometric, arrows,positioning}
\usetikzlibrary{arrows}
\usepackage{tikz}
\usepackage{xr} 
\usepackage{xcolor}

\usetikzlibrary{positioning}

\usepackage{lscape}
\usepackage{arydshln}
\usepackage{amssymb}
\usepackage{algorithm}
\usepackage{algpseudocode}

\makeatletter
\newcommand*{\addFileDependency}[1]{
	\typeout{(#1)}
	\@addtofilelist{#1}
	\IfFileExists{#1}{}{\typeout{No file #1.}}
}\makeatother

\newcommand*{\myexternaldocument}[1]{%
	\externaldocument{#1}%
	\addFileDependency{#1.tex}%
	\addFileDependency{#1.aux}%
}

\myexternaldocument{supplement}

\RequirePackage[colorlinks,citecolor=blue,linkcolor=blue,urlcolor=blue,pagebackref]{hyperref}

\startlocaldefs

\numberwithin{equation}{section}
\theoremstyle{plain}
\newtheorem{assumption}{Assumption} 
\newtheorem{theorem}{Theorem}
\newtheorem{lemma}{Lemma}
 
\newtheorem{corollary}{Corollary}[section]  

\theoremstyle{definition}
 
\numberwithin{assumption}{section}
\numberwithin{lemma}{section}
\numberwithin{definition}{section}

\theoremstyle{remark}
\newtheorem{remark}{Remark}

\def\argmax{\mathop{\rm arg\,max}}

\endlocaldefs

\begin{document}

	\renewcommand\thefootnote{\alph{footnote}}
	
	\begin{frontmatter}
		\title{Robust Priors in Nonlinear Panel Models with Individual and Time Effects}
		\runtitle{ } 
		\smallskip
		\begin{aug}
			Zizhong Yan\footnote{Institute for Economic and Social Research, Jinan University. Email: helloyzz@gmail.com.}, \qquad 
			Zhengyu Zhang\footnote{$^{,*}$ Corresponding author. School of Economics, Shanghai University of Finance and Economics. Email: zy.zhang@mail.shufe.edu.cn.}$^{,*}$, \qquad
			Mingli Chen\footnote{Department of Economics, University of Warwick. Email: m.chen.3@warwick.ac.uk.} \\
			Jingrong Li\footnote{College of Economics and Management, South China Agricultural University. Email: lijingrong@scau.edu.cn.}, \qquad
			Iv{\'a}n Fern{\'a}ndez-Val\footnote{Department of Economics, Boston University. Email: ivanf@bu.edu.} 
		\end{aug}

		
		\support{We are grateful to Manuel Arellano, St{\'e}phane Bonhomme, Han Hong, and Ji-Liang Shiu for helpful comments and suggestions. Yan acknowledges the National Natural Science Foundation of China (No. 72103079). Li and Yan acknowledge the  Grant for Humanities and Social Sciences Research of the Ministry of Education (No. 25YJC91000 and 25JDSZ3124). }
		
		\begin{abstract}
			We develop likelihood-based bias reduction for nonlinear panel models with additive individual and time effects. In two-way panels, integrated-likelihood corrections are attractive but challenging because the required integration is high dimensional and standard Laplace approximations may fail when the parameter dimension grows with the sample size. We propose a target-centered full-exponential Laplace--cumulant expansion that exploits the sparse higher-order derivative structure implied by additive effects, delivering a tractable approximation with a negligible remainder under large-$N,T$ asymptotics.
			The expansion motivates robust priors that yield bias reduction for both common parameters and fixed effects. We provide implementations for binary, ordered, and multinomial response models with two-way effects. For average partial effects, we show that the remaining first-order bias has a simple variance form and can be removed by a closed-form adjustment. Monte Carlo experiments and an empirical illustration show substantial bias reduction with accurate inference.
			
		\end{abstract}

	\end{frontmatter}
	
	
	\newpage
	\setcounter{footnote}{0}
	\renewcommand\thefootnote{\arabic{footnote}}

\section{Introduction}\label{sec:intro}
Nonlinear panel models with additive unobserved individual and time effects are widely used to control for persistent heterogeneity and aggregate shocks.
Yet in nonlinear models, fixed effects estimation that treats these unobserved effects as nuisance parameters to be estimated suffers from the incidental parameters problem \citep{neyman1948consistent}. 
Estimating $N+T$ nuisance parameters from $NT$ observations produces a first-order bias of order $\max\{N^{-1},T^{-1}\}$ in the other common model parameters, and this bias propagates to structural functions such as average partial effects (APEs).

Likelihood-based bias corrections modify the criterion function to produce first-order debiased estimators of model parameters.  
Among them, the \textit{integrated likelihood} approach, which averages the likelihood over fixed effects using a weight function (or prior), is particularly attractive.
Because many structural quantities, including APEs, are themselves defined through integration, debiasing the model parameters carries over naturally into bias reduction for partial effects within the same integrated-likelihood framework.
Moreover, by an appropriate choice of the prior, the same device can also absorb the higher-order bias from estimating the fixed effects, which makes the subsequent APE correction especially transparent.
In one-way panels, this approach is tractable because the likelihood separates at the unit level, allowing stable evaluation and debiasing \citep{arellano2009robust,pakel2019bias}. 
With two-way effects, however, separability is lost and the integration becomes $(N+T)$-dimensional yielding the classical Laplace expansions invalid. 
As a result, despite its conceptual simplicity, a rigorous integrated-likelihood debiasing method for two-way panels has been missing.

This paper develops such a framework under the large-$N$, large-$T$ asymptotic scheme of \citet{fernandez2016individual}.
Our starting point is the maximum integrated-likelihood estimator, obtained by integrating out the $(N+T)$ fixed effects. 
Since direct evaluation of this integral is infeasible in nonlinear two-way panels, we implement the correction through a bias-reducing adjustment to the objective function, which can be interpreted as a \textit{robust prior} for the integrated likelihood. 
Our analysis is entirely frequentist and characterizes the sampling behavior of the resulting estimators.

Justifying this strategy requires a new asymptotic theory. 
With two-way fixed effects, the integrated likelihood involves a high-dimensional integration, and standard Laplace expansions for posterior moments can fail when the parameter dimension grows with the sample size \citep{shun_mccullagh_1995}. 
We resolve this difficulty by developing a \emph{target-centered full-exponential} Laplace approximation that expands directly around the true parameter values, rather than the sample modes, and isolates the leading bias created  by integrating out the fixed effects.
This approximation turns bias reduction into a constructive problem:   choose a prior whose gradient offsets the leading score bias  induced by integrating out the fixed effects.

To validate the approximation, we develop a \emph{cumulant-based expansion} that exploits sparsity in additive two-way models. 
For each observation, only few third- and higher-order log-likelihood derivative terms are nonzero. 
This sparsity limits the tensor contractions that drive the high-order remainder even as the fixed-effects dimension diverges.
The resulting expansion yields asymptotic linear representations for posterior means of the full model parameter vector and for the associated APE estimators.

These representations have two practical implications. 
First, they guide the construction of  bias-reducing priors for both model parameters and fixed effects.
In particular, the prior removes the leading term of the bias for the model parameters and fixed effects. 
A key message is that debiasing  the fixed effects is essential for debiasing APEs.
Second, the same expansion characterizes APE estimators under the same priors. Our results show that, even after applying the bias-reducing priors, APE estimators generally retain an additional first-order component driven by fixed-effect estimation uncertainty. We show that this remaining bias has a simple variance form and provide a closed-form correction.

We also clarify the conceptual relationship between integrated and joint likelihood corrections. It is often assumed that priors for the integrated likelihood and penalties for the joint likelihood are interchangeable, since both can be interpreted as Bayesian priors. We show that integration introduces an additional curvature term, so the bias-reducing prior for the integrated likelihood generally differs from the penalty that debiases the joint likelihood.\footnote{A notable exception arises in dynamic AR($p$) panels: we find the same closed-form object delivers both the integrated-likelihood robust prior and the joint-likelihood penalty.} We provide the corresponding joint-likelihood penalties as a companion result.

Finally, we derive model-specific priors and provide a \emph{Guidance for Practitioners} for workhorse specifications. While the generic robust priors work broadly, exploiting specific model priors yields improvements in finite-sample performance.
With strictly exogenous regressors, exponential-family panels (e.g., binary logit and Poisson) admit closed-form, \textit{non-data-dependent} priors that are particularly stable in finite samples.
When regressors are predetermined, as in dynamic specifications with lagged outcomes, we propose a dedicated prior that improves small-$N$ performance by absorbing the cross-sectional component of the leading bias. When $T$ is also small, a simple add-on adjustment addresses the remaining cross-time component.
For dynamic AR($m$) panels, we derive an explicit fixed-$T$ consistent prior based on the roots of the AR polynomial, and we show that the same likelihood-based correction carries over to general AR($m$) specifications. 
For probit panels, comparable closed-form simplifications are unavailable, but the generic priors remain valid. 
Beyond these cases, the approach extends directly to richer likelihoods: we implement it for multinomial and ordinal response panels with two-way fixed effects and, to our knowledge, provide the first debiasing results for these models in the presence of two-way effects.\footnote{
Recently, \citet{honore2025dynamic} develop fixed-$T$ identification and GMM estimation for a dynamic ordered logit with individual fixed effects. Other existing work adopt random effects \citep{albarran2019estimation} or focus on large-$(N,T)$ debiasing in models with one-way effects  \citep{carro2014state,fernandez2017evaluating}. Our focus is large-$(N,T)$ debiasing under two-way effects. 
}

On the computation side, we propose a scalable Metropolis--Hastings algorithm with random blocking and self-adaptive proposals. This algorithm efficiently handles the high-dimensional parameter space, making integrated-likelihood inference practical for real-world datasets. We provide a Python package, \texttt{twowaypanel}, which covers the main model classes studied in our simulations and empirical illustration.\footnote{
    The \texttt{twowaypanel} package (source code and the replication archive) is available at \url{github.com/zizhongyan/twowaypanel}. 
    Online documentation is maintained at \url{twowaypanel.readthedocs.io}.}

\textbf{Literature review}.
This paper relates to the literature on the incidental parameters problem in nonlinear panels.
While early work emphasized short panels and focused on fixed-$T$ consistency \citep[e.g.,][]{chamberlain1980analysis,arellano1991some},
a complementary line adopts the large-$N,T$ framework in which fixed-effects estimators are consistent but display non-negligible first-order bias \citep{hahn2004jackknife,hahn2011bias,dhaene2015split,kim2016bootstrap,alvarez2022robust,higgins2024bootstrap}.
For nonlinear models with individual and time effects, our paper relies on the rigorous asymptotic framework of \citet{fernandez2016individual,chen2021nonlinear}; see also \citet{fernandez2018fixed} for a broader review. Our contribution connects to three main strands of this literature.

The first strand develops bias corrections based on the integrated likelihood  in one-way panels.
\citet{lancaster2002orthogonal} and \citet{arellano2009robust} characterize bias-reducing  priors and show how modified likelihood objectives can be interpreted as integrated likelihoods; related perspectives appear in \citet{woutersen2002robustness} and \citet{bester2005bias}.  \citet{pakel2019bias} extends robust-prior ideas to models with state dependence and \citet{schumann2021integrated} proposes a debiased method based on transforming the likelihood function.
We  deal instead with two-way panels, work with the standard likelihood, 
and develop corrections for both structural parameters and APEs.

Second, a related literature debiases the joint/profile likelihood via penalties, including the general penalized likelihood principle \citep{firth1993bias,bester2005bias}, panel-specific likelihood adjustments \citep{bester2009penalty,arellano2016likelihood}, and second-order refinements \citep{dhaene2021second}.  
Our analysis shows that, with two-way effects, the bias-reducing prior for the integrated likelihood generally differs from the penalty that debiases the joint likelihood, and we derive the latter as a companion result.

The third strand addresses models with two-way fixed effects. Existing solutions include the analytical and jackknife bias corrections of \citet{fernandez2016individual}, the differencing method of \citet{charbonneau2017multiple} for logit/Poisson models, and the recent likelihood-based corrections for dynamic heterogeneous panels by \citet{leng2025debiased}. \citet{jochmans2019likelihood} propose a modified likelihood for two-way models and illustrate an MCMC implementation. 
While their discussion does not develop a formal large-$N,T$ characterization of the resulting posterior mean, our results provide such a characterization within an integrated-likelihood framework. 
We formally show that posterior-mean bias reduction is tied to the prior associated with the integrated likelihood, whereas bias reduction for the maximized penalized joint likelihood requires a different correction term.
This distinction also underlies our bias-corrected APE procedures for two-way panels.

\textbf{Outline of the paper}.
The rest of the paper is organized as follows.
Section~\ref{sec:modelandestimators} introduces the model and the integrated likelihood framework.
Section~\ref{sec:robust_prior} develops the target-centered approximation and robust priors, and derives large-sample distributions for parameters and APEs.
Section~\ref{sec:model_specific} provides model-specific priors and practical guidance. The remaining sections present simulation evidence and an empirical illustration, followed by concluding remarks. Proofs and computational details are collected in the Appendix.

\section{Model and Estimators}\label{sec:modelandestimators}
\subsection{Model}
Let the observed data be denoted by $\left\{\left(Y_{it},X_{it}\right): i=1,\ldots,N, t=1,\ldots,T\right\}$, where $Y_{it}$ is the outcome variable and $X_{it}$ is a vector of explanatory variables. 
Observations are independent across  $i$ and weakly dependent in $t$.
The data are generated sequentially over time according to
\[
Y_{it}\mid X_i^t,\alpha,\gamma,\beta \sim f_Y\left(\cdot | X_{it};\ \beta,  \alpha_i+\gamma_t\right),
\quad
i=1,\ldots,N,\ t=1,\ldots,T,
\]
where $f_Y$ is a known probability function, $X_i^t=(X_{i1},\ldots,X_{iT})$, $\beta$ collects the finite dimensional \textit{model parameters} of interest,
and $(\alpha',\gamma')'=: \phi$ stacks individual and time effects with $\alpha=(\alpha_1,\ldots,\alpha_N)'$ and $\gamma=(\gamma_1,\ldots,\gamma_T)'$.
The regressors $X_{it}$ may be strictly exogenous or predetermined with respect to $Y_{it}$. Note that $X_{it}$ can include lags of $Y_{it}$ to allow dynamics.

We adopt the fixed effects approach for estimation, treating unobserved individual and time effects as parameters to estimate. 
The \emph{joint} log–likelihood is
\[
\mathcal L(\beta,\phi)
=\tfrac{1}{\sqrt{NT}}\left\{\textstyle\sum_{i=1}^N\sum_{t=1}^T \ln f_Y\left(Y_{it}\,\middle|\,X_{it};\ \beta, \alpha_i+\gamma_t\right)
-\frac{b}{2}\,(v'\phi)^2\right\},
\]
where $v=(\iota_N',-\iota_T')'$ and $b>0$ is an arbitrary constant. The quadratic term $-\frac{b}{2}(v'\phi)^2$ is included as a normalization device for the two-way effects. Adding a constant to all $\alpha_i$ and subtracting it from all $\gamma_t$ leaves $\alpha_i+\gamma_t$ unchanged, so $\phi$ is unidentified in one location direction. Following \citet{fernandez2016individual}, this term selects a normalized representative (e.g., $v'\phi=0$) and yields a unique maximizer in $\phi$ without affecting $\beta$.\footnote{Equivalently, one may impose $\sum_i \alpha_i=\sum_t \gamma_t$ or fix one effect such as $\alpha_1=0$ or $\gamma_1=0$).}
The factor $1/\sqrt{NT}$ is a normalization that leaves maximizers unchanged and simplifies the stochastic expansions used below.
Our main focus is on \emph{integrated} likelihood corrections. Given a prior   $p(\beta,\phi)$, the integrated log–likelihood is
\[
\mathcal L^{I}(\beta)
=\tfrac{1}{\sqrt{NT}}\ln \left(\textstyle\int \exp\!\big(\sqrt{NT}\,\mathcal L(\beta,\phi)\big)\,p(\beta,\phi) \, \mathrm{d}\phi.  \right)
\]

In parallel, one may add a penalty directly to the joint likelihood. Writing $p^J(\beta,\phi)$ for a prior (or penalty) used \emph{with} $\mathcal L(\beta,\phi)$, the penalized objective is
$\mathcal L(\beta,\phi)+\sqrt{NT}\,\log p^J(\beta,\phi).$\footnote{Section \ref{sec:subsecPenaltyforJoint} will return to the joint–likelihood route, and review the robust penalty used with the joint likelihood and clarify how it differs from the robust prior for the integrated likelihood. In Section \ref{sec:subsecGuidance}, we give model–specific penalty forms that complement the model–specific priors.}
From a probabilistic standpoint, both $p(\beta,\phi)$ and $p^J(\beta,\phi)$ can be read as priors.\footnote{
    The marginal posterior of $\beta$ under a prior $p(\beta, \phi)$ satisfies
    $ f(\beta|Y,X) \propto \int  f (Y|\beta,\phi;X) p(\beta,\phi)  \mathrm{d} \phi $,
    which is precisely the \emph{integrated} likelihood weighted by $p(\beta, \phi)$. 
    A Bayesian maximum a posteriori estimator based on the \emph{joint} likelihood uses
    $ f(\beta,\phi|Y,X) \propto f(Y|\beta,\phi;X) \times p^{J}(\beta,\phi),$
    so $p^J$ plays the role of a prior. 
    Thus both $p$ and $p^J$ enter as priors: maximizing the penalized log-likelihood is a MAP problem, while maximizing $\mathcal{L}^I(\beta)$ corresponds to integrating the prior out.
}
When the purpose is bias reduction, however, Section~\ref{sec:robust_prior} shows that the prior that removes the leading bias in the integrated likelihood generally \emph{differs} from the penalty for the joint likelihood, due to the unique structure of the panel data model. We therefore keep the notational distinction.

We define the true parameters $(\beta_0,\phi_0) $ as the solution to the conditional population problem
$
\max_{\beta,\phi} \mathbb{E}_{\phi}\left[ \mathcal{L}(\beta,\phi)\right],
$
where $\mathbb E_{\phi}$ denotes the expectation taken with respect to the distribution of the data conditional on the unobserved effects and initial conditions (including strictly exogenous regressors).

Finally, the quantities of primary economic interest are structural functions, most notably, the average partial effect (APE).  Let $\Delta(X_{it},\beta,\alpha_i+\gamma_t)$ be a per observation partial effect that depends on the fixed effects only through $\alpha_i+\gamma_t$. The sample APE at $(\beta,\phi)$ is
\[
\Delta(\beta,\phi)=\tfrac{1}{NT}\textstyle\sum_{i=1}^N\sum_{t=1}^T \Delta(X_{it},\beta,\alpha_i+\gamma_t),
\]
and the \textit{population} target APE is $\mathbb E[\Delta(\beta_0,\phi_0)]$, where the expectation $\mathbb E$ integrates over both the data and  unobserved effects, provided that the expectation exists.

\subsection{Estimators and likelihood-based corrections}
In nonlinear panel models, estimating high-dimensional fixed effects induces an incidental parameter bias of order $\max\{T^{-1},N^{-1}\}$ in model parameter estimators \citep[e.g.,][]{fernandez2016individual}. 
For any fixed $\beta$, we distinguish the \emph{sample mode} of the fixed effects
$
\widehat\phi(\beta)=\arg\max_{\phi}\ \mathcal L(\beta,\phi),
$
from the \emph{target mode}
$
\overline\phi(\beta)=\arg\max_{\phi}\ \mathbb E_{\phi}\!\left[\mathcal L(\beta,\phi)\right].
$
The concentrated and target log–likelihoods are then $\mathcal L\!\big(\beta,\widehat\phi(\beta)\big)$ and $\mathcal L\!\big(\beta,\overline\phi(\beta)\big)$, respectively. 
Under standard conditions, $\widehat\phi(\beta)$ converges to $\overline\phi(\beta)$ as $N,T\to\infty$, but its convergence rate is too slow relative to the sample size $NT$, and induces a bias of order $\max\{T^{-1},N^{-1}\}$ in $\widehat\beta=\arg\max_{\beta}\mathcal L(\beta,\widehat\phi(\beta))$.
In contrast, the hypothetical maximizer based on the infeasible target mode, $\arg\max_{\beta}\ \mathcal L\big(\beta,\overline\phi(\beta)\big)$, does not suffer from this first–order bias. These target quantities are our benchmark and guide the bias corrections that follow.

Our main focus is the integrated likelihood. The associated maximum integrated likelihood estimator (MILE) is
$\widehat\beta^{I} = \arg\max_{\beta}\ \mathcal L^I(\beta)$. 
We construct bias‐reducing priors under which $\widehat\beta^{I}$ is free of first–order bias. 
Direct evaluation of $\mathcal L^{I}(\beta)$ is infeasible with two–way effects because it requires integration over an $(N{+}T)$–dimensional $\phi$.\footnote{
    With one–way effects the likelihood factorizes by unit and the integral can be computed at the “atom” level, which admits stable quadrature  \citep{arellano2009robust,pakel2019bias}. Two–way effects couple individuals and time periods and break this factorization, so the integral is high dimensional.
} 
We therefore work with the joint posterior for $(\beta,\phi)$ induced by $p(\beta,\phi)$, and propose to use the \emph{posterior means}\footnote{
    Posterior \emph{mode} is a possible alternative. However, computing them typically requires an additional nonparametric approximation, and variance estimation is more involved.
}
\begin{align*}
    \widehat{\beta}^E
    =\tfrac{\iint \beta\ \exp\!\big(\sqrt{NT}\,\mathcal L(\beta,\phi)\big)\,p(\beta,\phi)\,d\phi\,d\beta}
    {\iint \exp\!\big(\sqrt{NT}\,\mathcal L(\beta,\phi)\big)\,p(\beta,\phi)\,d\phi\,d\beta},\qquad
    \widehat{\phi}^E
    =\tfrac{\iint \phi\ \exp\!\big(\sqrt{NT}\,\mathcal L(\beta,\phi)\big)\,p(\beta,\phi)\,d\phi\,d\beta}
    {\iint \exp\!\big(\sqrt{NT}\,\mathcal L(\beta,\phi)\big)\,p(\beta,\phi)\,d\phi\,d\beta}.
\end{align*}
Under the priors developed in Sections~\ref{sec:robust_prior}–\ref{sec:model_specific}, $\widehat{\beta}^E$ removes the leading bias without inflating asymptotic variance. 
In contrast to standard corrections that target $\beta$ only, the same priors eliminate the leading bias in $\widehat{\phi}^E$,  which is useful for bias correction of nonlinear functionals such as APEs. 

Under the integrated–likelihood setup, we  consider two APE estimators.
The first estimator is the \emph{posterior mean APE}\footnote{The estimator $\widehat\Delta^E$ is referred to as the Bayesian fixed effect APE  in \cite{arellano2009robust}. Related posterior-based average effects are studied by \citet{bonhomme2022posterior}, which has a different focus on posterior conditioning and distributional robustness in the random effects setting.}
$$
\widehat\Delta^E
=\tfrac{\iint \Delta(\beta,\phi)\ \exp\!\big(\sqrt{NT}\ \mathcal L(\beta,\phi)\big)\,p(\beta,\phi)\,d\phi\,d\beta}
{\iint \exp\!\big(\sqrt{NT}\ \mathcal L(\beta,\phi)\big)\,p(\beta,\phi)\,d\phi\,d\beta}.
$$
The second is the \emph{posterior plug–in APE}, $\Delta(\widehat\beta^{E},\widehat\phi^{E})$, which evaluates the APE at the posterior means $(\widehat\beta^{E},\widehat\phi^{E})$. Section~\ref{sec:robust_prior} shows that, once \textit{both} leading biases in estimating $(\beta,\phi)$ are removed via the integrated-likelihood based correction, the remaining APE bias has a simple variance–type form that can be subtracted in closed form. 

We approximate the posterior integrals by Markov chain Monte Carlo (MCMC). 
Appendix~\ref{app:mcmc} describes a scalable random-blocking, self-adaptive Metropolis--Hastings algorithm for the high-dimensional two-way setting.

For completeness we also consider penalized joint likelihood. The maximum penalized likelihood estimator is $(\widehat{\beta}^P, \widehat{\phi}^P) = \argmax_{\beta,\phi}\big(\mathcal L(\beta,\phi) + \sqrt{NT}\ln p^J(\beta,\phi)\big).$ One can analogously report the plug–in APE estimator $\Delta(\widehat\beta^{P},\widehat\phi^{P})$. 
Section \ref{sec:model_specific} will discuss model‐specific penalties that parallel our bias-reducing priors.

\subsection{Notation and Assumptions}\label{subsec:notation}
We write
$\mathcal S(\beta,\phi)=\partial_{\phi}\mathcal L(\beta,\phi)$ for the score in the fixed–effects direction and
$\mathcal H(\beta,\phi)= - \partial_{\phi\phi'}\mathcal L(\beta,\phi)$ for the corresponding negative Hessian, 
where the notation $\partial_{x}f$ denotes the partial derivative of $f$ with respect to $x$, and additional subscripts denote higher-order partial derivatives.
Bars denote conditional expectations given $\phi$, e.g.,
$\partial_{\beta}\overline{\mathcal L}=\mathbb E_{\phi}[\partial_{\beta}\mathcal L]$,
and tildes denote deviations, e.g.,
$\partial_{\beta}\widetilde{\mathcal L}=\partial_{\beta}\mathcal L-\partial_{\beta}\overline{\mathcal L}$.
When no confusion arises, arguments are omitted and functions are understood to be evaluated at the truth $(\beta_0,\phi_0)$, e.g. 
$\mathcal{L}= \mathcal{L}(\beta_0,\phi_0)$. 
Additionally, let  $\pi_{it}=\alpha_i+\gamma_t$, we define the per–observation log–likelihood as 
$$
\ell_{it}(\beta,\pi_{it})=\ln f_Y\left(Y_{it}\,\middle|\,X_{it};\  \beta,  \alpha_i+\gamma_t \right).
$$

We impose the following conditions, which are identical to those in \citet{fernandez2016individual} except for the last condition of Assumption \ref{assumption:main} on the prior.\footnote{
    Assumption \ref{assumption:main}(vi) is a mild high-level requirement ensuring that the prior acts as a bias-correction term and does not alter the leading stochastic orders of the score and Hessian of the criterion function.
}
\begin{assumption}[Panel Models]\label{assumption:main}
    Denote $\pi_{it0} = \alpha_{i0} + \gamma_{t0}$. Let $\nu > 0$ and $\mu > \frac{4(8 + \nu)}{\nu}$. Let $\epsilon > 0$ and let $B_{\epsilon}^0$ be a subset of            $\mathbb{R}^{\dim \beta + 1}$ that contains an $\epsilon$-neighborhood of $(\beta_0, \pi_{it0})$ for all $i, t, N, T$.
    \begin{enumerate}[label=(\roman*)]
        \item (Asymptotics) We consider limits of sequences where $N/T \to c^2$, $0 < c < \infty$, as $N, T \to \infty$.
        
        \item (Sampling) Conditional on $\phi$, $\{(Y_i^T, X_i^T) : 1 \leq i \leq N\}$ is independent across $i$; and for each fixed $i$, $\{(Y_{it}, X_{it}) : 1 \leq t \leq T\}$ is a $\alpha$-mixing process with mixing coefficients satisfying
        $
        \sup_i a_i(m) = O(m^{-\mu})
        $
        as $m \to \infty$, where
        $
        a_i(m) := \sup_t \sup_{A \in \mathcal{A}_t^i, B \in \mathcal{B}_{t+m}^i} \big| P(A \cap B) - P(A)P(B) \big|;
        $
        and for $Z_{it} = (Y_{it}, X_{it})$, $\mathcal{A}_t^i$ is the sigma field generated by $(Z_{it}, Z_{i,t-1}, \ldots)$, and $\mathcal{B}_{t}^i$ is the sigma field generated by $(Z_{it}, Z_{i,t+1}, \ldots)$.
        
        \item (Model) For $X_i^t = \{X_{is} : s = 1, \ldots, t\}$, we assume that for all $i, t, N, T$,
        $$
        Y_{it} \mid X_{it}, \phi, \beta \sim \exp[\ell_{it}(\beta, \alpha_i + \gamma_t)].
        $$
        The realizations of the parameters and unobserved effects that generate the observed data are denoted by $\beta_0$ and $\phi_0$.
        
        \item (Smoothness and moments) We assume that $(\beta, \pi) \to \ell_{it}(\beta, \pi)$ is four times continuously differentiable over $B_{\epsilon}^0$ a.s. The partial derivatives of $\ell_{it}(\beta, \pi)$ with respect to the elements of $(\beta, \pi)$ up to fourth order are bounded in absolute value uniformly over $(\beta, \pi) \in B_{\epsilon}^0$ by a function $H(Z_{it}) > 0$ a.s., and
        $
        \max_{i,t} \mathbb{E}_{\phi}[H(Z_{it})^{8+\nu}]
        $
        is a.s. uniformly bounded over $N, T$.
        
        \item (Concavity) For all $N, T$, $(\beta, \phi) \to \mathcal{L}(\beta, \phi)$ is strictly concave over $\mathbb{R}^{\dim \beta + N + T}$ a.s. Furthermore, there exist constants $c_{\min}$ and $c_{\max}$ such that for all $(\beta, \pi) \in B_{\epsilon}^0$,
        $
        0 < c_{\min} \leq -\mathbb{E}_{\phi}[\partial_{\pi^2} \ell_{it}(\beta, \pi)] \leq c_{\max}
        $
        a.s. uniformly over $i, t, N, T$.
        \item (Prior)
        The possibly data-dependent prior density $p(\beta,\phi)$ is locally positive on $B_\epsilon^0$ and $\ln p(\beta,\phi)$ is twice continuously differentiable on $B_\epsilon^0$. Uniformly over $B_\epsilon^0$ and under a common norm, the score and Hessian contributions of $\ln p(\beta,\phi)$ are of no larger stochastic order than the corresponding derivatives of the integrated log-likelihood or the joint log-likelihood.
    \end{enumerate}
\end{assumption}
\begin{assumption}[Partial Effects]\label{assumption:ape}
    Let $\nu > 0$, $\epsilon > 0$ and $B_{\epsilon}^0$ all be as in Assumption \ref{assumption:main}.
    \begin{enumerate}[label=(\roman*)]
        \item (Sampling) for all $N$, $T$, $\{X_{it},\alpha_i,\gamma_t\}_{NT}$ is identically distributed across $i$ and/or stationary across $t$.
        
        \item (Model) for all $i$, $t$, $N$, $T$, the partial effects depend on $\alpha_i$ and $\gamma_t$ through $\alpha_i+\gamma_t$:\\ 
        $\Delta(X_{it},\beta,\alpha_i,\gamma_t)=\Delta_{it}(\beta,\alpha_i+\gamma_t)$.\\
        The realizations of the partial effects are denoted by $\Delta_{it}:=\Delta_{it}(\beta_0,\alpha_{i0}+\gamma_{t0})$.
        
        \item (Smoothness and moments) The function $(\beta,\pi) \mapsto \Delta_{it}(\beta,\pi)$  is four times continuously differentiable over $B_{\epsilon}^0$ a.s. The partial derivatives of $\Delta_{it}(\beta,\pi)$ with respect to the elements of $(\beta,\pi)$ up to fourth order are bounded in absolute value uniformly over $(\beta,\pi)\in B_{\epsilon}^0$ by a function $M(Z_{it} ) > 0$ a.s., and $\max_{i,t}\mathbb{E}[M(Z_{it})^{8+\nu}]$ is a.s. uniformly bounded over $N$, $T$.
        
        \item (Non-degeneracy and moments) $0 < \min_{i,t}[\mathbb{E}(\Delta_{it}^2)-\mathbb{E}(\Delta_{it})^2] \leq \max_{i,t}[\mathbb{E}(\Delta_{it}^2)-\mathbb{E}(\Delta_{it})^2] <\infty $, uniformly over $N$, $T$.
    \end{enumerate}
\end{assumption} 

\begin{remark}[Assumptions \ref{assumption:main} and \ref{assumption:ape}]
    \quad
    \begin{enumerate}[label=(\roman*)]
        \item Classical Laplace expansions for panels often ask for stronger smoothness. For instance, \citet{pakel2019bias} impose higher–order conditions in one–way nonlinear panels with state-dependent observations. Our target-centered full-exponential Laplace approximation stays within the conditions up to fourth order. No extra differentiability is needed. This matters for discrete-choice and count models, where stronger smoothness conditions are often hard to justify.
        \item Assumption~\ref{assumption:main}(vi) is an auxiliary regularity condition needed only for constructing the bias-reducing priors. The more explicit derivative bounds are stated in Assumption~\ref{assumption:regularity}(ix) and verified in the supplemental material.
    \end{enumerate}
\end{remark}

\section{Characterization of the Robust Priors}\label{sec:robust_prior}
This section develops a  target-centered Laplace approximation, uses it to construct bias-reducing priors for generic nonlinear panels, and derives the large-sample distributions for parameters and APEs. Section \ref{sec:model_specific} then specializes these priors to widely used models and provides practical guidance.
\subsection{Target-Centered Full-Exponential Laplace Approximations}
We analyze the integrated likelihood and develop a new Laplace approximation that is pivotal for constructing bias-reducing priors for both the model parameters and the APE. 
Our expansion is centered at the \emph{target quantities}, i.e., the true parameter values, rather than at sample modes (MLE). This choice significantly simplifies the derivations and provides a clean asymptotic framework. In contrast to classical Laplace methods \citep[e.g.,][]{tierney1986accurate,tierney1989fully} and to \cite{arellano2009robust} and \citet{pakel2019bias} in the panel data context, which expand around the MLE, we expand around the truth. No smoothness beyond the regularity in \cite{fernandez2016individual} is required.
The next lemma links the integrated and target likelihoods:

\begin{lemma}\label{lemma:LaplaceLikelihood}
    Suppose Assumption \ref{assumption:main} holds. Then, as $N,T\rightarrow{\infty}$, for any $\beta$ satisfying $\| \beta-\beta_0\|=\mathcal{O}_P\big((NT)^{-1/4}\big)$,
    \begin{eqnarray*}
        \mathcal{L}^{I}(\beta)  -
        \mathcal{ L}\left( \beta,\overline{\phi}(\beta) \right) 
        &=& C^{st} + \tfrac{1}{\sqrt{NT}}  \left[ \ln p\left(\beta, \overline{\phi}(\beta)\right) - \mathcal{D}_{\mathcal{L}^I}\left( \beta,\overline{\phi}(\beta) \right) \right] +o_P(1),
    \end{eqnarray*}
    where $C^{st}$ is a constant term\footnote{$C^{st}$ or $c^{st}$ is an $\mathcal{O}(1)$ constant whose value may change across contexts.},  and 
    \begin{align*}
        \scalebox{0.92}{
            $\mathcal{D}_{\mathcal{L}^I}\left( \beta,\overline{\phi}(\beta) \right) 
            = \frac{1}{2}\ln{\det\left( \overline{\mathcal{H}}\left( \beta,\overline{\phi}(\beta) \right) \right)} 
            - \frac{1}{2}\mathrm{tr}\Big( {\overline{\mathcal{H}}}^{- 1}\left( \beta,\overline{\phi}(\beta) \right)\sqrt{NT}\mathcal{S}\left( \beta,\overline{\phi}(\beta) \right)\left\lbrack \mathcal{S}\left( \beta,\overline{\phi}(\beta) \right) \right\rbrack^{'} \Big).$				} 
    \end{align*}
\end{lemma}

The term $\mathcal{{D}}_{\mathcal{L}^I}\left( \beta,\overline{\phi}(\beta) \right) $coincides with the leading bias in $\mathcal{L}^I(\beta)$ and 
is $ \mathcal{O}_P(\sqrt{NT})$. Hence, if
$\frac{1}{\sqrt{NT}}  \left[ \ln p\left(\beta, \overline{\phi}(\beta)\right) - \mathcal{D}_{\mathcal{L}^I}\left( \beta,\overline{\phi}(\beta) \right) \right]  = o_P(1)$, then the first-order bias in the integrated log-likelihood vanishes and the maximum integrated likelihood estimator $\widehat{\beta}^{I}$ is free of first-order bias.
In the two-way case, direct evaluation of $\mathcal{L}^{I}(\beta)$ becomes  intractable, so we work instead with the posterior means $(\widehat{\beta}^E,\widehat{\phi}^E)$. 
Building on Lemma \ref{lemma:LaplaceLikelihood}, we next develop expansions for these posterior means \emph{centered at the truth} $(\beta_0,\phi_0)$.

Validating this expansion in the two-way setting is challenging. \cite{shun_mccullagh_1995} demonstrate that standard Laplace approximations can fail unless the parameter dimension grows slower than the cube root of the sample size. This condition is violated in our context.\footnote{
    Under Assumption \ref{assumption:main}(i), the parameter dimension  grows at rate $\mathcal{O}((NT)^{1/2})$. This exceeds the rate $o((NT)^{1/3})$ required by \cite{shun_mccullagh_1995} for the convergence of remainder terms.
}
To overcome this, we propose a \textit{cumulant-based expansion} that exploits the extreme sparsity of higher-order derivatives in two-way additive models. This sparsity restricts tensor contractions and keeps remainder terms under control.
We obtain the following asymptotic representations for $(\widehat{\beta}^E,\widehat{\phi}^E)$:

\begin{theorem}\label{thm:asy_expansions}
    Suppose Assumption \ref{assumption:main} holds. 
    Define $\mathcal{B}^{\beta} = \frac{\partial_{\beta} \ln p - \partial_{\beta}  \mathcal{D}_{\mathcal{L}^I} }{\sqrt{NT}}$, and $\mathcal{B}^{\phi} = \frac{\partial_{\phi} \ln p - \partial_{\phi}  \mathcal{D}_{\mathcal{L}^I} }{\sqrt{NT}}$.
    Then, as $N,T\rightarrow{\infty}$,
    \begin{enumerate}[label=(\roman*)]
        \item 
        $\widehat{\beta}^E - \beta_0 = 
        \frac{1}{\sqrt{NT}} \overline{W}^{-1} U^{(0)} 
        + \frac{1}{\sqrt{NT}} \overline{W}^{-1}  
        \left( \mathcal{B}^{\beta}   + [\partial_{\beta\phi'}\overline{\mathcal{L}}]\overline{\mathcal{H}}^{-1} \mathcal{B}^{\phi}\right)
        + R_\beta .
        $
        \item 	$
        \widehat{\phi}^E - \phi_0 = 
        \overline{\mathcal{H}}^{-1} \left(
        \mathcal{S}  + \frac{1}{\sqrt{NT}}[\partial_{\phi\beta'}\overline{\mathcal{L}}]\overline{W}^{-1} U^{(0)}
        \right)
        +
        \overline{\mathcal{H}}^{-1}  \Big(
        \mathcal{B}^{\phi}
        + \frac{1}{\sqrt{NT}}[\partial_{\phi\beta'}\overline{\mathcal{L}}]\overline{W}^{-1}  \big( \mathcal{B}^{\beta} $  \\ 
        $ \mathrm{\qquad\qquad}  + [\partial_{\beta\phi'}\overline{\mathcal{L}}]\overline{\mathcal{H}}^{-1} \mathcal{B}^{\phi}\big)
        \Big)
        + R_\phi.
        $
    \end{enumerate}
    where $ U^{(0)} = \partial_{\beta}{\mathcal{L}}  + [\partial_{\beta\phi'}\overline{\mathcal{L}}]\overline{\mathcal{H}}^{-1}   {\mathcal{S}} $, and 
    $\overline{W} = -\frac{1}{\sqrt{NT}}\left(\partial_{\beta\beta'}\overline{\mathcal{L}}  + [\partial_{\beta\phi'}\overline{\mathcal{L}}]\overline{\mathcal{H}}^{-1}  [\partial_{\beta'\phi}\overline{\mathcal{L}}]    \right)$, and remainder terms satisfy 
    $\lVert R_\beta \rVert =o_P((NT)^{-1/2})$ and
    $\lVert R_\phi \rVert = o_P((NT)^{-1/4})$. 
\end{theorem}
\begin{remark}
    In Theorem \ref{thm:asy_expansions}, $U^{(0)}$ and $\mathcal{S}$ have mean zero and satisfy central limit theorems. 
    The expansion in (i) shows that the first-order bias in $\widehat\beta^{E}$ is driven by $\mathcal B^{\beta}$ and, through the cross-derivative channel, by $\mathcal B^{\phi}$. 
    In (ii), $\widehat\phi^{E}-\phi_{0}$ is dominated by the centered term $\overline{\mathcal H}^{-1}\mathcal S$, so $\widehat\phi^{E}$ is first-order centered. 
    The $\mathcal B^{\beta}$ and $\mathcal B^{\phi}$ terms thus capture the second-order bias in $\widehat\phi^{E}$.
    Under our regularity conditions, $\Vert \overline{\mathcal{H}}^{-1} \Vert=\mathcal{O}_P(1)$, $\Vert\overline{W}^{-1}\Vert=\mathcal{O}_P(1)$, and $\Vert \partial_{\beta\phi'}\overline{\mathcal{L}} \Vert=\mathcal{O}_P((NT)^{1/4})$. 
    With this structure in place, it suffices to choose a prior whose gradient matches:
    $\Vert \partial_{\beta} \ln p - \partial_{\beta}  \mathcal{D}_{\mathcal{L}^I}  \Vert =  o_P\!\big((NT)^{1/2}\big)$
    and
    $\Vert \partial_{\phi} \ln p - \partial_{\phi}  \mathcal{D}_{\mathcal{L}^I}  \Vert_{\infty} = o_P(1)$.\footnote{
        Controlling the $\ell_{\infty}$-norm by $o_P(1)$ implies 
        $\Vert \partial_{\phi} \ln p - \partial_{\phi}  \mathcal{D}_{\mathcal{L}^I}  \Vert_{2} = o_P((NT)^{1/4})$:
        for a vector in $\mathbb R^{N+T}$, $\|\cdot\|_2\le \sqrt{N+T}\,\|\cdot\|_{\infty}=o_P((NT)^{1/4})$ because $N,T\to\infty$ with $N/T\to c^2$ (Assumption \ref{assumption:main}(i)).
    }
\end{remark}

Within the integrated-likelihood framework, both the posterior mean APE $\widehat{\Delta}^E$ and the posterior plug-in APE $\Delta(\widehat{\beta}^E,\widehat{\phi}^E)$ admit comparable asymptotic representations around the true APE
$\Delta := \Delta(\beta_0,\phi_0)$:

\begin{theorem}\label{thm:APE_expansions}
    Suppose Assumptions \ref{assumption:main} and \ref{assumption:ape} hold. Define $  \overline{\mathcal{B}}^{\Delta}_\infty= \lim_{N,T\rightarrow\infty}  \overline{\mathcal{B}}^{\Delta}$ with  $  \overline{\mathcal{B}}^{\Delta} =- \tfrac{1}{2}\mathrm{tr}\Big( (\partial_{\phi\phi'}\overline{\Delta}) \mathrm{diag}\left( \sum_t\mathbb{E}_{\phi}(\partial_{\pi^2}\ell_{it}), \sum_i\mathbb{E}_{\phi}(\partial_{\pi^2}\ell_{it})    \right)^{-1}  \Big) .$\footnote{The operator $\mathrm{diag}(\cdot)$   maps a vector to a diagonal matrix (or picks diagonal blocks, as appropriate).} 
    Then, as $N,T\rightarrow{\infty}$,
    \begin{enumerate}[label=(\roman*)]
        \item $\widehat{\Delta}^E - \Delta = 
        (\partial_{\beta'}\overline{\Delta})     (\widehat{\beta}^E - \beta_0)
        +(\partial_{\phi'}\overline{\Delta}) (\widehat{\phi}^E - \phi_0)   
        + 2 \overline{\mathcal{B}}^{\Delta}_{\infty} + o_P( (NT)^{-1/2} )$.
        \item $\Delta( \widehat{\beta}^E, \widehat{\phi}^E  ) - \Delta = 
        (\partial_{\beta'}\overline{\Delta})    (\widehat{\beta}^E - \beta_0)
        +(\partial_{\phi'}\overline{\Delta}) (\widehat{\phi}^E - \phi_0) 
        + \overline{\mathcal{B}}^{\Delta}_{\infty} + o_P( (NT)^{-1/2} )$.
    \end{enumerate}
\end{theorem}
\begin{remark}
    Theorem \ref{thm:APE_expansions} decomposes the first-order bias of APE estimators into three components: (i) the bias in $\widehat{\beta}^E$, (ii) the  bias in $\widehat{\phi}^E$, and (iii) the asymptotic variance of estimating $\phi(\beta_0)$.\footnote{We write $\overline{\mathcal{B}}^{\Delta}$ for the variance-driven component; it enters with coefficient $2$ for the posterior mean APE and with coefficient $1$ for the plug-in APE.} With a bias-reducing prior, (i)–(ii) become negligible while (iii) remains, an additional correction based on $\overline{\mathcal{B}}^{\Delta}_{\infty}$  is therefore required.
\end{remark}

\subsection{Robust Priors for Generic Models}\label{sec:robust_generic}
We turn to a general formulation of the bias-reducing prior. Our aim is straightforward: choose a prior $p(\beta,\phi)$ so that the posterior means $(\widehat{\beta}^E,\widehat{\phi}^E)$ have no first-order bias. 
Theorem \ref{thm:asy_expansions} shows that this bias is driven by the gap between the gradient of the log–prior and the gradient of the leading bias term in the integrated log–likelihood. 

To make this decomposition explicit, let $\Upsilon_{it}^{\pi}\in\mathbb{R}$ and $\Upsilon_{it}^{\beta}\in\mathbb{R}^{\dim\beta}$ denote the observation level contributions to the leading bias of gradients in the fixed-effects and common parameter directions, respectively.
In the proof of Theorem \ref{corollary:generic_prior} we show that 
\begin{enumerate}[label=(\roman*)]
	\item $\partial_{\alpha_i}  \mathcal{D}_{\mathcal{L}^I}        = \sum_{t=1}^{T} \Upsilon_{it}^{\pi} +o_P\left(1\right)$, and 
	          $\partial_{\gamma_t}  \mathcal{D}_{\mathcal{L}^I}    = \sum_{i=1}^{N} \Upsilon_{it}^{\pi}   +o_P\left(1\right)$,
	\item $\partial_{\beta}  \mathcal{D}_{\mathcal{L}^I}   = \sum_{i=1}^N \sum_{t=1}^T \Upsilon_{it}^{\beta}  +o_P((NT)^{1/2})$,
\end{enumerate}
Here,
\begin{equation}\label{eqn:Tree}
	\begin{aligned}
		\Upsilon_{it}^{\pi}
		&=
		\tfrac{\mathbb{E}_{\phi}(\partial_{\pi^{3}}\ell_{it}) +\mathbb{E}_{\phi}(\partial_{\pi^{2}}\ell_{it}\partial_{\pi}\ell_{it})}
		{\sum_{\tau = 1}^{T}{\mathbb{E}_{\phi}( \partial_{\pi^{2}}\ell_{i\tau} )}}
		+
		\tfrac{\mathbb{E}_{\phi}(\partial_{\pi^{3}}\ell_{it}) +\mathbb{E}_{\phi}(\partial_{\pi^{2}}\ell_{it}\partial_{\pi}\ell_{it})}
		{\sum_{j = 1}^{N}{\mathbb{E}_{\phi}( \partial_{\pi^{2}}\ell_{jt} )}}
		+
		\tfrac{\sum_{\tau=t+1}^{T}\mathbb{E}_{\phi}(\partial_{\pi^{2}}\ell_{i\tau}\partial_{\pi}\ell_{it})}
		{\sum_{\tau = 1}^{T}{\mathbb{E}_{\phi}( \partial_{\pi^{2}}\ell_{i\tau})}}. \\
		\Upsilon_{it}^{\beta}
		&=
		\tfrac{\mathbb{E}_{\phi}(\partial_{\beta\pi^{2}}\ell_{it}) +\mathbb{E}_{\phi}(\partial_{\beta\pi}\ell_{it}\partial_{\pi}\ell_{it})}
		{\sum_{\tau = 1}^{T}{\mathbb{E}_{\phi}( \partial_{\pi^{2}}\ell_{i\tau} )}}
		+
		\tfrac{\mathbb{E}_{\phi}(\partial_{\beta\pi^{2}}\ell_{it}) +\mathbb{E}_{\phi}(\partial_{\beta\pi}\ell_{it}\partial_{\pi}\ell_{it})}
		{\sum_{j = 1}^{N}{\mathbb{E}_{\phi}( \partial_{\pi^{2}}\ell_{jt} )}}
		+
		\tfrac{\sum_{\tau=t+1}^{T}\mathbb{E}_{\phi}(\partial_{\beta\pi}\ell_{i\tau}\partial_{\pi}\ell_{it})}
		{\sum_{\tau = 1}^{T}{\mathbb{E}_{\phi}( \partial_{\pi^{2}}\ell_{i\tau})}}. \\
	\end{aligned}
\end{equation}
To discuss the last term in \eqref{eqn:Tree} jointly, let $\theta\in\{\beta,\pi\}$, where $\partial_{\theta\pi}\ell_{i\tau}$ denotes $\partial_{\beta\pi}\ell_{i\tau}$ when $\theta=\beta$ and $\partial_{\pi^2}\ell_{i\tau}$ when $\theta=\pi$. 
Then the term $\sum_{\tau=t+1}^T\mathbb E_\phi(\partial_{\theta\pi}\ell_{i\tau}\,\partial_\pi\ell_{it})$ captures \textit{cross-time} dependencies. 
It is a spectral expectation that is generally nonzero in dynamic or predetermined settings,
and  as discussed in Section~\ref{sec:model_specific}, it largely determines the fixed-$T$ performance.
Following the standard spectral estimation, we replace
$\sum_{\tau=t+1}^T\mathbb E_\phi(\partial_{\theta\pi}\ell_{i\tau}\,\partial_\pi\ell_{it})$ by
$\sum_{\tau=1}^{\min\{L,T-t\}}T/(T-\tau) \mathbb E_\phi(\partial_{\theta\pi}\ell_{i,t+\tau}\,\partial_\pi\ell_{it})$
for a fixed integer $L$, 
and apply the degrees-of-freedom adjustment $T/(T-\tau)$ as in \cite{hahn2011bias}.\footnote{
	Here $L$ is the trimming parameter for estimating spectral expectations and we assume $L\rightarrow \infty$ and $L/T\to 0$ asymptotically.
    }

Representation \eqref{eqn:Tree} characterizes an \emph{equivalence class} of bias-reducing priors, that is, any log-prior whose gradient matches $\Upsilon$ at the below stated rates will eliminate the leading bias in \textit{both} $\widehat{\beta}^E$ and $\widehat{\phi}^E$:
\begin{equation} \label{eqn:differentialsystem}
    \begin{aligned}
        &\left\| \partial_\beta \ln p - \textstyle\sum_{i=1}^N\sum_{t=1}^T  \Upsilon^\beta_{it}  \right\| 
        = o_P( (NT)^{1/2} ), \\
        &\left| \partial_{\alpha_i} \ln p - \textstyle\sum_{t=1}^T \Upsilon^\pi_{it} \right| = o_P(1), 
        \qquad
        \left| \partial_{\gamma_t} \ln p - \textstyle\sum_{i=1}^N \Upsilon^\pi_{it} \right|  = o_P(1).
    \end{aligned}
\end{equation}
Among the priors that satisfy \eqref{eqn:differentialsystem}, we report two convenient closed-form choices:
\begin{theorem}[Bias-reducing priors for generic models]\label{corollary:generic_prior}
    Suppose Assumption \ref{assumption:main} holds. Let $L<\infty$ be fixed. Then the following priors are bias-reducing for $\widehat{\beta}^E$ and $\widehat{\phi}^E$:
    \begin{enumerate}[label=(\roman*)]
        \item (Prior GE-1: generic nonlinear panels) \\
        $\ln p(\beta,\phi) = 
        \frac{1}{2} \sum_{i=1}^{N} \left[ \ln\left(  \sum_{t = 1}^{T}{\mathbb{E}_{\phi}\left[ -\partial_{\pi^{2}}\ell_{it}(\beta,\pi) \right]}    \right)  - \frac{  \sum_{t = 1}^{T}{\mathbb{E}_{\phi}[\partial_{\pi}\ell_{it}(\beta,\pi)^2]}      }{ \sum_{t = 1}^{T}{\mathbb{E}_{\phi}\left[ - \partial_{\pi^{2}}\ell_{it}(\beta,\pi) \right]}} \right]  \\
        \mathrm{\qquad\qquad\quad}  + \frac{1}{2}\sum_{t=1}^{T} \left[ \ln\left(  \sum_{i = 1}^{N}{\mathbb{E}_{\phi}\left[ -\partial_{\pi^{2}}\ell_{it}(\beta,\pi) \right]}   \right) -   \frac{  \sum_{i = 1}^{N}{\mathbb{E}_{\phi}[\partial_{\pi}\ell_{it}(\beta,\pi)^2]}      }{ \sum_{i = 1}^{N}{\mathbb{E}_{\phi}\left[ -\partial_{\pi^{2}}\ell_{it}(\beta,\pi) \right]}} \right]  \\
        \mathrm{\qquad\qquad\quad}  -  \sum_{i=1}^{N} \left[    \frac{  \sum_{t = 1}^{T}{\sum_{\tau=1}^{\min\{L,T-t\}}T/(T-\tau)\mathbb{E}_{\phi}[\partial_{\pi}\ell_{i,t+\tau}(\beta,\pi)\partial_{\pi}\ell_{it}(\beta,\pi)]}      }{ \sum_{t = 1}^{T}{\mathbb{E}_{\phi}\left[ -\partial_{\pi^{2}}\ell_{it}(\beta,\pi) \right]}}  \right] .
        $
        \item (Prior GE-2: generic nonlinear panels) \\
        $\ln p(\beta,\phi) = 
        \sum_{i=1}^{N} \left[ \ln\left(  \sum_{t = 1}^{T}{\mathbb{E}_{\phi}\left[ -\partial_{\pi^{2}}\ell_{it}(\beta,\pi) \right]}    \right) \right]
        + \sum_{t=1}^{T} \left[ \ln\left(  \sum_{i = 1}^{N}{\mathbb{E}_{\phi}\left[- \partial_{\pi^{2}}\ell_{it}(\beta,\pi) \right]}   \right) \right]  \\
        \mathrm{\qquad\qquad\quad}  - \frac{1}{2}\sum_{t=1}^{T} \left[ \ln\left(  \sum_{i = 1}^{N}{\mathbb{E}_{\phi}\left[ (\partial_{\pi}\ell_{it}(\beta,\pi))^2 \right]}   \right) \right]  \\
        \mathrm{\qquad\qquad\quad}  -\frac{1}{2} \sum_{i=1}^{N} \Big[ \ln\Big(  
        \sum_{t= 1}^{T} \mathbb{E}_{\phi}\left[ (\partial_{\pi}\ell_{it}(\beta,\pi))^2 \right] \\
         \mathrm{\qquad\qquad\quad}  + 2\sum_{t= 1}^{T} \sum_{\tau=1}^{\min\{L,T-t\}}T/(T-\tau)\mathbb{E}_{\phi}[\partial_{\pi}\ell_{i,t+\tau}(\beta,\pi)\partial_{\pi}\ell_{it}(\beta,\pi)]    \Big)  \Big] .
        $
    \end{enumerate} 
\end{theorem}
\begin{remark}
    Bias-reducing priors are not unique. Any alternative log-prior $\ln \tilde p(\beta,\phi)$ that differs by $\ln \tilde p=\ln p+R_{NT}$ with 
    $\| \partial_\beta R_{NT}\| =o_P((NT)^{1/2})$ and
    $|\partial_{\alpha_i} R_{NT}|=o_P(1)$,  $|\partial_{\gamma_t} R_{NT}|=o_P(1)$
    in a neighborhood of the truth, also satisfies \eqref{eqn:differentialsystem} and is therefore equally bias-reducing.
    In particular, adding a constant to $\ln p$ is immaterial.
    Moreover, taking weighted averages of $\ln p$ preserves the condition.
    For instance, for any $\omega\in[0,1]$,
    $\ln p_{\omega}=\omega\ln p_{\mathrm{GE\text{-}1}}+(1-\omega)\ln p_{\mathrm{GE\text{-}2}}$  is also bias-reducing.
    We thus view GE-1 and GE-2 as convenient closed-form choices rather than unique solutions, and other primitives (including numerical solutions) are possible.
\end{remark}

Implementation requires evaluating the expectations $\mathbb{E}_{\phi}(\cdot)$, which are normally \textit{data-dependent}. Following \cite{arellano2016likelihood}, the expectations can be computed either in observed form (sample analog) or in closed form from the model. 
We adopt a \textit{mixed strategy}: use observed forms by default (consistent for the expectations), and substitute analytical expectations \textit{whenever available} to improve small-sample performance. The large-$N,T$ validity of the correction is unaffected, and our Monte Carlo evidence confirms that the two approaches coincide asymptotically.

In Section  \ref{sec:model_specific} we show that, for specific models, certain components of $\Upsilon$ (or of their expectations) vanish, which delivers closed-form simplifications of the prior. To our knowledge, these are the first model-specific bias-reducing priors.

\subsection{Asymptotic Distribution and Inference}
The next  results summarize our asymptotics under bias-reducing priors.
The first result describes the limiting distribution of the posterior means, while the second treats the two APE estimators and isolates a residual variance-type bias term:

\begin{corollary}\label{corollary:asy_parameter}
    Suppose Assumption \ref{assumption:main} holds and $ p(\beta,\phi)$ is a bias-reducing prior. Then,
    \begin{enumerate}[label=(\roman*)]
        \item $(NT)^{1/2}(\widehat{\beta}^E - \beta_0 )  \xrightarrow{d}\mathcal{N}(0, \overline{W}^{-1}_{\infty}) $.
        \item The sub-vector of $\widehat{\phi}^E$ of the fixed length $L\geq1$ satisfies
        
        $(NT)^{1/4}(\widehat{\phi}^E - \phi_0 )_{1:L}  \xrightarrow{d}\mathcal{N}\Big(0, (\overline{\mathcal{H}}_{\infty}^{-1})_{1:L,1:L}\Big)    $.
    \end{enumerate}
    where $\overline{W}_{\infty}$ and $\overline{\mathcal{H}}_{\infty}$ are the limits of $\overline{W}$ and $\overline{\mathcal{H}}$ (when they exist) as $N,T\to\infty$.
\end{corollary}
\begin{remark}
    Corollary \ref{corollary:asy_parameter} establishes the asymptotic normality of $\widehat{\beta}^E$ under a bias-reducing prior, thereby supporting standard likelihood-based inference. A Wald test is immediately available using the posterior mean and covariance from the MCMC output, which yield valid frequentist inference via Corollary \ref{corollary:asy_parameter}. A Score (LM) test is also feasible, as the gradient of integrated likelihood equals the posterior expectation of joint score, which can be estimated via the MCMC run with $\beta$ fixed under the null.
    These implications are in the same spirit as \citet{leng2025debiased}, who show that debiasing the joint likelihood restores standard large-sample likelihood-based inference.
\end{remark}

Before stating the asymptotic result for APE estimators, it is useful to emphasize that our APE estimators center at the population target $\mathbb E(\Delta)$ rather than at the sample average $\Delta$. 
As a consequence, the convergence rate need not be the parametric $\sqrt{NT}$ rate, and instead depends on the sampling properties of $(X_{it},\alpha_i,\gamma_t)$, as shown by \cite{fernandez2016individual}.
\begin{corollary}\label{corollary:asy_APE}
    Suppose Assumptions \ref{assumption:main} and \ref{assumption:ape} hold and $ p(\beta,\phi)$ is a bias-reducing prior.  
    Consider any sampling scheme under which the APE convergence rate $r_{NT}$ and asymptotic variance $\overline{V}_{\infty}^{\delta}$ are characterized as in Theorem 4.2 of \cite{fernandez2016individual}. Then,
    \begin{enumerate}[label=(\roman*)]
        \item $r_{NT}\left(\widehat{\Delta}^E - \mathbb{E}(\Delta)     - 2\overline{\mathcal{B}}^{\Delta}_{\infty}     \right)
        \xrightarrow{d}\mathcal{N}(0, \overline{V}_{\infty}^{\delta}) $.
        \item $r_{NT}\left(\Delta( \widehat{\beta}^E, \widehat{\phi}^E  )  - \mathbb{E}(\Delta)    -\overline{\mathcal{B}}^{\Delta}_{\infty}        \right)
        \xrightarrow{d}\mathcal{N}(0, \overline{V}_{\infty}^{\delta}) $.
    \end{enumerate}
    where $\overline{\mathcal{B}}^{\Delta}_{\infty}$ is defined in Theorem \ref{thm:APE_expansions}.
\end{corollary} 
\begin{remark}
	Corollary \ref{corollary:asy_APE}  shows how the posterior mean APE  $\widehat{\Delta}^{E}$ and posterior plug-in APE  $\Delta(\widehat{\beta}^{E},\widehat{\phi}^{E})$  inherit asymptotic normality once the convergence rate $r_{NT}$ is available under a given sampling scheme of $(X_{it},\alpha_i,\gamma_t)$. 
    Our analysis therefore provides the identification of the remaining variance-type bias term. Although this underlying bias component is the same, it enters the two APE estimators differently, i.e., $2\overline{\mathcal{B}}_{\infty}^{\Delta}$ for $\widehat{\Delta}^{E}$, while $\overline{\mathcal{B}}_{\infty}^{\Delta}$ for $\Delta(\widehat{\beta}^{E},\widehat{\phi}^{E})$.
\end{remark}

To our knowledge, there exists no tractable prior that simultaneously corrects both parameter and APE biases. Therefore, we address the remaining bias term directly. We construct the bias-corrected APE estimators by estimating $\mathcal{B}^{\Delta}$ at the posterior means and  performing the following subtraction:
$$
\widetilde{{\Delta}}^E  = \widehat{\Delta}^E -2{\mathcal{\widehat{B}}}^{\Delta}( \widehat{\beta}^E, \widehat{\phi}^E ),
\qquad
\widetilde{\Delta}( \widehat{\beta}^E, \widehat{\phi}^E )  = \Delta( \widehat{\beta}^E, \widehat{\phi}^E )    -{\mathcal{\widehat{B}}}^{\Delta}( \widehat{\beta}^E, \widehat{\phi}^E).
$$

For \textit{inference} on model parameters and APEs, one may either (i) compute analytical asymptotic variances (i.e., $\overline{W}^{-1}_{\infty}$, $\overline{\mathcal{H}}^{-1}_{\infty}$  and $\overline{V}_{\infty}^{\delta}$), or (ii) use posterior draws. The latter is justified by the Bernstein–von Mises theorem: although our priors are improper, Bernstein–von Mises applies under local positivity and continuity of the prior near the true parameter \cite[Ch. 10]{van2000asymptotic}, and implies asymptotic equivalence. Our assumption ensures these conditions, so both routes lead to the same first-order inference.

\subsection{Bias-reducing Penalties for Joint Likelihood}\label{sec:subsecPenaltyforJoint}
Together with \citet{jochmans2019likelihood}, the connection between integrated- and joint-likelihood corrections is direct. The penalty term is given by \textit{Prior GE} minus $\frac{1}{2}\ln{\det\left( \overline{\mathcal{H}}\left( \beta,\phi \right) \right)}$.\footnote{
    To see the link, the expansion of the concentrated log-likelihood around the target mode $\overline{\phi}(\beta)$ gives the leading bias
    $-\tfrac{1}{2}\,\mathrm{tr}\!\left\{\overline{\mathcal H}^{-1}(\beta,\overline{\phi}(\beta))\,\sqrt{NT}\,\mathcal S(\beta,\overline{\phi}(\beta))\,\mathcal S(\beta,\overline{\phi}(\beta))^{\prime}\right\}.$
    When moving to the integrated log-likelihood, the standard Laplace expansion about the sample mode $\widehat{\phi}(\beta)$ (cf. \citealp{tierney1989fully}) contributes an additional term $\frac{1}{2}\ln{\det\left(\cdot \right)}$. Hence the joint–likelihood penalty is obtained by subtracting this term: 
    $\mathrm{Penalty GE}=\mathrm{Prior GE}-\frac{1}{2}\ln{\det\left(\cdot\right)}$.
} Specifically, we obtain  detailed penalty formulas for generic nonlinear panels:	\footnote{\emph{Proof sketch.} Start from the leading bias of the concentrated log-likelihood in \citet{jochmans2019likelihood}, and 
    apply the same decomposition steps as in the proof of Theorem \ref{corollary:generic_prior} to collapse this matrix expression into scalar form, which yields \textit{Penalty GE}s in closed form. A step-by-step derivation note is available upon request.}

\begin{enumerate}[label=(\roman*)]
    \item (\textit{Penalty GE-1: generic models}) \\
    $\ln p^J(\beta,\phi) = 
    \frac{1}{2} \sum_{i=1}^{N} \left[   - \frac{  \sum_{t = 1}^{T}{\mathbb{E}_{\phi}[\partial_{\pi}\ell_{it}(\beta,\pi)^2]}      }{ \sum_{t = 1}^{T}{\mathbb{E}_{\phi}\left[ - \partial_{\pi^{2}}\ell_{it}(\beta,\pi) \right]}} \right] 
    + \frac{1}{2}\sum_{t=1}^{T} \left[  -   \frac{  \sum_{i = 1}^{N}{\mathbb{E}_{\phi}[\partial_{\pi}\ell_{it}(\beta,\pi)^2]}      }{ \sum_{i = 1}^{N}{\mathbb{E}_{\phi}\left[ -\partial_{\pi^{2}}\ell_{it}(\beta,\pi) \right]}} \right]  \\
        \mathrm{\qquad\qquad\quad}  -  \sum_{i=1}^{N} \left[    \frac{  \sum_{t = 1}^{T}{\sum_{\tau=1}^{\min\{L,T-t\}}T/(T-\tau)\mathbb{E}_{\phi}[\partial_{\pi}\ell_{i,t+\tau}(\beta,\pi)\partial_{\pi}\ell_{it}(\beta,\pi)]}      }{ \sum_{t = 1}^{T}{\mathbb{E}_{\phi}\left[ -\partial_{\pi^{2}}\ell_{it}(\beta,\pi) \right]}}  \right] .
    $
    \item (\textit{Penalty GE-2: generic models}) \\
    $\ln p^J(\beta,\phi) = 
    \frac{1}{2}\sum_{i=1}^{N} \left[ \ln\left(  \sum_{t = 1}^{T}{\mathbb{E}_{\phi}\left[ -\partial_{\pi^{2}}\ell_{it}(\beta,\pi) \right]}    \right) \right]\\
    \mathrm{\qquad\qquad\quad} + \frac{1}{2}\sum_{t=1}^{T} \left[ \ln\left(  \sum_{i = 1}^{N}{\mathbb{E}_{\phi}\left[- \partial_{\pi^{2}}\ell_{it}(\beta,\pi) \right]}   \right) \right]  \\
    \mathrm{\qquad\qquad\quad}  -\frac{1}{2} \sum_{i=1}^{N} \left[ \ln\left(  \sum_{t= 1}^{T} \sum_{\tau = 1}^{T}{\mathbb{E}_{\phi}\left[ \partial_{\pi}\ell_{i\tau}(\beta,\pi) \partial_{\pi}\ell_{it}(\beta,\pi) \right]}    \right)  \right]  \\
    \mathrm{\qquad\qquad\quad}  - \frac{1}{2}\sum_{t=1}^{T} \left[ \ln\left(  \sum_{i = 1}^{N}{\mathbb{E}_{\phi}\left[ \partial_{\pi}\ell_{it}(\beta,\pi) \right]^2}   \right) \right]  .
    $
\end{enumerate} 
	
\citet{jochmans2019likelihood} illustrate an MCMC scheme that targets the joint likelihood combined with \textit{Penalty GE}, and report the corresponding posterior mean of $\beta$. 
Our theory distinguishes the roles of \textit{Penalty GE} and \textit{Prior GE}. \textit{Penalty GE} is designed to remove the leading bias of the maximized penalized joint- likelihood estimator. 
For MCMC-based posterior means, in contrast, Corollary~\ref{corollary:asy_parameter} shows that bias reduction is achieved when the likelihood is paired with \textit{Prior GE}.

Given $(\widehat{\beta},\widehat{\phi})$ from the penalized joint likelihood, the APE is computed in plug-in form and can be bias-corrected in the same way:
$
\widetilde{\Delta}( \widehat{\beta}, \widehat{\phi} )  = \Delta( \widehat{\beta}, \widehat{\phi} )    -{\mathcal{\widehat{B}}}^{\Delta}( \widehat{\beta}, \widehat{\phi} ).
$

\section{Model-Specific Bias-reducing Priors}\label{sec:model_specific}
We begin with two workhorse specifications in applied microeconometrics: binary logit and  Poisson counts. Both models are \emph{one-parameter exponential-family}  models.\footnote{
    Examples of \emph{one-parameter} exponential family models include binary logit, Poisson, and normal regression with known variance, among others. Each observation is indexed by a single natural parameter and under canonical links, the variance is pinned down by the mean.}
Precisely because of this structure, the generic priors simplify to \emph{non–data-dependent} formulas in the settings with strict exogenous covariates. 
We then carefully cover logit and Poisson models with predetermined covariates, and other widely used specifications such as dynamic AR($m$) panels, multinomial logit, and probit, allowing for either strict exogeneity or predetermined covariates.

Our aims are twofold. First, where closed-form simplifications exist, we provide explicit priors that improve small sample behavior and are easy to compute. Second, where such simplifications are unavailable, we indicate when the generic priors from Section~\ref{sec:robust_prior} are the appropriate default and when score- or estimator-based corrections may be preferable.

\subsection{One-Parameter Exponential Family Panels}	 
We describe one-parameter exponential family panel models with additive fixed effects, following \cite{mccullagh2019generalized}. 
Let the conditional density of $Y_{it}$ belong to the exponential family with natural index $X_{it}^{\prime}\beta+\pi$, and write the log-likelihood in its canonical form
\begin{equation}\label{eqn:exp_family_likelihood}
    \ell_{it}(\beta,\pi)
    = Y_{it}(X_{it}^{\prime}\beta+\pi)-\mathcal{C}_{1}(X_{it}^{\prime}\beta+\pi)+\mathcal{C}_{2}(Y_{it}), 
\end{equation}for known functions $\mathcal{C}_1$ and $\mathcal{C}_2$. 
Then, 
$\partial_{\pi}\ell_{it} = Y_{it} - \partial_{\pi}\left[\mathcal{C}_1(X_{it}'\beta_0 + \pi_{it,0})\right]$, 
$\partial_{\pi^2}\ell_{it} =  - \partial_{\pi^2}\left[\mathcal{C}_1(X_{it}'\beta_0 + \pi_{it,0})\right]$,
$\partial_{\pi^3}\ell_{it} =  - \partial_{\pi^3}\left[\mathcal{C}_1(X_{it}'\beta_0 + \pi_{it,0})\right]$,
and the mixed derivatives satisfy the canonical identities
$\partial_{\beta\pi}\ell_{it} = -\partial_{\pi^2}\ell_{it} X_{it}$, and
$\partial_{\beta\pi^2}\ell_{it} = -\partial_{\pi^3}\ell_{it} X_{it}$. 
Hence $\partial_{\pi^{2}}\ell_{it}$ and its higher $\pi$-mixed derivatives do not depend on the data $Y_{it}$.

\subsubsection{Exponential family: strict exogeneity.}	 
When $X_{it}$ is strictly exogenous, the expressions for $\Upsilon^{\theta}_{it}$ in \eqref{eqn:Tree} collapse to
$
\Upsilon_{it}^{\theta}
=	 \frac{\partial_{\theta\pi^{2}}\ell_{it} }{\sum_{\tau = 1}^{T}{\partial_{\pi^{2}}\ell_{i\tau} }}
+
\frac{\partial_{\theta\pi^{2}}\ell_{it} }{\sum_{j = 1}^{N}{\partial_{\pi^{2}}\ell_{jt} }} ,
$ for $\theta\in\{\beta, \pi\}$. 
Because $\partial_{\pi^{2}}\ell_{it}$ and its higher derivatives are $Y$-free, observed and expected forms coincide. Integrating $\Upsilon^{\theta}$ yields a non–data-dependent prior:

(\textit{Prior SE:  exponential family with strictly exogenous regressors}) 
\begin{align*}\label{eq:prior:SE}
    \scalebox{1}{
        $\ln p(\beta,\phi) = 
        \sum_{i=1}^{N} \left[ \ln\left(  \sum_{t = 1}^{T}{\left[-\partial_{\pi^{2}}\ell_{it}(\beta,\pi)\right] }    \right) \right]
        + \sum_{t=1}^{T} \left[ \ln\left(  \sum_{i = 1}^{N}{ \left[- \partial_{\pi^{2}}\ell_{it}(\beta,\pi)\right]}   \right) \right] $. }
\end{align*}

This prior exactly cancels the leading bias in the integrated log-likelihood for  one-parameter exponential families with strict exogenous covariates. While the generic priors remain valid, \textit{Prior SE} is preferable in finite samples because it avoids estimated expectations.

\subsubsection{Exponential family: predetermined regressors.}\label{subsection:dynamicexp}
Using exponential-family properties, we have that  $\mathbb{E}_{\phi}(\partial_{\theta\pi}\ell_{it}\,\partial_{\pi}\ell_{it})=0$. With predetermined variable in $X_{it}$, the cross-time term in \eqref{eqn:Tree} survives: 
$$
\Upsilon_{it}^{\theta} = 
  \tfrac{\mathbb{E}_{\phi}(\partial_{\theta\pi^{2}}\ell_{it})}{\sum_{\tau = 1}^{T}{ \mathbb{E}_{\phi}( \partial_{\pi^{2}}\ell_{i\tau} ) }}
+\tfrac{\mathbb{E}_{\phi}(\partial_{\theta\pi^{2}}\ell_{it})}{\sum_{j = 1}^{N}{\mathbb{E}_{\phi}( \partial_{\pi^{2}}\ell_{jt} )}}
+ \tfrac{\sum_{\tau=1}^{\min\{L,T-t\}}T/(T-\tau) \mathbb E_\phi(\partial_{\theta\pi}\ell_{i,t+\tau}\,\partial_\pi\ell_{it})}
  {\sum_{\tau = 1}^{T}{ \mathbb{E}_{\phi}(\partial_{\pi^{2}}\ell_{i\tau} )}},
$$
for $\theta\in\{\beta, \pi\}$. 
An exact closed-form primitive for the differential system \eqref{eqn:differentialsystem} is unavailable. Nevertheless, the following simplified prior is bias-reducing:

(\textit{Prior PE:  exponential family models with predetermined regressors})
\begin{align*}
    & \scalebox{1}{
        $\ln p(\beta,\phi) =  
        \sum_{i=1}^{N} \left[ \ln\left(  \sum_{t = 1}^{T}{\mathbb{E}_{\phi}\left[ -\partial_{\pi^{2}}\ell_{it}(\beta,\pi) \right]}    \right) \right]
        + \sum_{t=1}^{T} \left[ \ln\left(  \sum_{i = 1}^{N}{\mathbb{E}_{\phi}\left[- \partial_{\pi^{2}}\ell_{it}(\beta,\pi) \right]}   \right) \right]    $ } \\
    & \scalebox{1}{
        $\mathrm{\qquad\qquad\quad}  -  \sum_{i=1}^{N} \left[    \frac{  \sum_{t = 1}^{T}{\sum_{\tau=1}^{\min\{L,T-t\}}T/(T-\tau)\mathbb{E}_{\phi}[\partial_{\pi}\ell_{i,t+\tau}(\beta,\pi)\partial_{\pi}\ell_{it}(\beta,\pi)]}      }{ \sum_{t = 1}^{T}{\mathbb{E}_{\phi}\left[ -\partial_{\pi^{2}}\ell_{it}(\beta,\pi) \right]}}  \right] $. } 
\end{align*}
	
Relative to the generic priors, \textit{Prior PE} improves small-$N$ performance and and its finite-sample accuracy does not depend on the individual dimension. 
Its remaining finite-sample error is driven by the cross-time dependence captured by the trimmed lead--lag term.
To see this, consider $\partial_{\alpha_i}\ln p^{PE}$. A direct calculation gives the remaining mismatch between $\partial_{\alpha_i}\ln p^{PE}$ and the target $\textstyle\sum_{t=1}^T  \Upsilon^\pi_{it} $:
\begin{eqnarray*}
			\partial_{\alpha_i}   \ln p^{PE} -  \textstyle\sum_{t=1}^T  \Upsilon^\pi_{it}   & = &
	\tfrac{\sum_{t=1}^T{ \mathbb{E}_{\phi} (\partial_{\beta\pi^{2}}\ell_{it}) } \sum_{t = 1}^{T}{\sum_{\tau=1}^{\min\{L,T-t\}}\frac{T}{T-\tau}\mathbb{E}_{\phi}(\partial_{\pi}\ell_{i,t+\tau}\partial_{\beta\pi}\ell_{it})}}{  \big(\sum_{t=1}^T\mathbb E_\phi({ \partial_{\pi^{2}}\ell_{it} })\big)^2  } \\
	&& 
	+ \tfrac{\sum_{t = 1}^{T}{\sum_{\tau=1}^{\min\{L,T-t\}}\frac{T}{T-\tau}\mathbb{E}_{\phi}(\partial_{\pi}\ell_{i,t+\tau}\partial_{\pi}\ell_{it})}}{\sum_{t = 1}^{T}{ \mathbb E_\phi(\partial_{\pi^{2}}\ell_{it} )}}.
\end{eqnarray*}
Both remainder terms involve the trimmed lead--lag moments, i.e., $\mathbb{E}_{\phi}(\partial_{\pi}\ell_{i,t+\tau}\partial_{\beta\pi}\ell_{it})$ and $\mathbb{E}_{\phi}(\partial_{\pi}\ell_{i,t+\tau}\partial_{\pi}\ell_{it})$.
By the martingale difference property, these expectations are zero at the truth.
In practice, however, they must be estimated by plug-in sample analogs.
As $T$ grows, these moments are consistently estimated, and the mismatch vanishes.
When $T$ is small, they can be noisy to estimate, so a finite-$T$ discrepancy may remain even under \textit{Prior PE}.

When $T$ is small and one seeks additional accuracy, a quick \textit{add-on correction} is available. 
This  alternative  is to estimate $(\widehat{\beta}^E,\widehat{\phi}^E)$ under \textit{Prior SE} and adjust for the cross-time component
$$\mathcal{B}_i^{(\theta,1)} = \tfrac{\textstyle\sum_{t = 1}^{T}{\sum_{\tau=1}^{\min\{L,T-t\}}\frac{T}{T-\tau}\mathbb{E}_{\phi}(\partial_{\beta\pi}\ell_{i,t+\tau}\partial_{\pi}\ell_{it})}}{\sum_{t = 1}^{T}{ \mathbb{E}_{\phi}(\partial_{\pi^{2}}\ell_{it}) }}  $$
to obtain
$$
\widetilde{\beta}^E = \widehat{\beta}^E - 
\tfrac{1}{\sqrt{NT}} \widehat{{W}}^{-1}
\left( \textstyle\sum_{i=1}^{N}\widehat{\mathcal{B}}_i^{(\beta,1)}  + [\partial_{\beta\phi'}\widehat{{\mathcal{L}}}]\widehat{{\mathcal{H}}}^{-1} \widehat{\mathcal{B}}^{(\pi,1)}    \right),
\quad
\widetilde{\phi}^E = \widehat{\phi}^E(\widetilde{\beta}^E) - 
\widehat{{\mathcal{H}}}^{-1}\widehat{\mathcal{B}}^{(\pi,1)},
$$
where 
${\mathcal{B}}^{(\pi,1)}=({\mathcal{B}}_1^{(\pi,1)},\ldots,{\mathcal{B}}_N^{(\pi,1)},0_T')'$,
and the notation $\widehat{f}=f(\widehat{\beta}^E,\widehat{\phi}^E)$ denotes the plug-in fixed effects estimator.

\subsubsection{Poisson panels: strict exogeneity.}	 
After the general analysis of exponential families,  we single out the Poisson model as a worked example because we obtain an additional result. 
In this case, $\mathcal{C}_1\left(X_{it}'\beta+\pi\right) = \exp\left(X_{it}'\beta+\pi\right)$. With $\omega_{it}=\exp\left(X_{it}'\beta_0 + \pi_{it,0}\right)  $, we have
$\partial_{\pi}\ell_{it}=Y_{it}-\omega_{it}$,
$\partial_{\pi^2}\ell_{it}=\partial_{\pi^3}\ell_{it}=-\omega_{it}$,
$\partial_{\beta}\ell_{it}=\partial_{\pi}\ell_{it}X_{it}$,
and $\partial_{\beta\pi}\ell_{it}=\partial_{\beta\pi^2}\ell_{it}=\partial_{\pi^2}\ell_{it}X_{it}$.
When all covariates are strictly exogenous, substituting into Theorem \ref{thm:asy_expansions},  the bias of $\widehat{\beta}^E$ depends on
$
\mathcal{B}^{\beta}   + [\partial_{\beta\phi'}\overline{\mathcal{L}}]\overline{\mathcal{H}}^{-1} \mathcal{B}^{\phi} 
=  \frac{\partial_{\beta}\ln p}{\sqrt{NT}}  
- \frac{1}{N} \sum_{i} 	 \frac{\sum_{t}\omega_{it}X_{it} (\partial_{\alpha_i} \ln p) }{\sum_{t}{ \omega_{it}  }} 
-
\frac{1}{T} \sum_{t} \frac{\sum_{i}\omega_{it}X_{it} (\partial_{\gamma_t} \ln p) }{\sum_{i}{ \omega_{it} }}   + o_P(1);
$
and leading bias of $\widehat{\phi}^E$ depends on
$
\mathcal{B}^{\phi_k}  =  \tfrac{\partial_{\phi_k} \ln p - 2 }{\sqrt{NT}} + o_P((NT)^{-1/2}),
$ for $k=1,\ldots,\dim\phi$.

If the goal is \textit{only} to debias $\widehat{\beta}^E$, a flat prior already removes the leading bias. This finding echoes  \cite{fernandez2016individual}, who show that in the  Poisson model with strict exogenous $X_{it}$, the fixed-effects estimator of $\beta$ is asymptotically unbiased.
However, our derivations indicates that a flat prior does not remove the second-order bias in $\widehat{\phi}^E$, and as a consequence, complicates the APE debiasing (Corollary \ref{corollary:asy_APE}). In contrast, \textit{Prior SE} eliminates the leading bias in both $\widehat{\beta}^E$ and $\widehat{\phi}^E$ and, for the APE correction, leaves only the variance-driven term in Corollary \ref{corollary:asy_APE} to subtract.

\subsection{Dynamic AR($m$)  Panels}\label{sec:arp}
Dynamic linear panels provide a transparent benchmark for incidental-parameter bias.
The classic ``Nickell bias'' arises under fixed-$T$ asymptotics \citep{nickell1981bias}, while explicit bias expansions and higher-order corrections are available under large-$N$, large-$T$ asymptotics \citep{kiviet1995bias, hahnkuersteiner2002,hahn2006reducing}.
Related likelihood-based adjustments for autoregressive panels with fixed effects are studied in \citet{dhaenejochmans2016} and \citet{alvarez2022robust}.
This subsection shows that, under our bias-reducing prior framework, a simple closed-form prior is available for Gaussian AR($m$) panels even in the presence of \emph{two-way} additive effects.

We allow unrestricted time effects $\{\gamma_t\}$ to absorb aggregate shocks.
Such time effects are generally incompatible with strict stationarity of the \emph{level} process $\{Y_{it}\}$ (unless $\gamma_t$ is constant and thus absorbed by normalization).
Accordingly, our derivation does not impose stationarity on $\{Y_{it}\}$.
Instead, we impose stability of the AR polynomial on the innovation-driven component obtained after removing $\alpha_i+\gamma_t$, which ensures an MA($\infty$) representation and delivers the moment calculations below.\footnote{Throughout, the stability restriction is imposed on $\widetilde Y_{it}=Y_{it}-\alpha_i-\gamma_t$.
    In estimation, $\gamma_t$ is still included (or profiled out) to absorb common shocks, and it will not enter the leading bias terms for common parameters.}

Consider
$
Y_{it} = \mu_{1,0} Y_{i,t-1}+\cdots+\mu_{m,0} Y_{i,t-m} + \alpha_{i,0}  + \gamma_{t,0} + \varepsilon_{it},
$
with
$
(\varepsilon_{i1},\ldots,\varepsilon_{iT}) \sim\mathcal{N}(0,\sigma_{0}^{2}I_{T})
$
conditional on the true fixed effects and initial conditions.
Let $X_{it}=(Y_{i,t-1},\ldots,Y_{i,t-m})^{\prime}$ and $\beta=(\mu^{\prime},\sigma)^{\prime}$, $\mu=(\mu_1,\ldots,\mu_m)'$.
The per-observation log-likelihood is given by
$
\ell_{it}(\beta,\pi) = -\frac{1}{2}\ln(2\pi)-\frac{1}{2}\ln(\sigma^2)-\frac{1}{2}\frac{(Y_{it}-X_{it}'\mu-\pi)^2}{\sigma^2}.
$
Then
$\partial_{\pi}\ell_{it}      =  \varepsilon_{it}\sigma^{-2}$,
$\partial_{\pi^2}\ell_{it}  =  -\sigma^{-2}$,
$\partial_{\mu\pi}\ell_{it}      =  \partial_{\pi^2}\ell_{it}X_{it}$,
$\partial_{\sigma\pi}\ell_{it}  = -2\sigma^{-1}\partial_{\pi}\ell_{it}$,
$\partial_{\sigma\pi^2}\ell_{it}  = -2\sigma^{-1}\partial_{\pi^2}\ell_{it}$, and
$\partial_{\pi^3}\ell_{it} =\partial_{\mu\pi^2}\ell_{it}  = 0$. 

Assume the AR polynomial $\Phi(z)=1-\sum_{k=1}^{m}\mu_{k}z^{k}$ is stable.
We now work with the innovation driven component $\widetilde Y_{it}=Y_{it}-\alpha_i-\gamma_t$.
Let $\widetilde Y_{it}=\sum_{h\ge 0}\psi_h \varepsilon_{i,t-h}$ be the MA($\infty$) representation with $\psi_0=1$ and
$\psi_h=\sum_{l=1}^{\min(m,h)}\mu_l\,\psi_{h-l}$ for $h\ge 1$ \citep{hamilton2020time}.
The required cross-moments involve the innovation-driven component only. The fixed effects contribute a deterministic term that is orthogonal to the innovations and therefore drops out from the leading bias expressions.

Standard moments imply, for $k=1,\ldots,m$,
$
\sum_{t,\tau=1}^T \mathbb E(\widetilde Y_{i,\tau-k}\varepsilon_{it})
= \sigma_\varepsilon^2\sum_{h=0}^{T-k-1}(T-k-h)\psi_h,
$
$
\sum_{t,\tau=1}^T \mathbb E(\varepsilon_{i\tau}\varepsilon_{it}) = T\,\sigma_\varepsilon^2,
$
$
\sum_{t<\tau}^T \mathbb E(\varepsilon_{i\tau}\varepsilon_{jt})=0 (\forall i\neq j).
$
Substituting into \eqref{eqn:Tree} gives
$\Upsilon_{it}^{\pi}=0$ and
$\Upsilon_{it}^{\sigma}=0$, while for $k=1,\ldots,m$, we obtain
$
\sum_i\sum_t\Upsilon_{it}^{\mu_k}
=
\sum_{h=0}^{T-k-1}(T-k-h)\psi_h .
$
Hence, for this Gaussian AR($m$) specification, the two-way additive effects do not enter the leading bias terms that determine the bias-reducing prior.
Equivalently, the same primitive applies whether one includes only individual effects or both individual and time effects, because the relevant moments depend on $\mu$ only through the innovation-driven recursion. 

(\textit{Prior AR($m$):  dynamic  AR($m$) panels })
\begin{align*}
    \scalebox{1}{
        $\ln p(\beta,\phi) = \ln p(\mu) 
        = \sum_{t=1}^{T-1}\left( \frac{T-t}{t}	 \sum_{k=1}^m r_k^{t}(\mu) \right),  $ } 
\end{align*}
where $r_{1},\ldots,r_{m}$ are the roots of the AR polynomial $\Phi(z)=1-\sum_{k=1}^{m}\mu_{k}z^{k}$.
\footnote{\textit{Proof sketch.} Since $\sum_{t=1}^T\frac{z^t}{t}\sum_{k=1}^{m}r_k^t = - \sum_{k=1}^{m}\ln(1-r_k z)= - \ln \Phi(z)$, the log prior can be written as $\ln p = -[z^0]\{A(z)z^{-m}\ln \Phi(z)\}$, with $A(z)= \sum_{t=0}^{T-1} (T-t) z^t$ and $[z^0]\{\cdot\}$ denoting the constant term in the Laurent expansion (\cite{ahlfors1979complex} Ch. 4)  around $z=0$. For $k=1,\ldots,m$,
    using 
    $\frac{\partial}{\partial\mu_k}\ln\Phi(z)=z^k/\Phi(z)$
    and
    $1/\Phi(z)=\sum_{t\ge0}^{\infty}\psi_h z^t$ gives that
    $\partial_{\mu_k}\ln p 
    = [z^0]\left\{A(z) \frac{\partial}{\partial\mu_k}[z^{-k}\ln\Phi(z)]\right\}
    = [z^0]\{A(z)/\Phi(z)\}
    =\sum_{t=0}^{T-1}(T-t)\psi_{t-k}
    = \sum_{h=0}^{T-k-1}(T-k-h)\psi_h
    =\sum_i\Upsilon_{\alpha,i}^{\mu_k}+\sum_t\Upsilon_{\gamma,t}^{\mu_k}$,
    which is exactly the desired expression. 
}

Although \textit{Prior AR($m$)} is expressed in terms of the characteristic roots, computing those roots via convolution or fast Fourier transformation is unnecessary. The power sums $\sum_{k=1}^{m} r_k^{t}$ can be generated recursively computed using Newton’s identities \citep[e.g.,][Ch.~1]{knuth1997art}, which is numerically stable for moderate $m$. This recursion yields a simple and fast implementation for  AR models. 
For AR(2) panel model, two roots can be obtained explicitly from the quadratic formula, yielding a closed-form representation.  For AR(1), the single root is simply $r_1=\mu_1$, so we have:

(\textit{Prior AR(1):  dynamic  AR(1) panel models})
\begin{align*}
    \scalebox{1}{
        $\ln p(\beta,\phi) = \ln p(\mu) 
        = \sum_{t=1}^{T-1}\left( \frac{T-t}{t} \mu_1^t \right). $ } 
\end{align*} 
Note that Prior AR(1) coincides with \cite{lancaster2002orthogonal,arellano2009robust} in the one-way case. 
The invariance to two-way additive effects is specific to this Gaussian linear AR($m$) benchmark. In nonlinear panels, unrestricted time effects can alter the leading bias terms and may require a modified asymptotic expansion \citep[cf.][]{fernandez2016individual}.

Given the prior, a convenient point estimator is the penalized fixed-effects Gaussian MLE,
$
(\hat\mu,\hat\sigma,\hat\alpha,\hat\gamma)
\in
\arg\max_{\mu,\sigma,\alpha,\gamma}
\left\{
\sum_{i=1}^N\sum_{t=m+1}^T \ell_{it}(\beta,\alpha_i+\gamma_t)
+\ln p(\mu)
\right\}.
$
Profiling out $(\alpha,\gamma)$ yields the usual two-way within transformation. 
Let $\ddot Y_{it}$ and $\ddot X_{it}$ denote the double-demeaned outcome and regressors over $i$ and $t$.
Then the bias corrected $\hat\mu$ solves the penalized normal equations:
$$
\big(\textstyle\sum_{i,t}\ddot X_{it}\ddot X_{it}'\big)\hat\mu
=
\sum_{i,t}\ddot X_{it}\ddot Y_{it}
+\hat\sigma^{2}\,\partial_{\mu}\ln p(\hat\mu),
$$
with $\hat\sigma^{2}$ given by the profiled Gaussian variance (e.g., $\hat\sigma^{2}=(NT)^{-1}\sum_{i,t}\hat\varepsilon_{it}^{2}$).
Moreover, the gradient admits the closed form $\partial_{\mu_k}\ln p(\mu)=\sum_{h=0}^{T-k-1}(T-k-h)\psi_h$, where ${\psi_h}$ are computed by the usual AR recursion.

\subsection{Multinomial Logit Panels}	 
Multinomial logit is a standard specification for polytomous panel outcomes. 
Building on the one-way fixed effects literature \citep{chamberlain1980analysis,wooldridge2010econometric}, we consider  two-way fixed effects. 
Our approach extends seamlessly.
We assume that $Y_{it}\in\{1,\ldots,J\}$ is a polytomous categorical variable, with category $1$ taken as the baseline. 
For each non-baseline category $k=2,\ldots,J$, let
$
\pi_{it,k}=\alpha_{i,k}+\gamma_{t,k},
$
where $\alpha_{i,k}$ and $\gamma_{t,k}$ are the individual and time effects for category $k$, relative to the baseline. 
Stack these objects as
$\beta=(\beta_2',\ldots,\beta_J')',$
$\alpha_i=(\alpha_{i,2},\ldots,\alpha_{i,J})',$
$\gamma_t=(\gamma_{t,2},\ldots,\gamma_{t,J})',$
and
$
\pi_{it}=(\pi_{it,2},\ldots,\pi_{it,J})'=\alpha_i+\gamma_t.
$
Accordingly, the fixed effect parameter is now vector valued. The choice probabilities are given by
$$
\Pr\left(  Y_{it} =  k  | X_{it},\pi_{it},\beta    \right) = 
\begin{cases}
    \dfrac{1}{1 + \sum_{j=2}^J \exp\!\left( X_{it}'\beta_j + \pi_{it,j}  \right)}, & \text{if } k=1, \\[2ex]
    \dfrac{\exp\!\left( X_{it}'\beta_k +\pi_{it,k}\right)}{1 + \sum_{j=2}^J \exp\!\left( X_{it}'\beta_j + \pi_{it,j} \right)}, & \text{if }  k=2,\ldots,J,
\end{cases}
$$
and per-observation log-likelihood is
$\ell_{it}(\beta,\pi) = \sum_{k=1}^{J} 1\left\{Y_{it}=k\right\}\ln\Pr\left(Y_{it}=k|X_{it},\pi_{it},\beta\right)$. 

Here $\beta=(\beta_2',\ldots,\beta_J')'$ and $\pi_{it}=(\pi_{it,2},\ldots,\pi_{it,J})'$ are stacked vectors. The key difference is that $\pi_{it}$ is now a $(J-1)$-dimensional vector to account for each categorical choice. Consequently, all scalar second derivative objects are replaced by their matrix analogues. In particular, $\partial_{\pi^{2}}\ell_{it}$ becomes the Hessian block $\partial_{\pi\pi'}\ell_{it}$ and cross-time products use $\partial_{\pi}\ell_{i\tau}\,\partial_{\pi'}\ell_{it}$. 
With this notation upgrade, the \textit{Prior SE/PE} constructions carry over:

(\textit{Prior SML: multinomial logit model with strict exogenous regressors})
\begin{align*}
    \scalebox{0.95}{
        $\ln p(\beta,\phi) = 
        \sum_{i=1}^{N} \left[ \ln\det\left(  \sum_{t = 1}^{T}{\left[-\partial_{\pi\pi'}\ell_{it}(\beta,\pi)\right] }    \right) \right]
        + \sum_{t=1}^{T} \left[ \ln\det\left(  \sum_{i = 1}^{N}{ \left[- \partial_{\pi\pi'}\ell_{it}(\beta,\pi)\right]}   \right) \right] . $ } 
\end{align*}

(\textit{Prior PML: multinomial logit model with predetermined regressors}) 
\begin{align*}
    & \scalebox{0.95}{
        $\ln p(\beta,\phi) =  
        \sum_{i=1}^{N} \left[ \ln\det\left(  \sum_{t = 1}^{T}{\mathbb{E}_{\phi}\left[ -\partial_{\pi\pi'}\ell_{it}(\beta,\pi) \right]}    \right) \right]
        + \sum_{t=1}^{T} \left[ \ln\det\left(  \sum_{i = 1}^{N}{ \mathbb{E}_{\phi}\left[- \partial_{\pi\pi'}\ell_{it}(\beta,\pi) \right]}   \right) \right]$ } \\
    & \scalebox{0.95}{
        $\mathrm{\qquad\qquad\quad}  -  \sum_{i=1}^{N} \mathrm{tr}\bigg[    
        \left(\sum_{t = 1}^{T}{\mathbb{E}_{\phi}\left[-\partial_{\pi\pi'}\ell_{it}(\beta,\pi)\right] }\right)^{-1} 
         \sum_{t = 1}^{T}{\sum_{\tau=1}^{\min\{L,T-t\}}\frac{T}{T-\tau}\mathbb{E}_{\phi}\left[\partial_{\pi}\ell_{i,t+\tau}(\beta,\pi)\partial_{\pi'}\ell_{it}(\beta,\pi)\right]}        \bigg].$ } 
\end{align*}

\subsection{Probit Panels}	 
Let $\Phi$ denote the standard normal cdf, $\Phi_{it}=\Phi(X_{it}^{\prime}\beta+\pi)$, and 
$H_{it}=\partial\Phi_{it}/(\Phi_{it}-\Phi_{it}^2)$. The log-likelihood is $\ell_{it}(\beta,\pi)=Y_{it}\ln\Phi_{it}+(1-Y_{it})\ln\{1-\Phi_{it}\},$
so that
$\partial_{\pi}\ell_{it} = H_{it} (Y_{it}-\Phi_{it} )$, 
$\partial_{\pi^2}\ell_{it} =   - H_{it}  \partial\Phi_{it} + \frac{\partial H_{it}}{H_{it}} \partial_{\pi}\ell_{it}  $,
$\partial_{\pi^3}\ell_{it} =   - H_{it}  \partial^2\Phi_{it} - 2\partial H_{it} \partial\Phi_{it} + \partial^2H_{it} (Y_{it}-\Phi_{it} )                 $, 
$\partial_{\beta\pi}\ell_{it} = - \partial_{\pi^2}\ell_{it} X_{it}$, and 
$\partial_{\beta\pi^2}\ell_{it} = - \partial_{\pi^3}\ell_{it} X_{it}$.
Also, $\mathbb{E}_{\phi}(\partial_{\pi^2}\ell_{it}\,\partial_{\pi}\ell_{it})=\mathbb{E}_{\phi}(\partial H_{it} \partial \Phi_{it})$.
Substituting these  into \eqref{eqn:Tree} gives
$$
\Upsilon_{it}^{\pi}
= 
\tfrac{\mathbb{E}_{\phi}  (H_{it}\partial^2\Phi_{it}   +\partial H_{it} \partial \Phi_{it} )   }{\sum_{\tau}{ \mathbb{E}_{\phi} (H_{i\tau}\partial\Phi_{i\tau}) }} +
\tfrac{  \mathbb{E}_{\phi}  (H_{it}\partial^2\Phi_{it}   +\partial H_{it} \partial \Phi_{it} )      }{\sum_{j}{ \mathbb{E}_{\phi} (H_{jt}\partial\Phi_{jt} )}}   
+
\tfrac{{\sum_{\tau>t} \mathbb{E}_{\phi} (H_{i\tau}\partial\Phi_{i\tau}\partial_{\pi}\ell_{it}})}{\sum_{\tau}{\mathbb{E}_{\phi}( H_{i\tau}\partial\Phi_{i\tau} )}}  .
$$

Carrying through the algebra shows that an explicit non–data-dependent prior that precisely cancels the bias is not available, even in the strict exogenous case. The generic priors (\textit{GE-1}/\textit{GE-2}) are therefore the default for probit panels. With small $N$ or small $T$, score or estimator-based corrections may yield tighter finite-sample performance.

\subsection{Guidance for Practitioners: Choosing a Prior and a Correction}\label{sec:subsecGuidance}
As a default, we recommend the robust priors  (GE-1/GE-2) given in Section~\ref{sec:robust_prior}. They remove the leading incidental parameter bias for a wide range of nonlinear panels. The only data dependence arises through terms of the form $\mathbb{E}_{\phi}(\cdot)$ and primarily reflects the cross-time dependence inherent in panel data context. 
In practice, whenever analytical expectations are available, they typically deliver better small-sample stability. When no closed form is available, the observed (sample-analog) form provides a convenient alternative, though it may be less stable in small samples.

With strictly exogenous regressors, one-parameter exponential family panels (logit, Poisson, and multinomial logit) work best with the \textit{non-data-dependent} priors \textit{SE}/\textit{SML}. These priors match the relevant bias exactly and usually deliver tighter finite-sample performance.
With predetermined regressors, the dedicated priors \textit{PE}/\textit{PML} are recommended. The remaining time-driven component vanishes as $T$ grows.
When $T$ is small, a simple improvement is to estimate under \textit{SE}/\textit{SML} and then apply the add-on correction $\mathcal{B}^{(\theta,1)}$ to account for the cross-time component.

For dynamic AR($m$) panels, the explicit \textit{Prior AR($m$)}  is best: it is non–data-dependent, fixed-$T$ consistent, and easy to compute via power-sum recursions. 
\textit{Prior AR($m$)} itself can be used as the joint-likelihood penalty as well. 
In probit models, no closed-form model-specific prior is available. 
The safe choice for probit panel models is GE-1/GE-2. With small $N$ or $T$, one may consider score or estimator corrections.

\begin{table}[h!]
    \caption{Recommended model-specific priors and practical guidance}\label{Table:allprior}
    \centering
    {\fontsize{8.0pt}{12pt}\selectfont
        \begin{tabular}{l c c c l}
            \toprule
            \multirow{2}{*}{Model} 
            & \multicolumn{1}{c}{Regressor} 
            & \multicolumn{1}{c}{Recommended} 
            & \multicolumn{1}{c}{Data-dependent} 
            & \multicolumn{1}{c}{\multirow{2}{*}{Practical note}} \\
            & \multicolumn{1}{c}{exogeneity} 
            & \multicolumn{1}{c}{prior} 
            & \multicolumn{1}{c}{prior} 
            & \\
            \midrule
            Generic models                & Either   & GE-1/GE-2  & \checkmark & Large-$N,T$ correction; broad applicability \\
            Logit                         & Strict   & SE         & $\times$   & Best finite-sample behavior \\
            Poisson                       & Strict   & SE/Flat prior & $\times$   & Use SE when debiasing $\phi$ and APEs \\
            Logit/Poisson                 & Predet   & PE         & \checkmark & Small-$N$ gains; SE plus add-on correction if $T$ small \\
            Multinomial logit             & Strict   & SML        & $\times$   & Matrix analogue of SE \\
            Multinomial logit             & Predet   & PML        & \checkmark & Matrix analogue of PE \\
            Probit                        & Either   & GE-1/GE-2  & \checkmark & Use score/estimator corrections for small $N,T$ \\
            Dynamic AR($m$)               & Predet   & AR($m$)    & $\times$   & Fixed-$T$ consistent; $\mathcal{O}(Tm)$ evaluation \\
            \bottomrule
        \end{tabular}
        \legend{ ``Strict'' indicates strictly exogenous regressors. ``Predet'' indicates predetermined regressors (e.g., lagged outcomes). ``Either'' indicates that the recommended prior applies under either ``Strict'' or ``Predet'' assumption.}
    }
\end{table}

Table \ref{Table:allprior} lists bias-reducing priors for the \emph{integrated} likelihood. For \emph{joint} likelihood corrections, as discussed in Section \ref{sec:subsecPenaltyforJoint}, the only additional component when moving from the joint to the integrated likelihood is the $\tfrac{1}{2}\ln\det\!\big(\overline{\mathcal{H}}(\beta,\phi)\big)$ term. The same logic delivers model-specific penalties by inserting the model’s Hessian into this term. The algebra is omitted for brevity.

\section{Monte Carlo Experiments}\label{sec:montecarlo}
This section studies the finite sample performance of our likelihood-based bias corrections. We consider four distinct nonlinear panel models: ordered logit, binary logit, multinomial logit, and probit, considering both static designs with strictly exogenous regressors and dynamic designs with predetermined regressors (including lagged dependent variables). These models span settings where model-specific priors are available (e.g., binary logit) as well as cases where only generic corrections are  feasible. 
To our knowledge, bias correction for ordered and multinomial response panels with two-way fixed effects has not previously been studied.  See \citet{honore2025dynamic} for related work on dynamic ordered logit panels with individual effects.

\subsection{Data-generating processes}\label{subsec:DGP}
We specify four data generating processes (DGPs) that encompass the models of interest. The design largely follows \citet{fernandez2016individual}. In all cases, the unobserved individual effects $\alpha_i$ and time effects $\gamma_t$ are drawn independently from $\mathcal{N}(0, 0.25^2)$. The exogenous covariate $Z_{it}$ evolves according to an AR(1) process:
$
Z_{it} = 0.5 Z_{i,t-1} + \alpha_i + \gamma_t + v_{it}, 
$
where $v_{it} \sim \mathcal{N}(0, 0.5),$ and $Z_{i0} \sim \mathcal{N}(0, 1)$.

For binary and ordered designs, outcomes are generated from a latent representation. 
In the static specification,
\[
Y_{it}^* = Z_{it} \beta_Z + \alpha_i + \gamma_t + \varepsilon_{it},
\]
and we set  $\beta_Z = 1$. In the dynamic case we introduce state dependence via
\[
Y_{it}^* =  Y_{i,t-1} \beta_Y + Z_{it} \beta_Z  + \alpha_i + \gamma_t + \varepsilon_{it},
\]
with $\beta_Z = 1$ and $\beta_Y = 0.5$.\footnote{The multinomial logit design uses an alternative-specific index and is described separately below.} The  disturbance $\varepsilon_{it}$ is standard logistic in logit-family designs and standard normal in probit.
Given $Y_{it}^*$, the observed outcome $Y_{it}$ is constructed as described next:

(1) \textit{Binary response models (logit/probit)}. The outcome is an indicator function of the latent variable:
$
Y_{it} = \mathbf{1}\{Y_{it}^* > 0\}.
$

(2) \textit{Ordered logit models}. The outcome takes values in $\{1, 2, 3, 4\}$. We generate the observed responses by discretizing the latent variable $Y_{it}^*$ according to three true cutoffs $\tau_1=-2.5$, $\tau_2=0.5$, and $\tau_3=2.5$:
$
Y_{it} = \sum_{k=1}^{4} k \cdot \mathbf{1}\left\{ \tau_{k-1} < Y_{it}^* \leq \tau_{k} \right\},
$
with $\tau_{0} = -\infty$ and $\tau_4 = \infty$. 
In estimation, the cutoffs are treated as common parameters. Since the index is identified only up to location, we fix $\tau_1$ and estimate $\beta=(\beta_Z,\beta_Y,\tau_2,\tau_3)'$. Note that the  cutoffs $\tau_2$ and $\tau_3$ are subject to incidental parameter bias, which our method corrects. In our simulations, the choice of cutoffs ensures that the ordered-response likelihood is well behaved under two-way effects, in the sense of the identification condition in Assumption~\ref{assumption:main}.\footnote{
For the ordered logit, Assumption~\ref{assumption:main}(v) rules out asymptotically individuals or time periods with outcomes concentrated in a single category, and implies local identification (a nonsingular Hessian in a neighborhood) with probability approaching one as $N,T\to\infty$.}

(3) \textit{Multinomial logit models}. We consider a setting with $J=3$ alternatives, where $Y_{it} \in \{1, 2, 3\}$. The probability of observing outcome $k$ is given by the standard multinomial logit formulation with base category 1:
\[
\Pr(Y_{it} = k \mid \cdot) = \exp(V_{it,k}) / \big(  1 + \textstyle\sum_{j=2}^{3} \exp(V_{it,j})    \big), \quad \text{for } k \in \{2, 3\},
\]
where the linear index $V_{it,k} = X_{it}'\beta_k + \alpha_{i,k} + \gamma_{t,k}$. Here, $\beta_k$ represents the alternative specific coefficients, and $\alpha_{i,k}, \gamma_{t,k}$ are the alternative specific fixed effects. In the static case, $X_{it} = Z_{it}$; in the dynamic case, $X_{it} = (\mathbf{1}\{Y_{i,t-1}=1\}, Z_{it})'$.

We set $N=45$ and consider $T\in\{15,45\}$,\footnote{Whenever a correction requires estimating cross-time spectral expectations, we set $L=1$ for $T=15$ and $L=2$ for $T=45$, in line with the small-lag choices commonly used in practice.} with 500 replications. Across designs we compare: (i) the uncorrected fixed-effects MLE; (ii) the analytical and jackknife corrections of \citet{fernandez2016individual} (FW16) when available; and (iii)
our bias-reducing estimators, implemented either as posterior mean under the corresponding  prior or as a penalized likelihood estimators. 
We report Monte Carlo bias, standard deviation (SD), the ratio of average standard errors to SD, and empirical 95\% coverage. APEs are computed using the plug-in estimator  described in Section~\ref{sec:modelandestimators}.\footnote{
    All simulation designs and estimators reported in this section can be reproduced using our open-source Python package \texttt{twowaypanel}. 
    In particular, the package implements the four model classes considered here under both static and dynamic specifications, and provides likelihood-based prior/penalty corrections together with inference for parameters and APEs. 
    For binary logit and probit panels, it also includes the analytical bias correction of FW16. 
}

\subsection{Monte Carlo results}\label{subsec:MC_results}
\emph{Ordered logit.}
We begin with ordered logit panels. 
Table \ref{Sim:ologit} reports the results for both static and dynamic specifications. 
The uncorrected estimator shows severe bias not only in the slope coefficients ($\beta$) but also in the cutoff parameters ($\tau$). 
Our likelihood-based corrections (\textit{GE-1}) reduce these biases sharply and restore coverage close to the nominal level, even in the dynamic specification. 
For average partial effects, the static model shows only modest bias to begin with, whereas the dynamic model displays large distortions. In both cases the corrections deliver well calibrated inference. 
A key advantage is that the procedure treats $\beta$ and the ancillary cutoffs $\tau$ in a unified way, avoiding derivation of specific bias formulas for each parameter.

\emph{Binary logit.}
Table \ref{Sim:logit} reports results for the binary logit model, a one-parameter exponential-family specification.
For the static model, we use the model-specific \textit{SE} correction, which is non-data-dependent and cancels the leading bias in the integrated likelihood. Consistent with this property, prior/penalty \textit{ SE} is nearly unbiased already at $T=15$, with no visible variance inflation.
For the dynamic logit, we use the dedicated \textit{PE} correction, which delivers substantial bias reduction with minimal implementation burden. At $T=15$, prior \textit{PE} cuts the large bias in $\beta_{Y_{t-1}}$ from $-64.1\%$ to $-13.0\%$. With short $T$, the remaining bias is largely cross-time. Accordingly, the add-on adjustment (\textit{SE+AC}) further improves performance (to $-5.6\%$ and $-3.7\%$). By $T=45$, residual bias under \textit{PE} is small (about $-4.2\%$), and \textit{SE+AC} remains close to the analytical benchmark. The same ordering holds for APEs. All likelihood-based corrections improve bias and coverage, with \textit{SE+AC} more helpful when $T$ is small.

\emph{Multinomial logit.}
We next consider the multinomial logit model (Table \ref{Sim:mlogit}). This model extends the one-parameter exponential family to a vector-valued setting, increasing the number of nuisance parameters. Despite this added complexity, the matrix versions \textit{Prior SML} (static) and \textit{Prior PML} (dynamic) substantially reduce the bias in the coefficients and yield well-calibrated inference. In the dynamic case with $T=15$, the uncorrected estimator for $\beta_Y$ exhibits a large bias of $-60.6\%$. The Prior PML reduces this to $-8.3\%$, while maintaining a standard error close to the uncorrected estimator. This demonstrates the robustness of our approach in handling vector-valued outcomes and higher-dimensional fixed effects parameters.

\emph{Binary Probit.}
In probit panels (Table~\ref{Sim:probit}), GE-1 reduces bias relative to the uncorrected estimator and improves coverage. While the bias reduction is significant, it is slightly less complete than in the logit cases for small $T$. For example, in the dynamic probit model with $T=15$, the bias in $\beta_{Y_{t-1}}$ is reduced from $-39.8\%$ (Uncorrected) to $-15.3\%$ (Prior GE-1), compared to $-4.6\%$ for FW16. However, as $T$ increases to 45, the performance gap closes, and the Prior GE-1 estimator becomes nearly unbiased. This pattern is consistent with our theoretical predictions of the generic correction for binary probit panel, which relies on large-$N,T$ asymptotics. 

\emph{Summary.} 
Across all four model classes, the proposed likelihood-based corrections provide a simple and reusable route to debiasing two-way nonlinear panels, including ordered and multinomial response models where alternative corrections were previously unavailable. For standard logit models, the model-specific priors deliver excellent finite-sample performance, while for probit models, the generic priors provide a robust improvement that converges rapidly as the time dimension grows.

\section{Empirical Illustration: Female Labor Force Participation }\label{sec:empirical}
To illustrate the utility of our method, we revisit the classic relationship between fertility and female labor force participation. This relationship is central to labor economics but is difficult to identify due to the joint determination of fertility decisions and labor supply (see, e.g., \citealp{angrist1998children}). In contrast to earlier one-way applications \citep[e.g.,][]{fernandez2009fixed}, we incorporate both individual heterogeneity (e.g., preferences for work) and aggregate time shocks (e.g., business cycles) via two-way fixed effects.

The data come from the Panel Study of Income Dynamics (PSID), covering calendar years 1979--1988 (waves 13--22). Our estimation window is 1980--1988 (9 waves) and the 1979 participation outcome is retained only to initialize the lagged dependent variable in the dynamic specification. The sample is restricted to women aged 18--60 in 1985 who remained married to a husband in the labor force throughout the sample period. 
The full sample contains 1,461 women observed for 9 years. Because the logit model with two-way additive fixed effects is not identified from individuals whose participation status never changes over the panel, estimation of the index parameters uses the subsample of 664 ``switchers'' (women who change participation status at least once).\footnote{In switcher subsample, about 16\% are black women; average age in 1985 is 35.6; average schooling is 12 years; the mean participation rate is 57\%. Average numbers of children per woman are 0.28 (ages 0--2), 0.36 (ages 3--5), and 1.11 (ages 6--17).}

We consider static and dynamic two-way fixed-effects logit models. Let $Y_{it}$ denote labor force participation and write $\beta=(\beta_Y,\beta_Z')'$ and
$X_{it}=(Y_{i,t-1},Z_{it}')'$, where $Z_{it}$ includes the numbers of children aged 0--2, 3--5, and 6--17, and log husband's earnings. We estimate
$$
	\Pr(Y_{it}=1\mid X_{it},\alpha_i,\gamma_t)
	=
	F\!\left(X_{it}'\beta+\alpha_i+\gamma_t\right),
	\qquad
	F(v)={e^{v}}/{(1+e^{v})},
$$
using the $664$ switchers ($i=1,\ldots,664$) observed over $t=1,\ldots,9$.

Although index parameters are estimated on switchers, we report population-scaled averages over the full panel of $1461\times 9$ person-years. In our data, 121 women never participate and 676 always participate over 1980--1988. 
For these 797 ``stayers'', their fixed effects index $\alpha_i$ diverges. As a result, the logit derivative $F'(\cdot)$ vanishes and so do the  partial effects.\footnote{
	For a continuous $X_{it,k}$, the per-observation partial effect is
	$
	\Delta_{it,k}(\beta,\phi)=\beta_k\,F'\!\left(X_{it}'\beta+\alpha_i+\gamma_t\right),
	$
	and for a binary regressor $X_{it,k}$, it is defined by
	$
	\Delta_{it,k}(\beta,\phi)
	=
	F\!\left(X_{it,-k}'\beta_{-k}+\beta_k+\alpha_i+\gamma_t\right)
	-
	F\!\left(X_{it,-k}'\beta_{-k}+\alpha_i+\gamma_t\right).
	$ 
	When $|\alpha_i|\to\infty$, $\Delta_{it,k}(\beta,\phi)\to 0$.
}
Accordingly, for APEs it is without loss to average over switchers and scale by the population size
$$
\Delta(\beta,\phi):=\tfrac{1}{1461\times 9}\textstyle\sum_{i=1}^{664}\sum_{t=1}^{9} \Delta(X_{it},\beta,\alpha_i+\gamma_t).
$$

In the dynamic specification, we follow \cite{browning2014dynamic} and define one-step transition probabilities in a two-state first order Markov representation. Specifically, conditional on $(Z_{it},\alpha_i+\gamma_t)$:
\begin{eqnarray*}
	P_{01}(Z_{it},\beta,\alpha_i+\gamma_t)
	&=&\Pr(Y_{it}=1\mid Y_{i,t-1}=0,Z_{it},\alpha_i+\gamma_t)	=F(Z_{it}'\beta_Z+\alpha_i+\gamma_t), \\
	P_{11}(Z_{it},\beta,\alpha_i+\gamma_t)
	&=&\Pr(Y_{it}=1\mid Y_{i,t-1}=1,Z_{it},\alpha_i+\gamma_t)	=F(Z_{it}'\beta_Z+\beta_Y+\alpha_i+\gamma_t).
\end{eqnarray*}

The \emph{marginal dynamic effect} is the difference in transition probabilities,
$$
\Delta_{M}(Z_{it},\beta,\alpha_i+\gamma_t)
:=P_{11}(Z_{it},\beta,\alpha_i+\gamma_t)-P_{01}(Z_{it},\beta,\alpha_i+\gamma_t),
$$
which coincides exactly with the partial effect $\Delta(X_{it},\beta,\alpha_i+\gamma_t)$ for the discrete change in the lagged indicator $Y_{i,t-1}$. Hence the reported APE for $Participation_{t-1}$ in Table~\ref{table:laborStaticResults} can be read as the \emph{average marginal dynamic effect} (AMDE): 
$$
\Delta_M(\beta,\phi)=\tfrac{1}{1461\times 9}\textstyle\sum_{i=1}^{664}\sum_{t=1}^{9} \Delta_{M}(Z_{it},\beta,\alpha_i+\gamma_t).
$$

\citet{browning2014dynamic} also report the \emph{long-run proportion of unit values} implied by a time-homogeneous two state Markov chain. In our two-way fixed effects setting, $\gamma_t$ is unrestricted and the implied transitions need not be time invariant. We therefore evaluate the long-run proportion \emph{locally} by freezing $(Z_{it},\alpha_i+\gamma_t)$ at each $(i,t)$ and applying their stationary formula:
\[
\Delta_{L}(Z_{it},\alpha_i+\gamma_t)
:=
\frac{P_{01}(Z_{it},\alpha_i+\gamma_t)}{1-P_{11}(Z_{it},\alpha_i+\gamma_t)+P_{01}(Z_{it},\alpha_i+\gamma_t)}.
\]
We summarize this local \emph{long-run participation probability} by the same population-scaled normalization used for APEs,
\[
\Delta_L(\beta,\phi)
=\tfrac{1}{1461\times 9}\textstyle\sum_{i=1}^{664}\sum_{t=1}^{9}\Delta_{L,it}(Z_{it},\alpha_i+\gamma_t).
\]

Table~\ref{table:laborStaticResults} reports results for four estimators: (i) the uncorrected fixed effects MLE; (ii) the analytical correction of \citet{fernandez2016individual} (FW16); (iii) the bias-reducing penalized likelihood estimator (Penalty); and (iv) our posterior-mean estimator (Prior).\footnote{
    For the cross-time spectral expectations entering the dynamic corrections we use $L=1$ given the short panel length ($T=9$).
} We employ the closed form corrections (SE) for the static specification and (PE) for the dynamic case. For the prior-based approach, we report the posterior plug-in estimates for the APEs. Standard errors are derived from analytical asymptotic formulae for the likelihood based estimators.

The posterior quantities were computed using the algorithm detailed in Appendix \ref{app:mcmc}.\footnote{As an empirical demonstration, we deliberately run a long MCMC chain so that the reported estimates are numerically stable in this high-dimensional two-way panel setting.
    Specifically, we run 150,000 iterations, discard the first 50,000 as burn-in, and thin every 20th draw thereafter.
    The reported estimates are not sensitive to shorter runs.} Convergence diagnostics indicate excellent performance. The \citet{geweke1992evaluating} test fails to reject the null hypothesis of stationarity for the key dynamic parameters (Table \ref{table:geweke}), and Figure~\ref{fig:empirical_combined} shows stable traces and fast-decaying autocorrelation, confirming good mixing.

Turning to Table~\ref{table:laborStaticResults}, the uncorrected estimates in both specifications exhibit the inflated bias predicted by the theory. All three corrected estimators move the coefficients in the same direction, which is consistent with bias reduction.
In the static model, the estimates from FW16, Penalty SE, and Prior SE are very similar, and the standard errors are almost unchanged. This provides a useful benchmark. In a setting where an analytical correction is already available, the likelihood-based implementations lead to essentially the same empirical conclusions.

The dynamic model is also the case in which the likelihood-based approach goes beyond the existing analytical correction. Using Prior PE and its penalized analogue, the largest changes arise for the state-dependence parameter and its associated APE. The APE for Participation$_{t-1}$, equivalently the average marginal dynamic effect $\Delta_M$ in \citet{browning2014dynamic}, increases under all bias corrections. This points to stronger persistence once incidental-parameter bias is taken into account. 

The likelihood-based correction extends directly to nonlinear dynamic functionals. In particular, Table \ref{table:laborStaticResults} reports the local long-run participation probability $\Delta_L$. We do not report an FW16-corrected counterpart because their analytical correction is derived for the index parameters and APEs, whereas $\Delta_L$ is a nonlinear Markov functional and would require a separate derivation. In our framework, once the corrected model parameters are available, $\Delta_L$ follows directly by plug-in together with the same variance-type adjustment. The similarity of the penalty and prior estimates for $\Delta_L$ is also reassuring for the substantive conclusion.

\section{Concluding Remarks}\label{sec:conclusion}
This paper develops a practical integrated likelihood approach to bias reduction in nonlinear panel models with additive two-way fixed effects under large-$N,T$ asymptotics. Because maximizing the integrated likelihood is infeasible in this setting, we implement the correction through posterior means under a carefully constructed prior, while maintaining an entirely frequentist interpretation.

Our main theoretical contribution is to make integrated-likelihood debiasing operational in the two-way regime, where classical Laplace approximations can fail. A target-centered full-exponential Laplace expansion, together with a cumulant-based argument, yields valid asymptotic approximations for posterior moments.  This delivers robust priors that are bias-reducing for both common parameters and fixed effects without variance inflation, and supports valid inference for APE after a simple variance-type adjustment. For empirical practice, we derive model-specific priors and companion joint-likelihood penalties for workhorse specifications and provide implementation guidance. 

Our approach opens several avenues for future research.  The theoretical machinery developed here, particularly the handling of slow-converging nuisance parameters via sparse cumulant expansions, could be adapted to models with interactive fixed effects or network dependencies. Extending the framework to unbalanced panels or to semiparametric settings, and developing corresponding bias-reducing priors, would further expand the set of tools available to empirical researchers.

	\newpage
	\section*{Figures \& Tables}
	
	\begin{figure}[h]
		\centering
		\begin{subfigure}{0.8\linewidth}
			\centering
			\includegraphics[width=\linewidth]{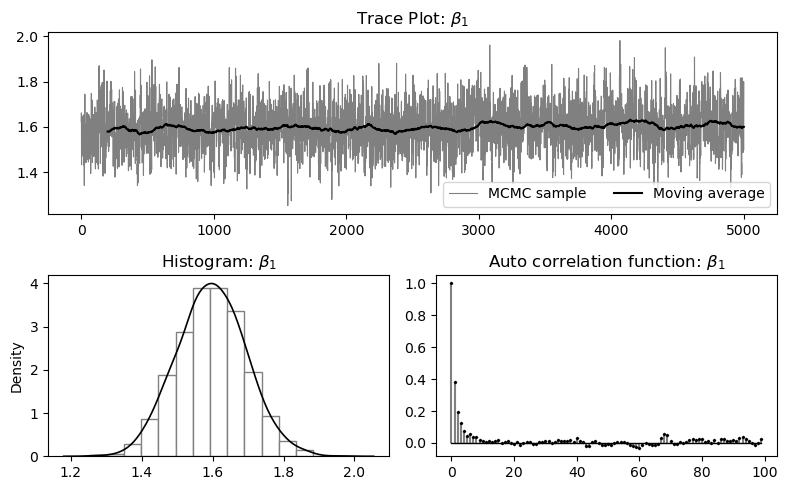}
			\caption{Diagnostic plots of coefficient on \textit{Participation}$_{t-1}$}
			\label{fig:empirical_participation}
		\end{subfigure}
		
		\vspace{1em}
		
		\begin{subfigure}{0.8\linewidth}
			\centering
			\includegraphics[width=\linewidth]{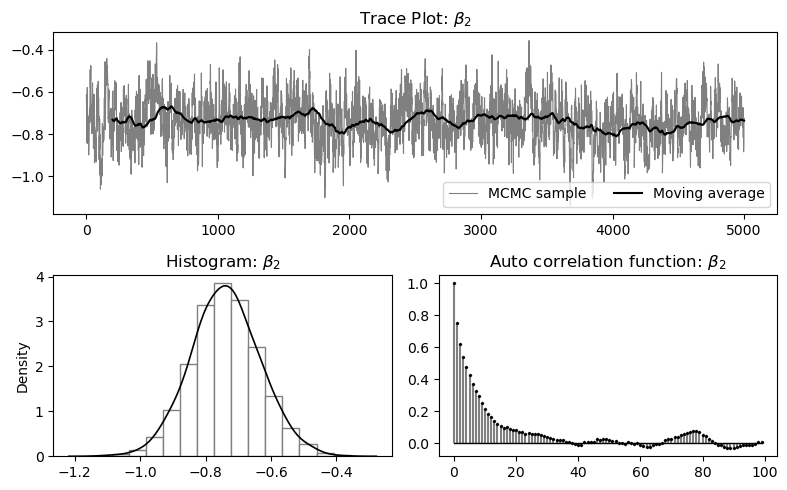}
			\caption{Diagnostic plots of coefficient on \textit{Kids 0--2}}
			\label{fig:empirical_kids}
		\end{subfigure}
		
		\caption{Female labor force participation --- MCMC diagnostics for selected coefficients in the dynamic logit model}
		\label{fig:empirical_combined}
	\end{figure}

	\begin{landscape}
		{\fontsize{10pt}{12.7pt}\selectfont
			\begin{longtable}{>{\centering\arraybackslash}m{0.2cm}p{2.3cm}rcccp{1cm}rcccp{1cm}rccc} 
				\caption{Finite-sample properties in  static and dynamic ordered logit models ($N=45$)} \label{Sim:ologit} \\
				\toprule
				\multicolumn{2}{c}{} 
				& \multicolumn{4}{c}{Model coefficients} 
				& 
				& \multicolumn{4}{c}{\begin{tabular}[c]{@{}c@{}}Average partial effects\\on $\Pr(Y=2|\cdot)$\end{tabular}} 
				& 
				& \multicolumn{4}{c}{\begin{tabular}[c]{@{}c@{}}Average partial effects\\on $\Pr(Y=4|\cdot)$\end{tabular}} \\ 
				\cline{3-6} \cline{8-11} \cline{13-16}
				\multicolumn{2}{c}{} 
				& Bias  & SD & SE/SD & $\widehat{p}_{.95}$ 
				& 
				& Bias  & SD & SE/SD & $\widehat{p}_{.95}$ 
				& 
				& Bias & SD & SE/SD & $\widehat{p}_{.95}$  \\
				\midrule
				
				\multicolumn{16}{l}{\cellcolor{gray!7}\textit{A: Static ordered logit model}} \\
				\hline
				& & \multicolumn{13}{l}{$T=15$} \\
				\cline{3-16}
				\multirow[c]{3}{*}{$\beta_{Z}$}
				& Uncorrected      & 6.3  & 0.12 & 0.96 & 0.91    & & -2.4  & 0.02 & 0.91 & 0.90   & & -0.5  & 0.02 & 0.91 & 0.91 \\
				& Prior GE-1       & -0.5 & 0.11 & 0.99 & 0.95    & & -2.1  & 0.02 & 0.85 & 0.89   & & -0.4  & 0.02 & 0.89 & 0.90 \\
				& Penalty GE-1     & -0.4 & 0.11 & 0.99 & 0.95    & & -1.1  & 0.02 & 0.85 & 0.89   & & -0.2  & 0.02 & 0.90 & 0.91 \\
				\cdashline{1-16}
				\multirow[c]{3}{*}{$\tau_{2}$}
				& Uncorrected      & 39.3 & 0.16 & 1.03 & 0.80    & &      &      &      &        & &      &      &      &       \\
				& Prior GE-1       & -1.8 & 0.15 & 1.03 & 0.97    & &      &      &      &        & &      &      &      &       \\
				& Penalty GE-1     & 3.0  & 0.15 & 1.02 & 0.96    & &      &      &      &        & &      &      &      &       \\
				\cdashline{1-16}
				\multirow[c]{3}{*}{$\tau_{3}$}
				& Uncorrected      & 13.1 & 0.23 & 0.95 & 0.69    & &      &      &      &        & &      &      &      &       \\
				& Prior GE-1       & -0.2 & 0.22 & 0.95 & 0.94    & &      &      &      &        & &      &      &      &       \\
				& Penalty GE-1     & 1.2  & 0.22 & 0.94 & 0.95    & &      &      &      &        & &      &      &      &       \\
				\cline{1-16}
				
				& & \multicolumn{13}{l}{$T=45$} \\
				\cline{3-16}
				\multirow[c]{3}{*}{$\beta_{Z}$}
				& Uncorrected      & 3.1  & 0.06 & 0.98 & 0.92    & & -1.0 & 0.01 & 0.96 & 0.94   & & 0.1  & 0.01 & 0.93 & 0.93 \\
				& Prior GE-1       & -0.3 & 0.06 & 1.00 & 0.94    & & -1.3 & 0.01 & 0.92 & 0.93   & & -0.6 & 0.01 & 0.93 & 0.92 \\
				& Penalty GE-1     & 0.2  & 0.06 & 1.00 & 0.94    & & -0.5 & 0.01 & 0.92 & 0.93   & & -0.5 & 0.01 & 0.93 & 0.93 \\
				\cdashline{1-16}
				\multirow[c]{3}{*}{$\tau_{2}$}
				& Uncorrected      & 19.2 & 0.10 & 0.97 & 0.84    & &      &      &      &        & &      &      &      &       \\
				& Prior GE-1       & 0.5  & 0.12 & 0.98 & 0.94    & &      &      &      &        & &      &      &      &       \\
				& Penalty GE-1     & 1.1  & 0.09 & 0.98 & 0.94    & &      &      &      &        & &      &      &      &       \\
				\cdashline{1-16}
				\multirow[c]{3}{*}{$\tau_{3}$}
				& Uncorrected      & 6.5  & 0.13 & 0.97 & 0.74    & &      &      &      &        & &      &      &      &       \\
				& Prior GE-1       & -0.4 & 0.12 & 0.98 & 0.95    & &      &      &      &        & &      &      &      &       \\
				& Penalty GE-1     & 0.5  & 0.12 & 0.98 & 0.94    & &      &      &      &        & &      &      &      &       \\
				\midrule
				
				\multicolumn{16}{l}{\cellcolor{gray!7}\textit{B: Dynamic ordered logit model}} \\
				\hline
				& & \multicolumn{13}{l}{$T=15$} \\
				\cline{3-16}
				\multirow[c]{3}{*}{$\beta_{Y}$}
				& Uncorrected      & -33.2 & 0.19 & 0.93 & 0.83    & & -39.6 & 0.02 & 0.94 & 0.72    & & -37.1 & 0.02 & 0.93 & 0.77 \\
				& Prior GE-1       & -5.7  & 0.18 & 0.97 & 0.94    & & -7.7  & 0.02 & 0.91 & 0.91    & & -6.7  & 0.02 & 0.97 & 0.92 \\
				& Penalty GE-1     & -5.4  & 0.18 & 0.97 & 0.94    & & -7.3  & 0.02 & 0.90 & 0.91    & & -7.1  & 0.02 & 0.96 & 0.92 \\
				\cdashline{1-16}
				\multirow[c]{3}{*}{$\beta_{Z}$}
				& Uncorrected      & 8.6 & 0.12 & 0.95 & 0.87    & & -1.3 & 0.02 & 0.93 & 0.93    & & 1.7 & 0.02 & 0.89 & 0.92 \\
				& Prior GE-1       & 1.2 & 0.11 & 0.98 & 0.96    & & -2.6 & 0.02 & 0.87 & 0.91    & & 0.3 & 0.02 & 0.88 & 0.91 \\
				& Penalty GE-1     & 1.8 & 0.11 & 0.98 & 0.96    & & -1.8 & 0.02 & 0.87 & 0.92    & & 0.4 & 0.02 & 0.88 & 0.91 \\
				\cdashline{1-16}
				\multirow[c]{3}{*}{$\tau_{2}$}
				& Uncorrected      & 41.5 & 0.20 & 0.91 & 0.78    & &      &      &      &        & &      &      &      &       \\
				& Prior GE-1       & 3.8  & 0.18 & 0.91 & 0.94    & &      &      &      &        & &      &      &      &       \\
				& Penalty GE-1     & -8.5 & 0.19 & 0.91 & 0.97    & &      &      &      &        & &      &      &      &       \\
				\cdashline{1-16}
				\multirow[c]{3}{*}{$\tau_{3}$}
				& Uncorrected      & 13.8 & 0.25 & 0.91 & 0.66    & &      &      &      &        & &      &      &      &       \\
				& Prior GE-1       & 1.3  & 0.18 & 0.91 & 0.94    & &      &      &      &        & &      &      &      &       \\
				& Penalty GE-1     & 1.8  & 0.11 & 0.98 & 0.96    & &      &      &      &        & &      &      &      &       \\
				\cline{1-16}
				
				& & \multicolumn{13}{l}{$T=45$} \\
				\cline{3-16}
				\multirow[c]{3}{*}{$\beta_{Y}$}
				& Uncorrected      & -12.5 & 0.10 & 1.00 & 0.92    & & -16.4 & 0.01 & 0.99 & 0.84    & & -14.8 & 0.01 & 1.00 & 0.86 \\
				& Prior GE-1       & -2.9  & 0.10 & 1.01 & 0.96    & & -3.5  & 0.01 & 0.95 & 0.93    & & -3.2  & 0.01 & 1.00 & 0.95 \\
				& Penalty GE-1     & -2.8  & 0.10 & 1.01 & 0.96    & & -3.4  & 0.01 & 0.95 & 0.93    & & -3.1  & 0.01 & 1.00 & 0.95 \\
				\cdashline{1-16}
				\multirow[c]{3}{*}{$\beta_{Z}$}
				& Uncorrected      & 4.3 & 0.06 & 0.94 & 0.87    & & -0.2 & 0.01 & 0.97 & 0.95    & & 1.2 & 0.01 & 0.95 & 0.94 \\
				& Prior GE-1       & 0.9 & 0.06 & 0.96 & 0.94    & & -0.5 & 0.01 & 0.93 & 0.94    & & 0.6 & 0.01 & 0.93 & 0.93 \\
				& Penalty GE-1     & 0.9 & 0.06 & 0.95 & 0.94    & & -0.4 & 0.01 & 0.94 & 0.95    & & 0.5 & 0.01 & 0.93 & 0.94 \\
				\cdashline{1-16}
				\multirow[c]{3}{*}{$\tau_{2}$}
				& Uncorrected      & 19.0 & 0.10 & 1.02 & 0.86    & &      &      &      &        & &      &      &      &       \\
				& Prior GE-1       & 1.0  & 0.09 & 1.02 & 0.94    & &      &      &      &        & &      &      &      &       \\
				& Penalty GE-1     & 2.1  & 0.09 & 1.01 & 0.94    & &      &      &      &        & &      &      &      &       \\
				\cdashline{1-16}
				\multirow[c]{3}{*}{$\tau_{3}$}
				& Uncorrected      & 6.4 & 0.12 & 1.01 & 0.77     & &      &      &      &        & &      &      &      &       \\
				& Prior GE-1       & 0.4 & 0.12 & 1.01 & 0.95     & &      &      &      &        & &      &      &      &       \\
				& Penalty GE-1     & 0.6 & 0.12 & 1.01 & 0.96     & &      &      &      &        & &      &      &      &       \\
				\bottomrule
				\multicolumn{16}{@{}p{\dimexpr\linewidth-2\tabcolsep\relax}@{}}{%
					\scriptsize\quad\emph{Notes:} 500 replications. ``Uncorrected'' is the standard fixed-effects MLE. 
					``FW16'' is the analytical correction of \citet{fernandez2016individual}. 
					``Penalty'' and ``Prior'' are the likelihood-based corrections proposed in this paper, implemented with the model-specific choices indicated in the row labels.
					Bias is reported in percent; SD is the Monte Carlo standard deviation; $\widehat p_{.95}$ is the empirical 95\% coverage. 
					Posterior means are computed by MCMC (26{,}000 iterations, 6{,}000 burn-in, thinning every 10 draws).%
				}
			\end{longtable}
		} 
	\end{landscape}

	{\fontsize{10pt}{11pt}\selectfont 
		\begin{longtable}{>{\centering\arraybackslash}m{0.6cm}p{2.8cm}rcccp{1cm}rccc}
			\caption{Finite-sample properties in  static and dynamic logit models ($N=45$)} \label{Sim:logit} \\
			\toprule
			\multicolumn{2}{c}{} & \multicolumn{4}{c}{Model coefficients} & & \multicolumn{4}{c}{Average partial effects} \\ 
			\cline{3-6} \cline{8-11}
			\multicolumn{2}{c}{} & Bias  & SD & SE/SD & $\widehat{p}_{.95}$ & & Bias & SD & SE/SD & $\widehat{p}_{.95}$  \\
			\midrule
			
			\multicolumn{11}{l}{\cellcolor{gray!7}\textit{A: Static logit model}} \\
			\hline
			& & \multicolumn{8}{l}{$T=15$} \\
			\cline{3-11}
			\multirow[c]{4}{*}{$\beta_{Z}$}
			& Uncorrected         & 10.1 & 0.15 & 0.94 & 0.89    & & -0.6 & 0.02 & 0.96 & 0.94 \\
			& FW16 Analytical     & -1.0 & 0.13 & 1.06 & 0.97    & & -1.4 & 0.02 & 0.97 & 0.94 \\
			& FW16 Jackknife     & -4.1 & 0.16 & 0.87 & 0.90    & & -1.8 & 0.03 & 0.77 & 0.85 \\
			& Prior SE            & -1.0 & 0.13 & 1.01 & 0.95    & & -1.0 & 0.02 & 0.97 & 0.94 \\
			& Penalty SE          & -0.2 & 0.13 & 1.01 & 0.95    & & -0.6 & 0.02 & 0.97 & 0.94 \\
			\cline{1-11}
			
			& & \multicolumn{8}{l}{$T=45$} \\
			\cline{3-11}
			\multirow[c]{4}{*}{$\beta_{Z}$}
			& Uncorrected         & 4.8  & 0.08 & 0.98 & 0.90    & & -0.2 & 0.01 & 1.02 & 0.95 \\
			& FW16 Analytical     & -0.4 & 0.07 & 1.03 & 0.94    & & -0.4 & 0.01 & 1.02 & 0.96 \\
			& FW16 Jackknife     & -1.0 & 0.08 & 0.95 & 0.92    & & -0.4 & 0.01 & 0.83 & 0.89 \\
			& Prior SE            & -0.9 & 0.07 & 1.01 & 0.94    & & -0.6 & 0.01 & 1.02 & 0.95 \\
			& Penalty SE          & -0.1 & 0.07 & 1.01 & 0.94    & & -0.2 & 0.01 & 1.02 & 0.95 \\
			\midrule
			
			\multicolumn{11}{l}{\cellcolor{gray!7}\textit{B: Dynamic logit model}} \\
			\hline
			& & \multicolumn{8}{l}{$T=15$} \\
			\cline{3-11}
			\multirow[c]{6}{*}{$\beta_{Y_{t-1}}$}
			& Uncorrected         & -64.1 & 0.21 & 0.98 & 0.66    & & -67.9 & 0.04 & 0.99 & 0.55 \\
			& FW16 Analytical     & -1.3  & 0.19 & 1.09 & 0.97    & & -1.6  & 0.04 & 0.95 & 0.93 \\
			& FW16 Jackknife     &  5.5 & 0.24 & 0.88 & 0.91    & & -10.7 & 0.04 & 0.80 & 0.87 \\
			& Prior SE+AC         & -5.6  & 0.19 & 1.06 & 0.96    & & -5.7  & 0.04 & 0.98 & 0.94 \\
			& Penalty SE+AC       & -3.7  & 0.19 & 1.06 & 0.96    & & -3.8  & 0.04 & 0.98 & 0.94 \\
			& Prior PE            & -13.0 & 0.19 & 1.05 & 0.95    & & -13.9 & 0.04 & 1.00 & 0.95 \\
			& Penalty PE          & -11.6 & 0.19 & 1.05 & 0.95    & & -12.9 & 0.04 & 1.00 & 0.94 \\
			\cdashline{1-11}
			\multirow[c]{6}{*}{$\beta_{Z}$}
			& Uncorrected         & 15.5 & 0.16 & 0.91 & 0.82    & & 2.8  & 0.02 & 0.98 & 0.94 \\
			& FW16 Analytical     & 0.7  & 0.14 & 1.06 & 0.98    & & -0.5 & 0.02 & 1.00 & 0.95 \\
			& FW16 Jackknife     & -3.2 & 0.17 & 0.87 & 0.90    & &0.8 & 0.03 & 0.81 & 0.88 \\
			& Prior SE+AC         & 0.6  & 0.14 & 1.00 & 0.97    & & -0.3 & 0.02 & 1.00 & 0.95 \\
			& Penalty SE+AC       & 1.1  & 0.14 & 1.00 & 0.97    & & 0.0  & 0.02 & 1.00 & 0.96 \\
			& Prior PE            & 2.7  & 0.14 & 0.99 & 0.96    & & 0.4  & 0.02 & 0.99 & 0.95 \\
			& Penalty PE          & 3.1  & 0.14 & 0.99 & 0.95    & & 0.3  & 0.02 & 0.99 & 0.95 \\
			\cline{1-11}
			
			& & \multicolumn{8}{l}{$T=45$} \\
			\cline{3-11}
			\multirow[c]{6}{*}{$\beta_{Y_{t-1}}$}
			& Uncorrected         & -19.7 & 0.11 & 1.00 & 0.87    & & -24.3 & 0.02 & 1.00 & 0.79 \\
			& FW16 Analytical     & -0.3  & 0.11 & 1.03 & 0.96    & & -0.5  & 0.02 & 0.97 & 0.95 \\
			& FW16 Jackknife     & -0.2 & 0.11 & 1.01 & 0.95    & & -2.2 & 0.02 & 0.95 & 0.93 \\
			& Prior SE+AC         & -2.3  & 0.11 & 1.03 & 0.96    & & -2.4  & 0.02 & 0.99 & 0.95 \\
			& Penalty SE+AC       & -0.7  & 0.11 & 1.02 & 0.96    & & -0.8  & 0.02 & 0.98 & 0.95 \\
			& Prior PE            & -4.2  & 0.11 & 1.03 & 0.96    & & -4.4  & 0.02 & 1.00 & 0.95 \\
			& Penalty PE          & -2.5  & 0.11 & 1.02 & 0.96    & & -3.0  & 0.02 & 0.99 & 0.95 \\
			\cdashline{1-11}
			\multirow[c]{6}{*}{$\beta_{Z}$}
			& Uncorrected         & 7.2 & 0.09 & 0.89 & 0.80    & & 1.5 & 0.01 & 0.94 & 0.93 \\
			& FW16 Analytical     & 0.7 & 0.08 & 0.95 & 0.95    & & 0.2 & 0.01 & 0.94 & 0.94 \\
			& FW16 Jackknife     & -0.4 & 0.09 & 0.88 & 0.91    & & 0.2 & 0.01 & 0.81 & 0.89 \\
			& Prior SE+AC         & 0.1 & 0.08 & 0.93 & 0.94    & & -0.2 & 0.01 & 0.93 & 0.94 \\
			& Penalty SE+AC       & 0.8 & 0.08 & 0.93 & 0.94    & & 0.3 & 0.01 & 0.94 & 0.94 \\
			& Prior PE            & 0.7 & 0.08 & 0.93 & 0.94    & & 0.0 & 0.01 & 0.94 & 0.93 \\
			& Penalty PE          & 1.3 & 0.08 & 0.92 & 0.94    & & 0.3 & 0.01 & 0.94 & 0.94 \\
			\bottomrule
			
			\multicolumn{11}{l}{\legend{Same as Table~\ref{Sim:ologit}. ``SE+AC'' denotes SE augmented with the add-on cross-time correction in Section~\ref{subsection:dynamicexp}.
			}}
		\end{longtable}
	}

	\begin{landscape}
		{\fontsize{9pt}{12.5pt}\selectfont
			\begin{longtable}{>{\centering\arraybackslash}m{0.2cm}lcccccccc p{0.1cm}cccccccc}
				\caption{Finite-sample properties in  static and dynamic panel multinomial logit models ($N=45$)}\label{Sim:mlogit} \\
				\toprule
				\multicolumn{2}{c}{} & \multicolumn{8}{c}{Model coefficients} & & \multicolumn{8}{c}{Average partial effects} \\ 
				\cline{3-10} \cline{12-19}
				\multicolumn{2}{c}{} 
				& \multicolumn{4}{c}{on $\Pr(Y=2|\cdot)$} & \multicolumn{4}{c}{on $\Pr(Y=3|\cdot)$} 
				& 
				& \multicolumn{4}{c}{on $\Pr(Y=2|\cdot)$} & \multicolumn{4}{c}{on $\Pr(Y=3|\cdot)$} \\ 
				\cline{3-6} \cline{7-10} \cline{12-15} \cline{16-19}
				\multicolumn{2}{c}{}
				& Bias  & SD & SE/SD & $\widehat{p}_{.95}$ & Bias  & SD & SE/SD & $\widehat{p}_{.95}$ 
				& 
				& Bias  & SD & SE/SD & $\widehat{p}_{.95}$ & Bias  & SD & SE/SD & $\widehat{p}_{.95}$ \\
				\midrule
				
				\multicolumn{19}{l}{\cellcolor{gray!7} \textit{A: Static panel multinomial logit model}} \\
				\hline
				& & \multicolumn{16}{l}{$T=15$} \\
				\cline{3-19}
				\multirow[c]{3}{*}{$\beta_{Z}$}
				& Uncorrected      & 12.7 & 0.17 & 0.94 & 0.88 & 12.6 & 0.18 & 0.90 & 0.87 & & 5.5 & 0.02 & 1.00 & 0.94 & 5.2 & 0.02 & 0.98 & 0.95 \\
				& Prior SML        & 1.6  & 0.16 & 1.01 & 0.95 & 2.0  & 0.16 & 0.96 & 0.94 & & 3.4 & 0.02 & 0.99 & 0.94 & 4.7 & 0.02 & 0.96 & 0.94 \\
				& Penalty SML      & 2.0  & 0.15 & 1.01 & 0.95 & 1.9  & 0.16 & 0.97 & 0.94 & & 4.3 & 0.02 & 0.98 & 0.93 & 3.9 & 0.02 & 0.96 & 0.94 \\
				\cline{1-19}
				
				& & \multicolumn{16}{l}{$T=45$} \\
				\cline{3-19}
				\multirow[c]{3}{*}{$\beta_{Z}$}
				& Uncorrected      & 5.8  & 0.09 & 0.96 & 0.91 & 5.3  & 0.09 & 0.97 & 0.91 & & 7.0 & 0.01 & 1.06 & 0.94 & 5.5 & 0.01 & 1.04 & 0.95 \\
				& Prior SML        & -0.1 & 0.08 & 1.00 & 0.96 & -0.4 & 0.08 & 1.02 & 0.95 & & 5.3 & 0.01 & 1.04 & 0.94 & 4.2 & 0.01 & 1.04 & 0.95 \\
				& Penalty SML      & 0.8  & 0.08 & 1.00 & 0.95 & 0.3  & 0.08 & 1.01 & 0.95 & & 6.1 & 0.01 & 1.04 & 0.94 & 4.6 & 0.01 & 1.02 & 0.95 \\
				\midrule
				
				\multicolumn{19}{l}{\cellcolor{gray!7} \textit{B: Dynamic panel multinomial logit model}} \\
				\hline
				& & \multicolumn{16}{l}{$T=15$} \\
				\cline{3-19}
				\multirow[c]{3}{*}{$\beta_{Y}$}
				& Uncorrected      & -60.6 & 0.25 & 0.98 & 0.77 & 11.0 & 0.26 & 0.95 & 0.93 & & -29.5 & 0.02 & 0.96 & 0.86 & -29.1 & 0.02 & 0.97 & 0.88 \\
				& Prior PML        & -8.3  & 0.23 & 1.06 & 0.95 & 3.0  & 0.24 & 1.03 & 0.96 & & -2.4  & 0.02 & 0.99 & 0.94 & -1.7  & 0.02 & 1.00 & 0.96 \\
				& Penalty PML      & -8.1  & 0.23 & 1.05 & 0.95 & 3.2  & 0.24 & 1.02 & 0.96 & & -2.0  & 0.02 & 0.98 & 0.94 & -1.5  & 0.02 & 0.99 & 0.94 \\
				\cdashline{1-19}
				\multirow[c]{3}{*}{$\beta_{Z}$}
				& Uncorrected      & 14.4 & 0.18 & 0.93 & 0.87 & 13.2 & 0.18 & 0.93 & 0.88 & & 7.5 & 0.02 & 0.96 & 0.93 & 3.6 & 0.02 & 0.96 & 0.94 \\
				& Prior PML        & 3.9  & 0.16 & 1.00 & 0.94 & 4.2  & 0.16 & 0.99 & 0.95 & & 4.0 & 0.02 & 0.96 & 0.93 & 5.6 & 0.03 & 0.95 & 0.94 \\
				& Penalty PML      & 3.5  & 0.16 & 1.00 & 0.94 & 3.6  & 0.16 & 0.99 & 0.95 & & 4.0 & 0.02 & 0.95 & 0.93 & 4.7 & 0.03 & 0.94 & 0.93 \\
				\cline{1-19}
				
				& & \multicolumn{16}{l}{$T=45$} \\
				\cline{3-19}
				\multirow[c]{3}{*}{$\beta_{Y}$}
				& Uncorrected      & -16.3 & 0.14 & 0.98 & 0.90 & 6.2 & 0.13 & 1.05 & 0.96 & & -6.1 & 0.01 & 1.04 & 0.95 & -5.8 & 0.01 & 1.03 & 0.96 \\
				& Prior PML        & -0.2  & 0.13 & 1.00 & 0.95 & 3.8 & 0.12 & 1.09 & 0.97 & & 3.6  & 0.01 & 1.04 & 0.95 & 4.7  & 0.01 & 1.01 & 0.96 \\
				& Penalty PML      & -0.3  & 0.13 & 1.00 & 0.95 & 2.1 & 0.12 & 1.09 & 0.97 & & 3.7  & 0.01 & 1.05 & 0.95 & 4.1  & 0.01 & 1.04 & 0.95 \\
				\cdashline{1-19}
				\multirow[c]{3}{*}{$\beta_{Z}$}
				& Uncorrected      & 6.4 & 0.09 & 1.00 & 0.90 & 6.0 & 0.09 & 0.95 & 0.88 & & 6.7 & 0.01 & 1.02 & 0.94 & 5.4 & 0.01 & 0.97 & 0.93 \\
				& Prior PML        & 1.4 & 0.07 & 1.04 & 0.96 & 1.2 & 0.08 & 0.98 & 0.94 & & 4.8 & 0.01 & 1.02 & 0.95 & 5.5 & 0.01 & 0.99 & 0.94 \\
				& Penalty PML      & 1.1 & 0.08 & 1.04 & 0.96 & 1.1 & 0.09 & 0.98 & 0.94 & & 4.8 & 0.01 & 1.01 & 0.95 & 5.4 & 0.01 & 0.95 & 0.93 \\
				\bottomrule
				
				\multicolumn{19}{@{}p{\dimexpr\linewidth-2\tabcolsep\relax}@{}}{%
					\scriptsize\quad\emph{Notes:} Same as Table~\ref{Sim:ologit}. 
				}
			\end{longtable}
		}
	\end{landscape}

	{\fontsize{10pt}{13pt}\selectfont 
		\begin{longtable}{>{\centering\arraybackslash}m{0.6cm}p{3.4cm}rcccp{1cm}rccc}
			\caption{Finite-sample properties in static and dynamic probit models ($N=45$)} \label{Sim:probit} \\
			\toprule
			\multicolumn{2}{c}{} & \multicolumn{4}{c}{Model coefficients} & & \multicolumn{4}{c}{Average partial effects} \\ 
			\cline{3-6} \cline{8-11}
			\multicolumn{2}{c}{} & Bias & SD & SE/SD & $\widehat{p}_{.95}$ & & Bias & SD & SE/SD & $\widehat{p}_{.95}$  \\
			\midrule
			
			\multicolumn{11}{l}{\cellcolor{gray!7}\textit{A: Static probit model}} \\
			\hline
			& & \multicolumn{8}{l}{$T=15$} \\
			\cline{3-11}
			\multirow[c]{4}{*}{$\beta_{Z}$} 
			& Uncorrected           & 13.3 & 0.12 & 0.92 & 0.76    & & 0.7  & 0.02 & 1.03 & 0.96 \\
			& FW16 Analytical       & 0.4  & 0.10 & 1.09 & 0.97    & & -0.1 & 0.02 & 1.03 & 0.96 \\
			& FW16 Jackknife      & -3.5  & 0.12 & 0.89 & 0.91    & & 0.3  & 0.02 & 0.87 & 0.91 \\
			& Prior GE-1            & 2.9  & 0.10 & 0.98 & 0.95    & & 1.4  & 0.02 & 1.00 & 0.95 \\
			& Penalty GE-1          & 4.0  & 0.10 & 0.98 & 0.94    & & 1.4  & 0.02 & 1.01 & 0.95 \\
			\cline{1-11}
			
			& & \multicolumn{8}{l}{$T=45$} \\
			\cline{3-11}
			\multirow[c]{4}{*}{$\beta_{Z}$} 
			& Uncorrected           & 5.5  & 0.06 & 0.94 & 0.85    & & 0.1  & 0.01 & 1.00 & 0.94 \\
			& FW16 Analytical       & -0.2 & 0.05 & 1.01 & 0.95    & & -0.1 & 0.01 & 1.00 & 0.95 \\
			& FW16 Jackknife      & -1.1  & 0.06 & 0.95 & 0.93    & & 0.1  & 0.01 & 0.96 & 0.94 \\
			& Prior GE-1            & 0.1  & 0.05 & 0.98 & 0.95    & & 0.2  & 0.01 & 0.98 & 0.95 \\
			& Penalty GE-1          & 0.6  & 0.06 & 0.97 & 0.94    & & 0.3  & 0.01 & 0.98 & 0.94 \\
			\midrule
			
			\multicolumn{11}{l}{\cellcolor{gray!7}\textit{B: Dynamic probit model}} \\
			\hline
			& & \multicolumn{8}{l}{$T=15$} \\
			\cline{3-11}
			\multirow[c]{4}{*}{$\beta_{Y_{t-1}}$} 
			& Uncorrected           & -39.8 & 0.15 & 1.00 & 0.74    & & -48.7 & 0.03 & 1.01 & 0.53 \\
			& FW16 Analytical       & -4.6  & 0.13 & 1.14 & 0.97    & & -5.6  & 0.03 & 0.97 & 0.93 \\
			& FW16 Jackknife      & 3.7  & 0.17 & 0.87 & 0.92    & & -9.3  & 0.04 & 0.79 & 0.87 \\
			& Prior GE-1            & -15.3 & 0.14 & 1.08 & 0.94    & & -19.1 & 0.03 & 1.04 & 0.90 \\
			& Penalty GE-1          & -14.7 & 0.14 & 1.07 & 0.94    & & -19.2 & 0.03 & 1.04 & 0.90 \\
			\cdashline{1-11}
			\multirow[c]{4}{*}{$\beta_{Z}$} 
			& Uncorrected           & 18.9  & 0.13 & 0.89 & 0.66    & & 3.3 & 0.02 & 1.03 & 0.94 \\
			& FW16 Analytical       & 1.7   & 0.11 & 1.09 & 0.97    & & 0.3 & 0.02 & 1.06 & 0.96 \\
			& FW16 Jackknife      & -5.4  & 0.14 & 0.83 & 0.86    & & 1.7  & 0.03 & 0.89 & 0.92 \\
			& Prior GE-1            & 7.0   & 0.12 & 0.95 & 0.92    & & 2.1 & 0.02 & 1.02 & 0.94 \\
			& Penalty GE-1          & 7.9   & 0.12 & 0.94 & 0.90    & & 2.0 & 0.02 & 1.02 & 0.95 \\
			\cline{1-11}
			
			& & \multicolumn{8}{l}{$T=45$} \\
			\cline{3-11}
			\multirow[c]{4}{*}{$\beta_{Y_{t-1}}$} 
			& Uncorrected           & -11.1 & 0.08 & 0.99 & 0.90    & & -16.9 & 0.02 & 1.00 & 0.82 \\
			& FW16 Analytical       & -0.3  & 0.08 & 1.03 & 0.95    & & -0.3  & 0.02 & 0.97 & 0.94 \\
			& FW16 Jackknife      & 0.3  & 0.09 & 0.92 & 0.92    & & -0.8  & 0.02 & 0.88 & 0.91 \\
			& Prior GE-1            & -2.6  & 0.08 & 1.01 & 0.95    & & -3.3  & 0.02 & 1.01 & 0.94 \\
			& Penalty GE-1          & -2.2  & 0.08 & 1.01 & 0.95    & & -2.0  & 0.02 & 1.01 & 0.95 \\
			\cdashline{1-11}
			\multirow[c]{4}{*}{$\beta_{Z}$} 
			& Uncorrected           & 7.7   & 0.06 & 0.96 & 0.73    & & 1.3 & 0.01 & 1.06 & 0.95 \\
			& FW16 Analytical       & 0.3   & 0.06 & 1.05 & 0.97    & & 0.1 & 0.01 & 1.07 & 0.96 \\
			& FW16 Jackknife      &-1.2  & 0.06 & 0.95 & 0.94    & & 0.1  & 0.01 & 1.02 & 0.95 \\
			& Prior GE-1            & 1.3   & 0.06 & 1.01 & 0.96    & & 0.3 & 0.01 & 1.04 & 0.95 \\
			& Penalty GE-1          & 1.8   & 0.06 & 1.00 & 0.95    & & 0.4 & 0.01 & 1.04 & 0.95 \\
			\bottomrule
			
			\multicolumn{11}{l}{\legend{Same as Table~\ref{Sim:ologit}. }}
		\end{longtable}
	}

	\newpage
	
	{\fontsize{9.5pt}{10.0pt}\selectfont
		\begin{longtable}{
				>{\raggedright\arraybackslash}p{5.0cm}
				>{\centering\arraybackslash}p{2.3cm}
				>{\centering\arraybackslash}p{2.3cm}
				>{\centering\arraybackslash}p{2.3cm}
				>{\centering\arraybackslash}p{2.3cm}
			}
			\caption{Female labor force participation --- Static and dynamic logit models}\label{table:laborStaticResults} \\
			\toprule
			& \multirow{2}{*}{Uncorrected}  & \multicolumn{3}{c}{Corrected}  \\
			\cline{3-5}
			&  & FW16 & Penalty & Prior \\
			\cline{2-5}
			& (i) & (ii) & (iii) & (iv) \\
			\hline
			
			\multicolumn{5}{l}{\cellcolor{gray!7}\textit{A: Static logit model}} \\
			\hline
			\multicolumn{5}{l}{\textit{\underline{Index coefficients ($\beta$)}}} \\
			Kids 0-2           & -1.174***  & -1.027***  & -1.032*** & -0.970*** \\
			& (0.098)    & (0.098)    & (0.093)   & (0.091)   \\
			Kids 3-5           & -0.591***  & -0.518***  & -0.520*** & -0.482*** \\
			& (0.086)    & (0.086)    & (0.082)   & (0.081)   \\
			Kids 6-17          & -0.016     & -0.013     & -0.013    & -0.001    \\
			& (0.061)    & (0.061)    & (0.058)   & (0.057)   \\
			ln(Husband income)   & -0.405***  & -0.357***  & -0.358*** & -0.377*** \\
			& (0.094)    & (0.094)    & (0.090)   & (0.089)   \\
			
			\multicolumn{5}{l}{\textit{\underline{Average partial effects}}} \\
			Kids 0-2           & -0.089**   & -0.098**   & -0.093**  & -0.092**  \\
			& (0.041)    & (0.041)    & (0.040)   & (0.038)   \\
			Kids 3-5           & -0.045**   & -0.049**   & -0.047**  & -0.046**  \\
			& (0.021)    & (0.021)    & (0.021)   & (0.020)   \\
			Kids 6-17          & -0.001     & -0.001     & -0.001    & 0.000     \\
			& (0.005)    & (0.005)    & (0.005)   & (0.005)   \\
			ln(Husband income)   & -0.031*    & -0.034**   & -0.032**  & -0.036**  \\
			& (0.016)    & (0.016)    & (0.015)   & (0.016)   \\
			\hline
			
			\multicolumn{5}{l}{\cellcolor{gray!7}\textit{B: Dynamic logit model}} \\
			\hline
			\multicolumn{5}{l}{\textit{\underline{Index coefficients ($\beta$)}}} \\
			Participation$_{t-1}$ & 1.278***  & 1.730***  &  1.687*** &  1.597***\\
			& (0.072)   & (0.072)     &  (0.072)  &  (0.071) \\
			Kids 0-2            & -0.904*** & -0.710***   & -0.768*** & -0.739*** \\
			& (0.101)   & (0.101)     &    (0.1)  &  (0.098) \\
			Kids 3-5            & -0.379*** & -0.253***   &  -0.27*** & -0.264*** \\
			& (0.090)   & (0.090)     &  (0.089)  &  (0.087) \\
			Kids 6-17           & 0.052     & 0.056       &     0.065 &     0.055 \\
			& (0.063)   & (0.063)     &  (0.063)  &  (0.061) \\
			ln(Husband income)  & -0.407*** & -0.343***   & -0.372*** & -0.287*** \\
			& (0.096)   & (0.096)     &  (0.095)  &  (0.092) \\
			
			\multicolumn{5}{l}{\textit{\underline{Average partial effects}}} \\
			Participation$_{t-1}$ (AMDE)   & 0.107***  & 0.206***  & 0.176*** & 0.172*** \\
			& (0.007)   & (0.007)   & (0.007)  & (0.007) \\
			Kids 0-2                & -0.064**  & -0.067**  & -0.057** & -0.058** \\
			& (0.030)   & (0.030)   & (0.026)  & (0.026) \\
			Kids 3-5               	& -0.027*   & -0.024*   &   -0.02* &  -0.021* \\
			& (0.014)   & (0.014)   & (0.011)  & (0.011) \\
			Kids 6-17              	& 0.004     & 0.005     &    0.005 &    0.004 \\
			& (0.005)   & (0.005)   & (0.005)  & (0.005) \\
			ln(Husband income)      & -0.029*   & -0.032**  & -0.028** &  -0.022* \\
			& (0.015)   & (0.015)   & (0.014)  & (0.012) \\

			\multicolumn{5}{l}{\textit{\underline{Markov summary (long-run proportion)}}} \\
			Long-run participation prob.  & 0.262**  & -  & 0.266** & 0.271** \\
			& (0.121)   &    & (0.120)  & (0.119) \\
			\bottomrule
			
			\multicolumn{5}{l}{\legend{
					Data source: PSID 1980--1988. All specifications include individual and time fixed effects. 
					``Uncorrected'' is the standard fixed-effects maximum likelihood estimator. 
					``FW16'' refers to the analytical bias correction of \citet{fernandez2016individual}. 
					"Penalty" and "Prior" refer to the likelihood-based corrections proposed in this paper, using \textit{SE} specifications for the static one-parameter exponential model and \textit{PE} for the dynamic case. 
					Asymptotic standard errors are in parentheses. ***, **, and * indicate significance at the 1\%, 5\%, and 10\% levels.
			}}
		\end{longtable}
	}

	\begin{table}[]
		\caption{Female labor force participation --- Convergence diagnosis}\label{table:geweke}
		\centering
		{\fontsize{10.0pt}{11pt}\selectfont
			\begin{tabularx}{\linewidth}{
					>{\raggedright\arraybackslash}X
					>{\centering\arraybackslash}p{0.16\linewidth}
					>{\centering\arraybackslash}p{0.20\linewidth}
					>{\centering\arraybackslash}p{0.20\linewidth}
				}
				\toprule
				& \multirow{2}{*}{Acceptance(\%)}  & \multicolumn{2}{c}{\cite{geweke1992evaluating} Convergence Diagnostic Test}  \\
				\cline{3-4}
				&  & Average $p$-value & Smallest $p$-value \\
				\hline
				\multicolumn{4}{l}{\cellcolor{gray!7}\textit{A: Static logit model}} \\
				\hline
				Kids 0–2             & 49.3\%         & 0.868           & 0.652            \\
				Kids 3–5             & 48.0\%         & 0.838           & 0.591            \\
				Kids 6–17            & 51.9\%         & 0.794           & 0.487            \\
				ln(Husband income)   & 36.8\%         & 0.658           & 0.331            \\
				\hline
				\multicolumn{4}{l}{\cellcolor{gray!7}\textit{B: Dynamic logit model}} \\
				\hline
				$\text{Participation}_{t-1}$& 38.5\%        & 0.897       & 0.770    \\
				Kids 0–2             		& 44.3\%        & 0.903       & 0.724    \\
				Kids 3–5             		& 43.3\%        & 0.891       & 0.744    \\
				Kids 6–17            		& 45.1\%        & 0.879       & 0.620    \\
				ln(Husband income)   		& 28.2\%        & 0.644       & 0.093    \\
				\bottomrule 
		\end{tabularx}}
		{\legend{Acceptance(\%) reports the Metropolis--Hastings acceptance rate for the corresponding parameter. 
				\cite{geweke1992evaluating} diagnostics comparing the first 10\% of retained draws to 20 batches from the last 50\%. 
				Average and smallest $p$-values summarize the batch-wise Geweke tests; high values indicate no evidence against convergence.}}
	\end{table}

	
	\newpage
	
	\begin{appendix}

		\section{Computation of Posterior Means and APEs}\label{app:mcmc}
		We adopt a Metropolis--Hastings (MH) sampler with \emph{random blocking} and a \emph{self-adaptive} proposal scheme to compute the posterior means $(\widehat\beta^E,\widehat\phi^E)$ and the associated APE estimators.
		Because two-way additive effects are only identified up to one normalization,
		we impose  $\alpha_1=0$ and estimate the remaining parameters $
		\theta
		=
		(\beta',\alpha_2,\ldots,\alpha_N,\gamma_1,\ldots,\gamma_T)'
		\in\mathbb R^{\dim(\beta)+N+T-1}
		$.\footnote{Alternative normalizations (e.g., $\gamma_1=0$ or
			$\sum_i\alpha_i=0$) are equivalent.}
		
		We use a random-walk MH algorithm, but \emph{do not} update the full parameter vector $\theta$ at each iteration. 
		Since the dimension of parameters is large, a full-dimensional random-walk proposal yields very low acceptance rates of the posterior.
		This is consistent with the standard high-dimensional behavior of random-walk MH: as the proposal dimension increases, the step size must shrink (roughly at rate $d^{-1/2}$ for $d$ updated coordinates) to maintain a non-degenerate acceptance probability, which in turn leads to poor mixing.
		To address this, we adopt \emph{random blocking}. At each iteration we update a small subset of coordinates, allowing the Markov chain to move while still accounting for the dependence across parameters.
		
		The block size $B$ is a tuning parameter. In our implementation, we may set $B=10$ or draw $B$ uniformly from $\{8,9,10,11,12\}$ at each iteration.
		Given $B$, the index set $J\subset\{1,\ldots,\dim(\beta)+N+T-1\}$ is sampled uniformly among all subsets of size $B$. We draw $J$ from the \emph{full} coordinate set, so a block mixes elements of $\beta$, $\alpha$, and $\gamma$.

		\begin{algorithm}[htbp]
			\caption{Random-block  MH with adaptive proposal covariance}\label{algorithm}
			\begin{enumerate}
				\item \textbf{Initialization}. Choose $\theta^{(0)}$ (we use the fixed-effects
				MLE or another consistent starting value). Fix the number of iterations $M$ and
				a burn-in length $M_0$.
				
				\item \textbf{Iteration.} For $m=1,\ldots,M$:
				\begin{enumerate}
					\item Draw a block size $B\in\{8,9,10,11,12\}$ and a coordinate set
					$J$ with $|J|=B$.
					
					\item Let $\theta_{-J}$ denote coordinates not in $J$.
					Propose $\theta^\star$ by keeping $\theta^\star_{-J}=\theta^{(m-1)}_{-J}$ and
					drawing
					$
					\theta^\star_{J}\sim \mathcal N\!\left(\theta^{(m-1)}_{J},\ \lambda\,\mathsf C_{J}\right),
					$
					where $\mathsf C_{J}$ is the $B\times B$ submatrix of 
					$\mathsf C$ corresponding to indices in $J$, and $\lambda$ is a scaling factor.
					
					\item Set $\theta^{(m)}=\theta^\star$ with probability 
					$
					\min\Big\{1,\ \frac{\exp\big(\sqrt{NT}\,\mathcal L(\theta^\star)\big)\,p(\theta^\star)}
					{\exp\big(\sqrt{NT}\,\mathcal L(\theta^{(m-1)})\big)\,p(\theta^{(m-1)})  }\Big\},
					$
					and otherwise set $\theta^{(m)}=\theta^{(m-1)}$.
				\end{enumerate}
				\item  \textbf{Estimators.} After discarding the burn-in draws $\{\theta^{(m)}\}_{m\le M_0}$, compute
				posterior means using the remaining $M^\star=M-M_0$ draws:
				\[
				\widehat{\beta}^E
				=
				\frac{1}{M^\star}\sum_{m=M_0+1}^{M}\beta^{(m)},
				\widehat{\phi}^E
				=
				\frac{1}{M^\star}\sum_{m=M_0+1}^{M}\phi^{(m)},
				\widehat{\Delta}^E
				=
				\frac{1}{M^\star}\sum_{m=M_0+1}^{M}\Delta\!\big(\beta^{(m)},\phi^{(m)}\big),
				\]
				and the posterior plug-in APE is $\Delta(\widehat{\beta}^E,\widehat{\phi}^E)$.
				\item \textbf{Inference.} One may estimate the asymptotic variance formulas in Corollaries~\ref{corollary:asy_parameter}--\ref{corollary:asy_APE}
				by replacing  expectations with sample analogs and evaluating  at
				$(\widehat{\beta}^E,\widehat{\phi}^E)$. 
				Alternatively, we use the sample variance of the posterior draws. For instance,
				the posterior covariance of $\beta$ is estimated by
				$\mathrm{Var}\{\beta^{(m)}\}_{m>M_0}$, and
				analogously for any fixed-dimensional subvector of $\phi$ and for $\Delta$.
			\end{enumerate}
		\end{algorithm}

		Algorithm \ref{algorithm} lists the  implementation details. In Algorithm \ref{algorithm}, the efficiency  depends on the proposal covariance $\lambda\mathsf C$. We determine $\mathsf C$ and $\lambda$ in two stages,
		following standard practice in the MCMC literature (e.g., \citealp{gelman2003bayesian}).
		We initialize $\mathsf C=\mathrm{diag}(\eta^2)$, where $\eta$ is a vector of standard deviations. We choose $\eta$ so that the overall acceptance rate is around 20--25\%. 
		After the pilot run reaches approximate convergence, we estimate the empirical covariance matrix of the draws,
		$\mathsf C=\mathrm{Var}(\theta^{(m)})$.
		We fix the scaling factor according to the rule of thumb for Gaussian random-walk MH \citep{gelman1996efficient}, i.e.,
		$
		\lambda={2.38^2}/{B}.
		$

		\section{Norms Embeddings and Contractions: A Panel Data Toolkit}\label{sec:appA:NotationNorms}
		Norms are as in \citet{fernandez2016individual}.
		This appendix collects the specific norm inequalities and sparsity properties required to bound the remainder terms in our target-centered Laplace expansions. 

		\subsection*{A.1 Embeddings for Vectors and Matrices}
		For a vector $x\in\mathbb R^{\dim\phi}$, the standard embedding yields, 
		$
		\|x\|\ \le\ (\dim\phi)^{1/2-1/q}\ \|x\|_q $ 
		for $q\ge2$, 
		where 
		$
		\|x\|:=\|x\|_2$
		and 
		$\|x\|_q:=\Big(\sum_{j=1}^{\dim\phi}|x_j|^q\Big)^{1/q}$.
		For vectors in the fixed-dimensional $\beta$-space, all norms are equivalent, $\|x\| \asymp \|x\|_q$ uniformly.

		For a matrix $B\in\mathbb{R}^{\dim\phi\times\dim\phi}$, we use the spectral norm $\|B\|$, the Frobenius norm $\|B\|_F=(\sum_{jk}B_{jk}^2)^{1/2}$, and the entrywise $\ell_q$ norm $\|B\|_q=(\sum_{jk}|B_{jk}|^q)^{1/q}$. The relevant embedding inequality is:
		$
		\|B\|_F\ \le\ (\dim\phi)^{\,1-{2}/{q}}\ \|B\|_q
		$ for $q\ge2$.
		When $B\in\mathbb{R}^{\dim\phi\times\dim\beta}$ or $\mathbb{R}^{\dim\beta\times\dim\phi}$ (one fixed dimension), this bound tightens to
		$
		\|B\|_F\ \le\ (\dim\phi)^{\,1/2-1/q}\ \|B\|_q.
		$
		While for $\mathbb{R}^{\dim\beta\times\dim\beta}$ all matrix induced norms are stochastically equivalent.
		These embeddings are essential for controlling the contraction of higher-order tensors against the rate of dimension growth.

		\subsection*{A.2 Contractions for Higher-Order Derivatives} 
		We view the third and fourth derivatives of $\mathcal{L}$ as tensors on the $\phi$-space, denoted by $\mathcal{T} \in (\mathbb{R}^{\dim\phi})^{\otimes 3}$ and $\mathcal{Q} \in (\mathbb{R}^{\dim\phi})^{\otimes 4}$, respectively. We define the mode-1 slices as $\mathcal{T}_i^{(1)} \in (\mathbb{R}^{\dim\phi})^{\otimes 2}$ with entries $(\mathcal{T}_i^{(1)})_{jk} = \mathcal{T}_{ijk}$, and similarly for $\mathcal{Q}_i^{(1)}$.

		At third order, only $\{\partial_{\alpha_i^3}, \partial_{\gamma_t^3}, \partial_{\alpha_i^2\gamma_t}, \partial_{\alpha_i\gamma_t^2}\}$ are nonzero. Consequently, each mode-1 slice $\mathcal{T}_i^{(1)}$ contains only $\mathcal{O}(T)$  or $\mathcal{O}(N)$ non-zero elements.
		Our smoothness assumptions stop at fourth order, where the sparsity pattern implies the only nonzeros are  $\{\partial_{\alpha_i^4}, \partial_{\gamma_t^4}, \partial_{\alpha_i^3\gamma_t}, \partial_{\alpha_i\gamma_t^3}, \partial_{\alpha_i^2\gamma_t^2}\}$. Since $\dim\beta$ is fixed, any object carrying $\beta$–indices contributes $\mathcal{O}(1)$ factors. 
		
		Although all norms on the finite dimensional spaces are equivalent, the choice of $q$ is critical for the asymptotic analysis. By selecting $q > 4$, the embedding factor $(\dim\phi)^{1/2-1/q}$ becomes sufficiently small to counteract the dimensional growth of the fixed effects parameters when bounding tensor contractions.

		\section{Regularity Conditions}\label{sec:app_RegularityConditions}
		This appendix collects the high-level conditions used in the stochastic expansions underlying our Laplace approximation.
		Conditions (i)-(viii) below are identical to those in \citet{fernandez2016individual} and are implied by Assumption~\ref{assumption:main}(i)-(v) and Assumption\ref{assumption:ape}, and we therefore do not re-verify them.
		The only additional requirement is condition (ix), which controls the contribution of the log-prior to the score and Hessian. 
		Verification of condition (ix) (and hence the prior condition in Assumption~\ref{assumption:main}(vi)) is provided in Appendix~\ref{sec:app_VerifyRegularity} of the supplemental material.
		
		\begin{assumption}[Regularity conditions for stochastic expansions]\label{assumption:regularity}
			Let $q>4$ and $0 \leq \epsilon < 1/8-1/(2q)$. Let $r_{\beta} = r_{\beta,NT} > 0$ and $r_\phi= r_{\phi,NT} > 0$ with $r_{\beta} = o\big( (NT)^{-1/(2q)-\epsilon} \big)$ and $r_{\phi} = o\big((NT)^{-\epsilon}\big)$. Define $\pi_{it} = \alpha_i + \gamma_t$. We assume that
			\begin{enumerate}[label=(\roman*)]
				\item
				$\frac{\dim \phi}{\sqrt{NT}} \rightarrow a, 0 < a < \infty$.
				
				\item
				With probability approaching one, in $\mathscr{B}(r_\beta,\beta_0) \times \mathscr B_q(r_\phi,\phi_0)$, $(\beta,\phi) \rightarrow \mathcal{L}(\beta,\phi)$ is four times continuously differentiable, $(\beta,\phi) \rightarrow \Delta(\beta,\phi)$ is three times continuously differentiable,
				and $(\beta,\phi) \rightarrow \ln p(\beta,\phi)$ is twice continuously differentiable.
				
				\item
				$\overline{\mathcal{H}}\succ 0$, $\Vert \overline{\mathcal{H}} \Vert_q = \mathcal{O}_P(1)$ and $\Vert \overline{\mathcal{H}} \Vert = \mathcal{O}_P(1)$.
				
				\item 
				$\Vert \mathcal{S} \Vert_q = \mathcal{O}_P\big( (NT)^{-\frac14+\frac{1}{2q}} \big)$, 
				$\Vert \mathcal{S} \Vert = \mathcal{O}_P(1),$
				$\Vert \partial_{\beta} \mathcal{L} \Vert = \mathcal{O}_P(1),$ 
				$\Vert \partial_{\beta\beta'} \mathcal{L} \Vert = \mathcal{O}_P(\sqrt{NT}),$ \\
				$\Vert \partial_{\beta\phi\phi} \mathcal{L} \Vert = \mathcal{O}_P\big((NT)^{\epsilon}\big),$ 
				$\big\|\sum_{g=1}^{\dim\phi} \partial_{\phi \phi^{\prime} \phi_g} \mathcal{L}[\overline{\mathcal{H}}^{-1} \mathcal{S}]_g\big\|
				=\mathcal{O}_P\big((N T)^{-1 / 4+1 /(2 q)+\epsilon}\big).$
				
				\item 
				$\Vert \widetilde{\mathcal{H}} \Vert_q = o_P(1),$ 
				$\Vert \widetilde{\mathcal{H}} \Vert = o_P\big( (NT)^{-1/8} \big),$ 
				$\Vert \partial_{\beta\beta'}\widetilde{\mathcal{L}} \Vert = \mathcal{O}_P(1),$ 
				$\Vert \partial_{\beta\phi}\widetilde{\mathcal{L}} \Vert = \mathcal{O}_P(1),$  \\
				$\Vert \partial_{\beta\phi\phi}\widetilde{\mathcal{L}} \Vert = o_P\big( (NT)^{-1/8} \big),$ 
				$ \Vert  \sum_{g,h=1}^{\dim\phi} \partial_{\phi\phi_g\phi_h} \widetilde{\mathcal{L}} [\overline{\mathcal{H}}^{-1}\mathcal{S}]_g[\overline{\mathcal{H}}^{-1}\mathcal{S}]_h\Vert  = o_P\big((NT)^{-1/4}\big)$.

				\item Uniformly over $\beta\in\mathscr B(r_\beta,\beta_0)$  and $\phi\in\mathscr B_q(r_\phi,\phi_0)$, \\
				$\Vert \partial_{\beta\phi'} \mathcal{L} (\beta,\phi) \Vert_q = \mathcal{O}_P\big((NT)^{1/(2q)}\big),$
				$\Vert \partial_{\beta\phi'} \mathcal{L} (\beta,\phi) \Vert = \mathcal{O}_P\big((NT)^{1/4}\big),$\\
				$\Vert \partial_{\beta\beta\beta}\mathcal{L}(\beta,\phi) \Vert= \mathcal{O}_P(\sqrt{NT}),$
				$\Vert \partial_{\beta\beta\phi}\mathcal{L}(\beta,\phi) \Vert= \mathcal{O}_P\big((NT)^{1/(2q)}\big),$ \\
				$\Vert \partial_{\beta\beta\phi\phi}\mathcal{L}(\beta,\phi) \Vert= \mathcal{O}_P\big((NT)^{\epsilon}\big),$ 
				$\Vert \partial_{\beta\phi\phi\phi}\mathcal{L}(\beta,\phi) \Vert_q= \mathcal{O}_P\big((NT)^{\epsilon}\big),$\\
				$\Vert \partial_{\phi\phi\phi} \mathcal{L}(\beta,\phi)  \Vert_q = \mathcal{O}_P\big((NT)^{\epsilon}\big),$  
				$\Vert \partial_{\phi\phi\phi\phi}\mathcal{L}(\beta,\phi) \Vert_q= \mathcal{O}_P\big((NT)^{\epsilon}\big).$
				
				\item $\left\Vert \partial_{\beta} \Delta\right\Vert=\mathcal{O}_P(1),$
				$\left\Vert \partial_{\phi} \Delta\right\Vert_q=\mathcal{O}_P\big( (NT)^{-1/2+1/(2q)} \big),$ \\
				$\left\Vert \partial_{\phi\phi} \Delta\right\Vert_q=\mathcal{O}_P\big( (NT)^{-1/2+\epsilon} \big),$
				$\left\Vert \partial_{\phi\phi} \Delta\right\Vert =\mathcal{O}_P\big( (NT)^{-1/2+\epsilon} \big),$
				$\left\Vert \partial_{\beta} \widetilde{\Delta} \right\Vert=o_P( 1 ),$  \\
				$\left\Vert \partial_{\phi} \widetilde{\Delta} \right\Vert=\mathcal{O}_P\big( (NT)^{-1/2} \big),$ 
				$\left\Vert \partial_{\phi\phi} \widetilde{\Delta} \right\Vert=o_P\big( (NT)^{-5/8} \big).$ 
				\item Uniformly over $\beta\in\mathscr B(r_\beta,\beta_0)$ and $\phi\in\mathscr B_q(r_\phi,\phi_0)$, \\
				$ 	\Vert \partial_{\beta\beta}\Delta(\beta,\phi) \Vert= \mathcal{O}_P\big( 1 \big),$ 
				$ 	\Vert \partial_{\beta\phi'}\Delta(\beta,\phi) \Vert_q= \mathcal{O}_P\big( (NT)^{1/(2q)-1/2} \big),$\\
				$ \Vert \partial_{\phi\phi\phi}\Delta(\beta,\phi) \Vert_q= \mathcal{O}_P\big( (NT)^{\epsilon-1/2} \big).$

				\item Uniformly over $\beta\in\mathscr B(r_\beta,\beta_0)$ and $\phi\in\mathscr B_q(r_\phi,\phi_0)$, \\
				$\big\|\partial_{\phi}\ln p(\beta,\phi)\big\| = \mathcal{O}_P\big((NT)^{1/4}\big),$
				$\big\|\partial_{\phi}\ln p(\beta,\phi)\big\|_q = \mathcal{O}_P\big((NT)^{1/(2q)}\big),$ \\
				$\big\|\partial_{\beta}\ln p(\beta,\phi)\big\| = \mathcal{O}_P\big((NT)^{1/2}\big),$
				$\big\|\partial_{\beta\beta'}\ln p(\beta,\phi)\big\| = \mathcal{O}_P\big((NT)^{1/2}\big),$\\
				$\big\|\partial_{\phi\beta'}\ln p(\beta,\phi)\big\| = \mathcal{O}_P\big((NT)^{1/4}\big),$
				$\big\|\partial_{\phi\phi'}\ln p(\beta,\phi)\big\| 	= \mathcal{O}_P(1).$
				
			\end{enumerate} 
		\end{assumption}

		Once established, Assumption~\ref{assumption:regularity} implies several useful consequences:

		\begin{lemma}\label{lemma:AfterRegularity}
			Suppose Assumption \ref{assumption:regularity} holds. 
			Then, uniformly over $\beta\in\mathscr B(r_\beta,\beta_0)$  and $\phi\in\mathscr B_q(r_\phi,\phi_0)$, 
			\begin{enumerate}[label=(\roman*)]
				\item  
				$\overline{\mathcal H}(\beta,\phi)\succ 0$ wpa 1, $\| \overline{\mathcal H}^{-1}(\beta,\phi) \|_q = \mathcal{O}_P(1)$ and $\| \overline{\mathcal H}^{-1}(\beta,\phi) \| = \mathcal{O}_P(1).$

				\item 
				$\|\overline\phi(\beta) - \phi_0   \| = o_P((NT)^{1/8})$,
				$\|\overline\phi(\beta) - \phi_0   \|_q = o_P((NT)^{-\epsilon}).$

				\item 
				$\big\|\widetilde{\mathcal H}(\beta,\overline\phi(\beta))\big\| \ =\ o_P(1),$
				$\| {\mathcal{H}}(\beta,\phi) - {\mathcal{H}}(\beta,\overline{\phi}(\beta)) \|  = o_P(1),$ \\
				$\big\|{\mathcal S}(\beta,\overline\phi(\beta))\big\| 	= \mathcal{O}_P(1) +  \mathcal{O}_P\big((NT^{1/4} \|\beta-\beta_0\|)\big).$ 
				
				\item Define  $\Sigma_{\phi\beta}:= \frac{1}{NT}\,\overline{\mathcal H}^{-1}(\partial_{\phi\beta'}\overline{\mathcal L})\overline W^{-1}$,  
				$\Sigma_{\beta\beta}:= \frac{1}{NT}\overline W^{-1}$,
				$\Sigma_{\phi\phi}:=\Sigma_{\phi\phi}^{(1)}+\Sigma_{\phi\phi}^{(2)}$ with
				$\Sigma_{\phi\phi}^{(1)}:= \frac{1}{\sqrt{NT}}\overline{\mathcal H}^{-1}$ and
				$\Sigma_{\phi\phi}^{(2)}:=  \frac{1}{{NT}} \overline{\mathcal H}^{-1}(\partial_{\phi\beta'}\overline{\mathcal L})\overline W^{-1}(\partial_{\beta\phi'}\overline{\mathcal L})\overline{\mathcal H}^{-1}$.
				Then, \\
				$\|\sum_{k,g=1}^{\dim\phi}[\partial_{\phi\phi_k\phi_g}\mathcal L(\beta,\phi) ]    (\Sigma_{\phi\phi}^{(1)})_{kg}\ \| =o_P \big((NT)^{-1/8}\big) $, \\
				$\|\sum_{k,g=1}^{\dim\phi}[\partial_{\phi\phi_k\phi_g}\mathcal L(\beta,\phi) ]   (\Sigma_{\phi\phi}^{(2)})_{kg}\ \| =o_P \big((NT)^{-1/2}\big) $,\\
				$\|\sum_{k,g=1}^{\dim\phi}[\partial_{\beta\phi_k\phi_g}\mathcal L(\beta,\phi)] (\Sigma_{\phi\phi}^{(1)})_{kg}\ \| =o_P \big((NT)^{1/8}\big) $,\\
				$\|\sum_{k,g=1}^{\dim\phi}[\partial_{\beta\phi_k\phi_g}\mathcal L(\beta,\phi)] (\Sigma_{\phi\phi}^{(2)})_{kg}\ \| =o_P \big((NT)^{-1/4}\big) $,\\ 
				$\| \sum_{k}^{\dim\beta}\sum_{g}^{\dim\phi} [\partial_{\phi\beta_k\phi_g}{\mathcal L}({\beta},{\phi})] (\Sigma_{\beta\phi})_{kg}\| =o_P \big((NT)^{-1/4}\big) $,\\ 
				$\|  \sum_{k}^{\dim\beta}\sum_{g}^{\dim\phi} [\partial_{\beta\beta_k\phi_g}{\mathcal L}({\beta},{\phi})] (\Sigma_{\beta\phi})_{kg} \| = \mathcal{O}_P\big((NT)^{-1/2}\big)$, \\
				$\|   \sum_{k,g}^{\dim\beta}[\partial_{\beta\beta_k\beta_g}{\mathcal L}({\beta},{\phi})]  (\Sigma_{\beta\beta})_{kg}  \| = \mathcal{O}_P \big((NT)^{-1/2}\big)$, \\
				$\big| \mathrm{tr}\!\big((\partial_{\phi\phi'}\overline\Delta)   \Sigma_{\phi\phi}^{(2)}   \big) \big|=o_P\big((NT)^{-1/2}\big)$.  
			\end{enumerate}
		\end{lemma}
		
		The proof of Lemma \ref{lemma:AfterRegularity} is provided in Appendix \ref{sec:app_RegularityConditions_proof} of the supplemental material.
		Part (i) upgrades positive definiteness and the norm bounds for the expected Hessian from the truth to a full neighborhood.  
		We provide  the $\ell_2$ and $\ell_q$ controls for $\overline\phi(\beta)-\phi_0$ in part (ii), which feeds directly into part~(iii).  
		
		Because the two–way panel structure renders high‐order derivative tensors sparse, parts~(iv) formalize the key tensor contractions that enter the Laplace remainder. We bound the \emph{contractions} themselves, rather than bounding tensors and vectors separately. This avoids the loose rates that arise from naive products of norms.

		\section{Proof of the main results}\label{sec:app_ProofMain}
		Before proving the main results, we collect two auxiliary lemmas used repeatedly in the calculations below. Their proofs are provided in Appendix~\ref{sec:app_intermedLemma} of the Supplemental Material.
		The first is a quadratic completion identity, and the second gives Gaussian contraction bounds.
		
		\begin{lemma}\label{lemma:pre1}
			Let $u,s\in\mathbb{R}^d$ and let $\Omega\in\mathbb{R}^{d\times d}$ be symmetric positive definite.
			Then
			\[
			s'u-\tfrac{1}{2}\,u' \Omega u
			= -\tfrac{1}{2}\,(u-\Omega^{-1}s)' \Omega (u-\Omega^{-1}s)
			+\tfrac{1}{2}\,s' \Omega^{-1} s .
			\]
		\end{lemma}
		
		\begin{lemma}\label{lemma:pre2}
			Given a positive definite covariance matrix $\Omega\in\mathbb{R}^{d\times d}$, let $\mathcal{U}\sim\mathcal{N}(0,\Omega)$ and 
			write $\langle g\rangle:=\mathbb{E}[g(\mathcal{U})]$ for expectation with respect to $\mathcal{U}$.
			Let $\kappa_m(\cdot)=\kappa(\cdot^{\times m})$ be the $m$th order cumulant function with respect to    $\mathcal{U}$.
			Let $b\in\mathbb{R}^{d}$ and let $A\in\mathbb{R}^{d\times d}$ be symmetric.  Then, there exists an absolute constant $c^{st}>0$ such that
			\begin{enumerate}[label=(\roman*)]
				\item $\big|\langle \mathcal{U}' A \mathcal{U}\rangle\big|\le d\,\|A\|\,\|\Omega\|.$ 
				\item $\big|\kappa((b'\mathcal U+\mathcal U'A\mathcal U)^{\times m})\big|    \leq c^{st} d \|\Omega \|^m\| A\|^m   + c^{st}  \|\Omega\|^{m-1}\|A\|^{m-2} \|b\|^2$, for $m\geq 2$.
			\end{enumerate}
		\end{lemma}

\begin{proof}[Proof of Lemma \ref{lemma:LaplaceLikelihood}]
    For conciseness, let $u_{\bar{\phi}}:=\phi-\overline{\phi}(\beta)$ and
    \[
    \Omega_{\bar\phi}:=\tfrac{1}{\sqrt{NT}}\,\overline{\mathcal{H}}(\beta,\overline\phi(\beta))^{-1},\quad
    \mu_{\bar\phi}:=\sqrt{NT}\Omega_{\bar\phi} \, \mathcal S(\beta,\overline\phi(\beta)),\quad
    \mathcal U_{\bar\phi}:=u_{\bar\phi}-\mu_{\bar\phi}.
    \]
    
    Consider a mean–value expansion of $\sqrt{NT}\,\mathcal{L}(\beta,\phi)+\ln p(\beta,\phi)$ around $\phi=\overline{\phi}(\beta)$:
    \begin{eqnarray}
        &&\sqrt{NT}\,\mathcal{L}(\beta,\phi)+\ln p(\beta,\phi)
        -\Big\{\sqrt{NT}\,\mathcal{L}(\beta,\overline{\phi}(\beta))+\ln p(\beta,\overline{\phi}(\beta))\Big\}  \notag \\
        &=&\sqrt{NT}\,u_{\bar{\phi}}'\,\mathcal{S}(\beta,\overline{\phi}(\beta))
        -\tfrac{\sqrt{NT}}{2}\,u_{\bar{\phi}}'\,\overline{\mathcal{H}}(\beta,\overline{\phi}(\beta))\,u_{\bar{\phi}}
        + R_{\mathcal{L}^I}(\beta)         ,          \label{eqn:prof1}
    \end{eqnarray}
    where
    $
    R_{\mathcal{L}^I}(\beta)
    =-\tfrac{\sqrt{NT}}{2}\,u_{\bar{\phi}}'\,\widetilde{\mathcal{H}}(\beta,\overline{\phi}(\beta))\,u_{\bar{\phi}}
    -\tfrac{\sqrt{NT}}{2}\,u_{\bar{\phi}}' \left( {\mathcal{H}}(\beta,\dot{\phi}(\beta)) - {\mathcal{H}}(\beta,\overline{\phi}(\beta)) \right) u_{\bar{\phi}} \\
    + u_{\bar{\phi}}' [\partial_{\phi}\ln p(\beta,\check{\phi}(\beta))],
    $
    and $\dot{\phi}(\beta)$ and $\check{\phi}(\beta)$ are intermediate points between $\phi$ and $\overline{\phi}(\beta)$.
    Completing the square via Lemma~\ref{lemma:pre1}, we rewrite equation \eqref{eqn:prof1} as
    \begin{eqnarray*}
        \tfrac{\sqrt{NT}}{2}\,\mathcal{S}(\beta,\overline{\phi}(\beta))'
        \,\overline{\mathcal{H}}(\beta,\overline{\phi}(\beta))^{-1}\,
        \mathcal{S}(\beta,\overline{\phi}(\beta))
        -\tfrac{\sqrt{NT}}{2}\,\mathcal{U}_{\bar{\phi}}'\,\overline{\mathcal{H}}(\beta,\overline{\phi}(\beta))\,\mathcal{U}_{\bar{\phi}}
        +R_{\mathcal{L}^I}(\beta) ,
    \end{eqnarray*}
    and $\mathrm{d}\phi=\mathrm{d}u_{\bar{\phi}}=\mathrm{d}\mathcal{U}_{\bar{\phi}}$ by translation.
    
    Suppose $\mathcal{U}_{\bar{\phi}}$ follows a $\dim \phi$–variate normal distribution with mean zero and covariance 
    $\Omega_{\bar\phi}$, and
    write $\langle g\rangle_{\bar\phi}:=\mathbb{E}[g(\mathcal{U}_{\bar{\phi}})]$ for expectation with respect to $\mathcal{U}_{\bar{\phi}}$.
    Recalling $e^{\sqrt{NT}\,\mathcal{L}^{I}(\beta)}
    =\int e^{\sqrt{NT}\,\mathcal{L}(\beta,\phi)+\ln p(\beta,\phi)} \mathrm{d}\phi$, and evaluating the Gaussian integral, we obtain
    $
    e^{\sqrt{NT} \mathcal{L}^I(\beta)} =   (\frac{2\pi}{\sqrt{NT}} )^{\frac{\dim\phi}{\sqrt{NT}}}  
    e^{
        \sqrt{NT}\mathcal{L} ( \beta,\overline{\phi}(\beta)  ) +\ln p (\beta, \overline{\phi}(\beta)  ) 
        - \mathcal{D}_{\mathcal{L}^I} ( \beta,\overline{\phi}(\beta)  )
    }
    \times  \langle e^{R_{\mathcal{L}^I}(\beta)}  \rangle_{\bar\phi} ,
    $
    where $\mathcal{D}_{\mathcal{L}^I}$ is defined in Lemma \ref{lemma:LaplaceLikelihood}. Taking logs and dividing by $\sqrt{NT}$ gives
    \begin{eqnarray*} 
        \mathcal{L}^I(\beta) - \mathcal{L} ( \beta,\overline{\phi}(\beta)  )  
        &= & C^{st}   +\tfrac{\ln p (\beta, \overline{\phi}(\beta)  ) - \mathcal{D}_{\mathcal{L}^I} ( \beta,\overline{\phi}(\beta)  )}{\sqrt{NT}}
        + \tfrac{\ln  \langle e^{ R_{\mathcal{L}^I}(\beta)  }  \rangle_{\bar\phi}}{\sqrt{NT}} .
    \end{eqnarray*}
    
    It remains to show that $ |\ln\langle e^{  R_{\mathcal{L}^I}(\beta)  }\rangle_{\bar{\phi}} |=o_P((NT)^{1/2})$ uniformly over $\beta\in\mathscr{B}(r_{\beta},\beta_0)$. 
    We decompose the remainder 
    $
    R_{\mathcal{L}^I}(\beta)= R^{(0)}_{\mathcal{L}^I} + R^{(1)}_{\mathcal{L}^I} +R^{(2)}_{\mathcal{L}^I},
    $
    \begin{eqnarray*}
        R^{(0)}_{\mathcal{L}^I}&=& 	-\tfrac{\sqrt{NT}}{2}\,\mu_{\bar{\phi}}'\,\widetilde{\mathcal{H}}(\beta,\overline{\phi}(\beta))\,\mu_{\bar{\phi}}
        -\tfrac{\sqrt{NT}}{2}\,\mu_{\bar{\phi}}' \left( {\mathcal{H}}(\beta,\dot{\phi}(\beta)) - {\mathcal{H}}(\beta,\overline{\phi}(\beta)) \right) \mu_{\bar{\phi}}  \notag \\
        && + \mu_{\bar{\phi}}'\,[\partial_{\phi}\ln p(\beta,\check{\phi}(\beta))], \label{eqn:R0} \\
        R^{(1)}_{\mathcal{L}^I}&=&   \mathcal{U}_{\bar{\phi}}'b , \label{eqn:R1}  \qquad
        R^{(2)}_{\mathcal{L}^I}= \mathcal{U}_{\bar{\phi}}' A \mathcal{U}_{\bar{\phi}} 
        , \label{eqn:R2}  
    \end{eqnarray*}
    with
    $
    b:=-\sqrt{NT} \widetilde{\mathcal{H}}(\beta,\overline{\phi}(\beta))\mu_{\bar{\phi}}
    -\sqrt{NT}\left( {\mathcal{H}}(\beta,\dot{\phi}(\beta)) - {\mathcal{H}}(\beta,\overline{\phi}(\beta)) \right) \mu_{\bar{\phi}} 
    + \partial_{\phi}\ln p(\beta,\check{\phi}(\beta)),
    $
    and
    $
    A:=-\tfrac{\sqrt{NT}}{2} \widetilde{\mathcal{H}} (\beta,\overline{\phi}(\beta))
    -\tfrac{\sqrt{NT}}{2} \left( {\mathcal{H}}(\beta,\dot{\phi}(\beta)) - {\mathcal{H}}(\beta,\overline{\phi}(\beta)) \right).
    $
    Here
    $R^{(0)}_{\mathcal{L}^I}$ is constant in $\mathcal U_{\bar\phi}$, 
    $R^{(1)}_{\mathcal{L}^I}$ is linear, and 
    $R^{(2)}_{\mathcal{L}^I}$ is a quadratic form.
    
    Before proceeding, we record several baseline controls for the bounds. By Assumption \ref{assumption:regularity} and Lemma \ref{assumption:regularity}, uniformly in $\beta$ and $\phi$,
    $\| \overline{\mathcal{H}}^{-1}(\beta,\overline{\phi}(\beta)) \| = \mathcal{O}_P(1)$,
    $\| \widetilde{\mathcal{H}}(\beta,\overline{\phi}(\beta)) \| = o_P(1)$,
    $\| {\mathcal{H}}(\beta,\dot{\phi}(\beta)) - {\mathcal{H}}(\beta,\overline{\phi}(\beta)) \| =o_P(1)$, and
    $\| \partial_{\phi} \ln p(\beta,\phi)  \| =\mathcal{O}_P((NT)^{1/4})$.
    Further assuming $\| \beta-\beta_0\| = \mathcal{O}_P((NT)^{-1/4})$, Lemma \ref{assumption:regularity} gives
    $\mathcal{S}(\beta,\overline{\phi}(\beta))=\mathcal{O}_P(1)$.
    Consequently,
    \begin{eqnarray}
        &&\|\Omega_{\bar\phi}\|=\mathcal{O}_P((NT)^{-1/2}),
        \qquad   
        \|\mu_{\bar\phi}\|=\mathcal{O}_P(1),
        \qquad   
        |R^{(0)}_{\mathcal{L}^I}| = o_P((NT)^{1/2}), \notag \\
        &&   
        \| b \| = o_P((NT)^{1/2}), 
        \qquad   
        \| A \| = o_P((NT)^{1/2}) .\label{eqn:boundsProofL1}
    \end{eqnarray}
    These bounds are used in cumulant calculations below.
    By a cumulant expansion, 
    $$
    \ln\langle e^{R_{\mathcal{L}^I}(\beta)}\rangle_{\bar{\phi}} = \sum_{m=1}^{\infty} \dfrac{\kappa_m (R_{\mathcal{L}^I}(\beta) )}{m!},
    $$
    where $\kappa_m(\cdot)=\kappa(\cdot^{\times m})$ denotes the $m$th order cumulant.
    The first cumulant $\kappa_1$ equals the mean. 	
    In particular, 
    $\langle R^{(1)}_{\mathcal{L}^I}\rangle_{\bar\phi} =0$ by oddness.
    Thus,
    $
    \kappa_1 (R^I_{\mathcal L}(\beta) )=\langle R^I_{\mathcal L}(\beta)\rangle_{\bar\phi}
    = R^{(0)}_{\mathcal{L}^I} +\langle R^{(2)}_{\mathcal{L}^I}  \rangle_{\bar\phi}.
    $
    By Lemma \ref{lemma:pre2}(i),
    $\langle R^{(2)}_{\mathcal{L}^I}  \rangle_{\bar\phi} \le\mathcal{O}(N+T)\|A\| \|\Omega_{\bar\phi}\|   = o_P((NT)^{1/2})$. Therefore,  $|\kappa_1 (R^I_{\mathcal L}(\beta) )|= o_P((NT)^{1/2})  $.

    For $m\ge 2$, the constant term  $R^{(0)}_{\mathcal{L}^I}$ does not enter the cumulant.  By Lemma \ref{lemma:pre2}(ii) and together with the bounds provided in \eqref{eqn:boundsProofL1}, 
    $$
    \kappa_m (R_{\mathcal{L}^I}(\beta)) =\kappa((b'\mathcal{U}_{\bar{\phi}} +  \mathcal{U}_{\bar{\phi}}' A \mathcal{U}_{\bar{\phi}} )^{\times m}) =  o_P((NT)^{1/2}),
    \qquad
    m\ge 2.
    $$
    Combining the foregoing results we have
    $
    \big|  \ln\langle e^{R^I_{\mathcal{L}}(\beta)}\rangle_{\bar{\phi}} \big| =o_P((NT)^{1/2})
    $.
    The conclusion follows, as stated in Lemma~\ref{lemma:LaplaceLikelihood}.
\end{proof}

		We next develop a target–centered Laplace approximation for posterior means: 
		Theorem~\ref{thm:asy_expansions} (posterior means of $(\beta,\phi)$) and Theorem~\ref{thm:APE_expansions}(i) (posterior mean of the APE) follow from the similar argument. We prove Theorem~\ref{thm:asy_expansions} first.

		\begin{proof}[Proof of Theorem \ref{thm:asy_expansions}]
			\quad\\
			\underline{\textit{Step 1) Gaussian completion and posterior means.}} 
			
			Let $u_\beta =\beta-\beta_0$, $u_\phi =\phi-\phi_0$,
			\(
			\theta=(\beta',\phi')',\ 
			u_\theta=(u_\beta',u_\phi')',\ 
			\partial_\theta\mathcal L=(\partial_{\beta'}\mathcal L,\ \mathcal S')'
			\)
			and
			\(
			\Omega_{\theta\theta}
			:=
			-\begin{pmatrix}
				\partial_{\beta\beta'}\overline{\mathcal L} &\quad  \partial_{\beta\phi'}\overline{\mathcal L}\\
				\partial_{\phi\beta'}\overline{\mathcal L} &\quad  -\overline{\mathcal H}
			\end{pmatrix}.
			\)
			A mean-value expansion of $\sqrt{NT}\,\mathcal L(\beta,\phi)+\ln p(\beta,\phi)$ around $(\beta_0,\phi_0)$ yields
			\begin{eqnarray*}
				&&	\sqrt{NT}\,\mathcal{L}(\beta,\phi)+\ln p(\beta,\phi) \\
				&=&  \sqrt{NT}\,\mathcal{L}+\ln p+ \sqrt{NT}u_{\theta}' \partial_{\theta}\mathcal{L}  -\tfrac{\sqrt{NT}}{2} u_{\theta}' \Omega_{\theta\theta} u_{\theta}  + R_{\mathcal{L}}(\beta,\phi) \\
				&=&  \sqrt{NT}\,\mathcal{L}+\ln p+     \tfrac{\sqrt{NT}}{2}   (\partial_{\theta'}\mathcal{L}) \Omega_{\theta\theta}^{-1} (\partial_{\theta}\mathcal{L})    - \tfrac{\sqrt{NT}}{2}\mathcal{U}_{\theta}' \Omega_{\theta\theta} \mathcal{U}_{\theta} + R_{\mathcal{L}}(\beta,\phi) 
			\end{eqnarray*}
			where the second equality completes the square (Lemma \ref{lemma:pre1}) with
			$\mathcal U_\theta:=u_\theta-\mu_\theta$, $\mu_\theta:=\Omega^{-1}\partial_\theta\mathcal{L}$,
			and the explicit remainder $R_{\mathcal L}(\theta)$ is given in Appendix~\ref{app:remainders} of the Supplemental Material (see equations \eqref{eqn:Rtheta}--\eqref{eqn:R_quardic}).
			
			Let $\Sigma_{\theta\theta}:=\Omega_{\theta\theta}^{-1}/\sqrt{NT}$.
			Using the $2\times 2$ block inversion formula,
			\begin{footnotesize}
				\[
				\Sigma_{\theta\theta}
				=
				\begin{pmatrix}
					\frac{1}{NT}\overline W^{-1}=:\Sigma_{\beta\beta}
					&\mathrm{\qquad}
					\frac{1}{NT}\,\overline W^{-1}(\partial_{\beta\phi'}\overline{\mathcal L})\overline{\mathcal H}^{-1}=:\Sigma_{\beta\phi}
					\\
					\frac{1}{NT}\,\overline{\mathcal H}^{-1}(\partial_{\phi\beta'}\overline{\mathcal L})\overline W^{-1}=:\Sigma_{\phi\beta}
					&\mathrm{\qquad}
					\frac{1}{\sqrt{NT}}\overline{\mathcal H}^{-1}
					+\frac{1}{NT}\,\overline{\mathcal H}^{-1}(\partial_{\phi\beta'}\overline{\mathcal L})\overline W^{-1}(\partial_{\beta\phi'}\overline{\mathcal L})\overline{\mathcal H}^{-1}=:\Sigma_{\phi\phi}
				\end{pmatrix},
				\]
			\end{footnotesize}
			where $\overline W$ is defined in Theorem \ref{thm:asy_expansions}.
			We then have 
			$$\mu_{\beta} = \tfrac{1}{\sqrt{NT}} \overline{W}^{-1}U^0, \qquad
			\mu_{\phi} =  \overline{\mathcal{H}}^{-1}\left(  \mathcal{S}   + (\partial_{\phi\beta'}\overline{\mathcal{L}} ) \mu_{\beta}     \right),  $$
			where $U^0$ in defined in Theorem \ref{thm:asy_expansions}. 
			Applying Assumption \ref{assumption:regularity} and norm inequalities, it is straightforward to obtain
			$\|  \mu_{\beta} \| =  \mathcal{O}_P((NT)^{-1/2})$ ,
			$\|\mu_{\phi} \| = \mathcal{O}_P(1)$,
			$\|\Sigma_{\beta\beta} \| = \mathcal{O}_P((NT)^{-1})$,
			$\| \Sigma_{\beta\phi} \| =  \mathcal{O}_P((NT)^{-3/4})$, and
			$\| \Sigma_{\phi\phi} \| =  \mathcal{O}_P((NT)^{-1/2})$.

			Treating $\mathcal U_\theta\sim\mathcal N(0,\Sigma_{\theta\theta})$ and writing
			\(
			\langle g\rangle_\theta:=\mathbb E[g(\mathcal U_\theta)],
			\)
			the Gaussian integral gives the normalization factor $\mathcal Z$ and the posterior-means\footnote{By the Gaussian integral,
				$
				\int e^{\sqrt{NT}\mathcal{L}(\beta,\phi)+\ln p(\beta,\phi)} \mathrm{d}\theta 
				\;=\; \mathcal{Z}\cdot \langle e^{R_{\mathcal{L}}(\theta)} \rangle_{\theta},
				$
				with the normalizing factor 
				\(
				\mathcal{Z} =(2\pi)^{(\dim\beta+\dim\phi)/2}(\det\Sigma_{\theta\theta})^{-1/2}
				\exp\!\big\{\sqrt{NT}\mathcal L+\ln p+\frac{1}{2}
				(\partial_{\theta'}\mathcal L)  \Sigma_{\theta\theta}   (\partial_\theta\mathcal L)\big\}.
				\)
				Moreover, using $\theta=\theta_0+\mu_\theta + \mathcal{U}_\theta $,  we have 
				$
				\widehat\theta^E
				=\frac{\langle (\theta_0+\mu_\theta+\mathcal U_\theta)e^{R_{\mathcal L}(\theta)}\rangle_\theta}{\langle e^{R_{\mathcal L}(\theta)}\rangle_\theta}
				=\theta_0+\mu_\theta+\frac{\langle \mathcal U_\theta e^{R_{\mathcal L}(\theta)}\rangle_\theta}{\langle e^{R_{\mathcal L}(\theta)}\rangle_\theta}
				$.
			}
			\begin{equation}
				\widehat\beta^E=\beta_0+\mu_\beta+\frac{\langle \mathcal U_\beta e^{R_{\mathcal L}(\beta,\phi)}\rangle_\theta}{\langle e^{R_{\mathcal L}(\beta,\phi)}\rangle_\theta},
				\qquad
				\widehat\phi^E=\phi_0+\mu_\phi+\frac{\langle \mathcal U_\phi e^{R_{\mathcal L}(\beta,\phi)}\rangle_\theta}{\langle e^{R_{\mathcal L}(\beta,\phi)}\rangle_\theta}.\label{eq:post-mean-ratio}
			\end{equation}

			\underline{\textit{Step 2) Cumulant expansions.}}
			
			For any measurable $f(\mathcal U_\theta)$, the following identity holds: 
			\begin{equation}\label{eq:cumulant-ratio}
				\frac{\langle f(\mathcal U_\theta)e^{R_{\mathcal L}(\beta,\phi)}\rangle_\theta}{\langle e^{R_{\mathcal L}(\beta,\phi)}\rangle_\theta}
				=
				\sum_{m=0}^\infty\frac{1}{m!}\kappa\big(f(\mathcal U_\theta),R_{\mathcal L}^{\times m}(\beta,\phi)\big).
			\end{equation}
			where $\kappa(\cdot)$ denotes the joint cumulant. 
			We apply \eqref{eq:cumulant-ratio} with $f(\mathcal U_\theta)=\mathcal U_\beta$ and $f(\mathcal U_\theta)=\mathcal U_\phi$. 
			Note that $\kappa(U_\bullet)=\langle\mathcal U_\bullet\rangle_\theta=0$, the $m=0$ term is zero, we therefore have
			$
			\frac{\big\langle \mathcal U_\bullet,e^{R_{\mathcal L}}\big\rangle_\theta}{\big\langle            e^{R_{\mathcal L}}\big\rangle_\theta}
			= \sum_{m=1}^{\infty}\frac{1}{m!}\kappa\big(\mathcal U_\bullet,R_{\mathcal L}^{\times m}\big).
			$
			
			For $m=1$, the first cumulant $\kappa(\mathcal U_\bullet,R_{\mathcal L}(\beta,\phi)) = \mathrm{Cov}(\mathcal U_\bullet,R_{\mathcal L}(\beta,\phi)) = \langle \mathcal U_\bullet R_{\mathcal L}(\beta,\phi) \rangle_\theta   $.
			Write $R_{\mathcal L}(\beta,\phi)=R^{(0)}+R^{(1)}+R^{(2)}+R^{(3)}+R^{(4)}$, where $R^{(j)}$ collects the terms that contain precisely $j$ factors of $(\mathcal U_\beta,\mathcal U_\phi)$ (see \eqref{eqn:Rtheta_decompose} in Appendix~\ref{app:remainders}).
			By oddness, $\kappa(\mathcal U_\bullet,R^{(0)})=\kappa(\mathcal U_\bullet,R^{(2)})=\kappa(\mathcal U_\bullet,R^{(4)})=0$. Hence  the only contributions to the first cumulant  come from $R^{(1)}$ (\textit{linear} part) and $R^{(3)}$ (\textit{cubic} part).
			
			\underline{\emph{2a) First cumulant: linear part $R^{(1)}$.}}
			From \eqref{eqn:Rtheta_decompose}, the linear terms are exactly those listed in \eqref{eqn:R_linear}.
			Contracting each of them with $\mathcal U_\phi$ and using the Gaussian covariance blocks
			$\Sigma_{\beta\beta},\Sigma_{\beta\phi},\Sigma_{\phi\phi}$, we obtain
			\begin{eqnarray*}
				&&\kappa(\mathcal U_\phi,R^{(1)})
				=
				\sqrt{NT}\,\Sigma_{\phi\phi}(-\widetilde{\mathcal H})\mu_\phi
				+\sqrt{NT}\,\Sigma_{\phi\beta}(\partial_{\beta\phi'}\widetilde{\mathcal L})\mu_\phi
				+\Sigma_{\phi\phi}\,\partial_\phi\ln p
				+\Sigma_{\phi\beta}\,\partial_\beta\ln p \notag\\
				&& \quad
				+\sqrt{NT}\,\Sigma_{\phi\phi}\!\sum_{k,g} (\partial_{\phi\phi_k\phi_g}\overline{\mathcal L})\,\mu_{\phi,k}\mu_{\phi,g}
				+\sqrt{NT}\,\Sigma_{\phi\beta}\!\sum_{k,g} (\partial_{\beta\phi_k\phi_g}\overline{\mathcal L})\,\mu_{\phi,k}\mu_{\phi,g} 
				+\mathsf r_\phi^{(1)}, \label{eq:K1phi}
			\end{eqnarray*}
			where, by Assumption~\ref{assumption:regularity} and Lemma~\ref{lemma:AfterRegularity},
			$\|\mathsf r^{(1)}_\phi\|=o_P((NT)^{-1/4})$.\footnote{We show three representative bounds in $ \mathsf{r}^{(1)}_\phi$. By Assumption \ref{assumption:regularity} and Lemma \ref{lemma:AfterRegularity}:
				
				(i) $\|\kappa(\mathcal U_\phi,\sqrt{NT}\mathcal U_\phi'(\partial_{\beta\phi'}\widetilde{\mathcal L})\mu_\beta)\|
				\le \sqrt{NT}\|\Sigma_{\phi\phi}\|\|\partial_{\beta\phi'}\widetilde{\mathcal L}\| \|\mu_\beta\| \\
				\mathrm{\qquad}
				= \sqrt{NT} \mathcal O_P((NT)^{-1/2}) \mathcal O_P(1) \mathcal O_P((NT)^{-1/2})
				=\mathcal O_P((NT)^{-1/2})$.
				
				(ii) $\|\kappa(\mathcal U_\phi,\sqrt{NT} \partial_{\beta\phi^2}\widetilde{\mathcal L}[\mathcal U_\beta,\mu_\phi^{\times2}])\| \\
				\mathrm{\qquad}
				\le \sqrt{NT}\|\Sigma_{\beta\phi}\| \|\partial_{\beta\phi\phi}\widetilde{\mathcal L}\| \|\mu_\phi\|^2
				=\sqrt{NT} \mathcal O_P((NT)^{-3/4})  o_P((NT)^{-1/8})  \mathcal O_P(1)
				=o_P((NT)^{-1/4})$.
				
				(iii) $\|\kappa(\mathcal U_\phi,\sqrt{NT} \partial_{\beta\phi^2}\overline{\mathcal L}[\mu_\beta,\mathcal U_\phi,\mu_\phi])\| 
				\le \sqrt{NT}\|\Sigma_{\phi\phi}\| \|\partial_{\beta\phi\phi}\overline{\mathcal L}\| \|\mu_\beta\| \|\mu_\phi\| \\
				\mathrm{\qquad}
				=\sqrt{NT} \mathcal O_P((NT)^{-1/2}) \mathcal O_P((NT)^\epsilon) \mathcal O_P((NT)^{-1/2}) \mathcal O_P(1)
				=\mathcal O_P((NT)^{-1/2+\epsilon})=o_P((NT)^{-1/4})$.}
			Further, we decompose $\mu_\phi=\mu_\phi^{(1)}+\mu_\phi^{(2)}$ with
			$\mu_\phi^{(1)}=\overline{\mathcal H}^{-1}\mathcal S$ and
			$\mu_\phi^{(2)}=\overline{\mathcal H}^{-1}(\partial_{\phi\beta'}\overline{\mathcal L})\mu_\beta$,
			where $\|\mu_\phi^{(2)}\|=\mathcal O_P((NT)^{-1/4})$. Thus
			\begin{eqnarray*}
				&&	\kappa(\mathcal U_\phi,R^{(1)})
				=
				\sqrt{NT}\,\Sigma_{\phi\phi}(-\widetilde{\mathcal H})\mu_\phi^{(1)}
				+\sqrt{NT}\,\Sigma_{\phi\beta}(\partial_{\beta\phi'}\widetilde{\mathcal L})\mu_\phi^{(1)}
				+\Sigma_{\phi\phi}\,\partial_\phi\ln p
				+\Sigma_{\phi\beta}\,\partial_\beta\ln p \notag\\
				&& \quad
				+\tfrac{\sqrt{NT}}{2}\Sigma_{\phi\phi}\!\sum_{k,g} (\partial_{\phi\phi_k\phi_g}\overline{\mathcal L})\,\mu_{\phi,k}^{(1)}\mu_{\phi,g}^{(1)}
				+\tfrac{\sqrt{NT}}{2}\Sigma_{\phi\beta}\!\sum_{k,g} (\partial_{\beta\phi_k\phi_g}\overline{\mathcal L})\,\mu_{\phi,k}^{(1)}\mu_{\phi,g}^{(1)}
				+ \widetilde{\mathsf r}^{(1)}_\phi,\label{eq:K1phi-final}
			\end{eqnarray*}
			with $\| \widetilde{\mathsf r}_\phi\|=o_P((NT)^{-1/4})$.
			A completely analogous contraction gives
			\begin{eqnarray*}
				&&	\kappa(\mathcal U_\beta,R^{(1)})
				=
				\sqrt{NT}\,\Sigma_{\beta\phi}(-\widetilde{\mathcal H})\mu_\phi^{(1)}
				+\sqrt{NT}\,\Sigma_{\beta\beta}(\partial_{\beta\phi'}\widetilde{\mathcal L})\mu_\phi^{(1)}
				+\Sigma_{\beta\phi}\,\partial_\phi\ln p
				+\Sigma_{\beta\beta}\,\partial_\beta\ln p \notag\\
				&&\quad
				+\tfrac{\sqrt{NT}}{2}\Sigma_{\beta\phi}\!\sum_{k,g} (\partial_{\phi\phi_k\phi_g}\overline{\mathcal L})\,\mu_{\phi,k}^{(1)}\mu_{\phi,g}^{(1)}
				+\tfrac{\sqrt{NT}}{2}\Sigma_{\beta\beta}\!\sum_{k,g} (\partial_{\beta\phi_k\phi_g}\overline{\mathcal L})\,\mu_{\phi,k}^{(1)}\mu_{\phi,g}^{(1)} 
				+ \widetilde{\mathsf r}_\beta^{(1)},\label{eq:K1beta-final}
			\end{eqnarray*}
			with $\| \widetilde{\mathsf r}_\beta^{(1)}\|=o_P((NT)^{-1/2})$.

			\underline{\emph{2b) First cumulant: cubic part $R^{(3)}$.}}
			$R^{(3)}$ consists of both the purely cubic terms in $(\mathcal U_\beta,\mathcal U_\phi)$ and the mixed cubic terms involving one deterministic mean, and its explicit form is given in \eqref{eqn:R_cubic} in the Supplemental Material. 
			By Wick/Isserlis formulae, we have
			$
			\langle  \mathcal{U}_\bullet \mathcal{U}_\phi'  \sum_{k,g}(\partial_{\phi\phi_k\phi_g}\overline{\mathcal L}) \mathcal{U}_{\phi_k} \mathcal{U}_{\phi_g}   \rangle_\theta
			= 3 \Sigma_{\bullet\phi}\sum_{k,g}(\partial_{\phi\phi_k\phi_g}\overline{\mathcal L})\,(\Sigma_{\phi\phi})_{kg},
			$
			and
			$$
			\langle  \mathcal{U}_\bullet \mathcal{U}_\beta'  \sum_{k,g}(\partial_{\beta\phi_k\phi_g}\overline{\mathcal L}) \mathcal{U}_{\phi_k} \mathcal{U}_{\phi_g}   \rangle_\theta
			=  \Sigma_{\bullet\beta}\sum_{k,g}(\partial_{\beta\phi_k\phi_g}\overline{\mathcal L})\,(\Sigma_{\phi\phi})_{kg} 
			+  \Sigma_{\bullet\phi}\sum_{k,g}(\partial_{\phi\beta_k\phi_g}\overline{\mathcal L})\,(\Sigma_{\beta\phi})_{kg}.
			$$
			By Assumption \ref{assumption:regularity}(iv),  the second term on the RHS vanishes asymptotically, i.e.,
			$
			\Big\|\Sigma_{\bullet\phi}\sum_{k,g}(\partial_{\phi\beta_k\phi_g}\overline{\mathcal L})\,(\Sigma_{\beta\phi})_{kg} \Big\| \leq \|\Sigma_{\bullet\phi}\| o_P((NT)^{-1/4}).
			$
			Therefore
			\begin{align}
				\kappa(\mathcal U_\phi,R^{(3)})
				&=
				\tfrac{\sqrt{NT}}{2}\,\Sigma_{\phi\phi}\sum_{k,g}(\partial_{\phi\phi_k\phi_g}\overline{\mathcal L})\,(\Sigma_{\phi\phi})_{kg}
				+\tfrac{\sqrt{NT}}{2}\,\Sigma_{\phi\beta}\sum_{k,g}(\partial_{\beta\phi_k\phi_g}\overline{\mathcal L})\,(\Sigma_{\phi\phi})_{kg}
				+\mathsf r^{(3)}_\phi, \notag
			\end{align}
			where $\mathsf r^{(3)}_\phi$ collects all other terms, and Assumption \ref{assumption:regularity} is sufficient to ensure   $\| \mathsf r^{(3)}_\phi \|=o_P((NT)^{-1/4})$. 
			Further, we  decompose  $\Sigma_{\phi\phi}=\Sigma_{\phi\phi}^{(1)}+\Sigma_{\phi\phi}^{(2)}$ with
			$\Sigma_{\phi\phi}^{(1)}=\overline{\mathcal H}^{-1}/\sqrt{NT}$ and
			$\Sigma_{\phi\phi}^{(2)}=\overline{\mathcal H}^{-1}(\partial_{\phi\beta'}\overline{\mathcal L})\overline W^{-1}(\partial_{\beta\phi'}\overline{\mathcal L})\overline{\mathcal H}^{-1}/(NT)$.
			Lemma~\ref{lemma:AfterRegularity}(iv)  implies that the contribution of $\Sigma_{\phi\phi}^{(2)}$  can be absorbed into the remainder. Thus
			\begin{align}
				\kappa(\mathcal U_\phi,R^{(3)})
				&=
				\tfrac{\sqrt{NT}}{2}\,\Sigma_{\phi\phi}\sum_{k,g}(\partial_{\phi\phi_k\phi_g}\overline{\mathcal L})\,(\Sigma_{\phi\phi}^{(1)})_{kg}
				+\tfrac{\sqrt{NT}}{2}\,\Sigma_{\phi\beta}\sum_{k,g}(\partial_{\beta\phi_k\phi_g}\overline{\mathcal L})\,(\Sigma_{\phi\phi}^{(1)})_{kg}
				+\widetilde{\mathsf r}^{(3)}_\phi, \notag
			\end{align}
			with $\|\widetilde{\mathsf r}^{(3)}_\phi\|=o_P((NT)^{-1/4})$.
			An identical contraction yields
			\begin{align}
				\kappa(\mathcal U_\beta,R^{(3)})
				&=
				\tfrac{\sqrt{NT}}{2}\,\Sigma_{\beta\phi}\sum_{k,g}(\partial_{\phi\phi_k\phi_g}\overline{\mathcal L})\,(\Sigma_{\phi\phi}^{(1)})_{kg}
				+\tfrac{\sqrt{NT}}{2}\,\Sigma_{\beta\beta}\sum_{k,g}(\partial_{\beta\phi_k\phi_g}\overline{\mathcal L})\,(\Sigma_{\phi\phi}^{(1)})_{kg}
				+\widetilde{\mathsf r}^{(3)}_\beta,\notag
			\end{align}
			with $\|\widetilde{\mathsf r}^{(3)}_\beta\|=o_P((NT)^{-1/2})$.
			
			\underline{\emph{2c) First cumulant: final collection.}}
			By Lemma \ref{lemma:LaplaceLikelihood}, the leading bias in the integrated log–likelihood
			$\mathcal L^I(\beta)$ is given by $\mathcal D_{\mathcal L^I}(\beta,\phi)$.
			Evaluated at the truth $(\beta_0,\phi_0)$, its gradient admits the representation
			\begin{align}
				&&\scalebox{0.94}{$\partial_\beta \mathcal D_{\mathcal L^I}
					= \sqrt{NT}\partial_{\beta\phi'}\widetilde{\mathcal L}\ \overline{\mathcal H}^{-1}\mathcal S
					+\tfrac{\sqrt{NT}}{2}\sum_{k,g}\partial_{\beta\phi_k\phi_g}\overline{\mathcal L}
					\Big( (\overline{\mathcal H}^{-1}\mathcal S)_k(\overline{\mathcal H}^{-1}\mathcal S)_g
					+\tfrac{(\overline{\mathcal H}^{-1})_{kg}}{\sqrt{NT}}\Big),$} \label{eq:biasbeta} \\
				&&\scalebox{0.94}{$\partial_\phi \mathcal D_{\mathcal L^I}
					= -\sqrt{NT}\widetilde{\mathcal H}\,\overline{\mathcal H}^{-1}\mathcal S
					+\tfrac{\sqrt{NT}}{2}\sum_{k,g}\partial_{\phi\phi_k\phi_g}\overline{\mathcal L}
					\Big( (\overline{\mathcal H}^{-1}\mathcal S)_k(\overline{\mathcal H}^{-1}\mathcal S)_g
					+\tfrac{(\overline{\mathcal H}^{-1})_{kg}}{\sqrt{NT}}\Big),$} \label{eq:biasphi}
			\end{align}
			where we use the fact that $\mathcal{S}=\mathcal{S}-\overline{\mathcal{S}}=\widetilde{\mathcal{S}}$.
			Collecting the bounds from the foregoing steps and substituting
			\(
			\Sigma_{\beta\beta},
			\Sigma_{\beta\phi},
			\Sigma_{\phi\phi}^{(1)},
			\)
			together with $\mu_\beta$ and
			$\mu_\phi^{(1)}$,
			we therefore obtain
			\begin{eqnarray*}
				&& \kappa(\mathcal U_\beta,R_{\mathcal L}(\beta,\phi))=\overline W^{-1}\Big(\mathcal B^\beta+(\partial_{\beta\phi'}\overline{\mathcal L})\overline{\mathcal H}^{-1}\mathcal B^\phi\Big)/\sqrt{NT} + {\mathsf r}^{(1)}_\beta + \widetilde{\mathsf r}^{(3)}_\beta  \\
				&& \kappa(\mathcal U_\phi,R_{\mathcal L}(\beta,\phi))=\overline{\mathcal{H}}^{-1}  \Big(
				\mathcal{B}^{\phi}
				+ \tfrac{1}{\sqrt{NT}}[\partial_{\phi\beta'}\overline{\mathcal{L}}]\overline{W}^{-1}  \big( \mathcal{B}^{\beta} 
				+ [\partial_{\beta\phi'}\overline{\mathcal{L}}]\overline{\mathcal{H}}^{-1} \mathcal{B}^{\phi}\big)
				\Big)
				+ {\mathsf r}^{(1)}_\phi + \widetilde{\mathsf r}^{(3)}_\phi    .
			\end{eqnarray*}

			\underline{\emph{2d) Higher cumulants.}} 
			By Lemma \ref{lemma:remainders}(ii), for all $m\geq 2$, the higher cumulants satisfy
			$\|\kappa(\mathcal U_\beta,R_{\mathcal L}^{\times m}(\beta,\phi))\| = o_P\big((NT)^{-1/2}\big)$ and 
			$\|\kappa(\mathcal U_\phi,R_{\mathcal L}^{\times m}(\beta,\phi))\|  = o_P\big((NT)^{-1/4}\big)$. Thus, the proof is complete.
		\end{proof}

		\begin{proof}[Proof of Theorem \ref{thm:APE_expansions}]
			\quad\\
			\underline{\#\textit{Part (i): posterior mean of the APE.  }} \\
			Recall the notation introduced in the proof of Theorem \ref{thm:asy_expansions}.
			A second–order expansion of $\Delta(\beta,\phi)$ at $(\beta_0,\phi_0)$ gives
			\begin{align*}
				\Delta(\beta,\phi)
				= \Delta +(\partial_{\beta'}\overline\Delta) u_\beta +(\partial_{\phi'}\overline\Delta) u_\phi
				+\tfrac12 \mu_\phi'(\partial_{\phi\phi'}\overline\Delta)\mu_\phi
				+\tfrac12 \mathcal U_\phi'(\partial_{\phi\phi'}\overline\Delta)\mathcal U_\phi
				+R_\Delta(\beta,\phi),
			\end{align*}
			where the remainder $R_\Delta(\beta,\phi)$ is reported in \eqref{eqn:RDelta_decompose} in Appendix \ref{app:remainders}.
			For the posterior mean of APE $\widehat\Delta^E  =\frac{\langle \Delta(\beta,\phi)e^{R_{\mathcal L}(\beta,\phi)}\rangle_\theta}{\langle e^{R_{\mathcal L}(\beta,\phi)}\rangle_\theta} $,
			taking posterior expectations via the Gaussian completion used in Step 1 of the proof of Theorem \ref{thm:asy_expansions} yields
			\begin{align}
				\widehat\Delta^E
				&= \Delta 
				+(\partial_{\beta'}\overline\Delta)\!\left(\mu_\beta+\frac{\langle \mathcal U_\beta e^{R_{\mathcal L}(\beta,\phi)}\rangle_\theta}{\langle e^{R_{\mathcal L}(\beta,\phi)}\rangle_\theta}\right)
				+(\partial_{\phi'}\overline\Delta)\!\left(\mu_\phi+\frac{\langle \mathcal U_\phi e^{R_{\mathcal L}(\beta,\phi)}\rangle_\theta}{\langle e^{R_{\mathcal L}(\beta,\phi)}\rangle_\theta}\right) \notag\\
				&\quad
				+\tfrac12\,\mu_\phi'(\partial_{\phi\phi'}\overline\Delta)\mu_\phi
				+\tfrac12\,\frac{\langle \mathcal U_\phi'(\partial_{\phi\phi'}\overline\Delta)\mathcal U_\phi e^{R_{\mathcal L}(\beta,\phi)}\rangle_\theta}{\langle e^{R_{\mathcal L}(\beta,\phi)}\rangle_\theta}
				+\frac{\langle R_\Delta(\theta) e^{R_{\mathcal L}(\beta,\phi)}\rangle_\theta}{\langle e^{R_{\mathcal L}(\beta,\phi)}\rangle_\theta}.
				\label{eq:APE-mean-master}
			\end{align}
			The two linear terms in \eqref{eq:APE-mean-master} equal
			$
			(\partial_{\beta'}\overline\Delta)\,(\widehat\beta^E-\beta_0)
			+
			(\partial_{\phi'}\overline\Delta)\,(\widehat\phi^E-\phi_0),
			$
			exactly by \eqref{eq:post-mean-ratio}, i.e.,
			$\widehat\beta^E-\beta_0=\mu_\beta+\frac{\langle \mathcal U_\beta e^{R_{\mathcal L}(\beta,\phi)}\rangle}{\langle e^{R_{\mathcal L}(\beta,\phi)}\rangle}$ and the analogous identity for $\phi$.
			
			Therefore, it remains to analyze the last three terms in \eqref{eq:APE-mean-master}.

			\underline{\textit{1) Quadratic  term $\mathcal U_\phi'(\partial_{\phi\phi'}\overline\Delta)\mathcal U_\phi$.  }} 
			Applying the connected cumulant expansion \eqref{eq:cumulant-ratio} to 
			$f(\mathcal U_\theta)=\mathcal U_\phi'(\partial_{\phi\phi'}\overline\Delta)\mathcal U_\phi$ gives
			\[
			\frac{\langle \mathcal U_\phi'(\partial_{\phi\phi'}\overline\Delta)\mathcal U_\phi\,e^{R_{\mathcal L}(\beta,\phi)}\rangle_\theta}{\langle e^{R_{\mathcal L}(\beta,\phi)}\rangle_\theta}
			=\langle  \mathcal U_\phi'(\partial_{\phi\phi'}\overline\Delta)\mathcal U_\phi  \rangle_\theta
			+\sum_{m\ge1}\frac{\kappa\left(\mathcal U_\phi'(\partial_{\phi\phi'}\overline\Delta)\mathcal U_\phi,R_{\mathcal L}^{\times m}(\beta,\phi)\right)}{m!}.
			\]
			
			The second term ($m\ge1$ terms) is $o_P\big((NT)^{-1/2}\big)$ as provided by Lemma \ref{lemma:remainders}(ii). 
			For the first term, by Gaussianity,
			$
			\langle \mathcal U_\phi'(\partial_{\phi\phi'}\overline\Delta)\mathcal U_\phi\rangle_\theta
			=\mathrm{tr}\!\big((\partial_{\phi\phi'}\overline\Delta)\Sigma_{\phi\phi}\big)$.
			Decomposing $\Sigma_{\phi\phi}=\Sigma_{\phi\phi}^{(1)}+\Sigma_{\phi\phi}^{(2)}$,
			Lemma \ref{lemma:AfterRegularity}(iv) gives
			$\mathrm{tr}\!\big((\partial_{\phi\phi'}\overline\Delta)\Sigma_{\phi\phi}^{(2)}\big)=o_P\big((NT)^{-1/2}\big)$. Therefore 
			\begin{equation*}
				\frac{\langle \mathcal U_\phi'(\partial_{\phi\phi'}\overline\Delta)\mathcal U_\phi\,e^{R_{\mathcal L}(\beta,\phi)}\rangle_\theta}{\langle e^{R_{\mathcal L}}\rangle_\theta}
				=\tfrac{1}{\sqrt{NT}}\;\mathrm{tr}\!\big((\partial_{\phi\phi'}\overline\Delta)\,\overline{\mathcal H}^{-1}\big)
				+o_P\big((NT)^{-1/2}\big).
				\label{eq:UphiDUphi-ratio}
			\end{equation*}

			\underline{\textit{2) Deterministic term $\tfrac12\,\mu_\phi'(\partial_{\phi\phi'}\overline\Delta)\mu_\phi$.  }} 
			Decompose $\mu_\phi=\mu_\phi^{(1)}+\mu_\phi^{(2)}$ with $\mu_\phi^{(1)}=\overline{\mathcal H}^{-1}\mathcal S$ and $\mu_\phi^{(2)}=\overline{\mathcal H}^{-1}(\partial_{\phi\beta'}\overline{\mathcal L})\mu_\beta$.
			Under Assumption \ref{assumption:regularity}, $\|\mu_\phi^{(1)}\| = \mathcal O_P(1)$, $\|\mu_\phi^{(2)}\| = \mathcal O_P((NT)^{-1/4})$, and $\|\partial_{\phi\phi'}\overline\Delta\| = \mathcal{O}_P((NT)^{\epsilon-1/2})$. Consequently, the quadratic form involving $\mu_\phi^{(2)}$ and the cross-term are asymptotically negligible:
			\[
			\|\mu_\phi^{(2)\prime}(\partial_{\phi\phi'}\overline\Delta)\mu_\phi^{(2)}\| = o_P\big((NT)^{-1/2}\big), \quad \|\mu_\phi^{(1)\prime}(\partial_{\phi\phi'}\overline\Delta)\mu_\phi^{(2)}\| = o_P\big((NT)^{-1/2}\big).
			\]
			Therefore,
			$
			\tfrac12\,\mu_\phi'(\partial_{\phi\phi'}\overline\Delta)\mu_\phi
			= \tfrac12\,\mathcal S'\overline{\mathcal H}^{-1}(\partial_{\phi\phi'}\overline\Delta)\overline{\mathcal H}^{-1}\mathcal S
			+o_P\big((NT)^{-1/2}\big).
			$
			Taking the conditional expectation of the leading term on the RHS, and invoking the information matrix identity ($\mathbb E_\phi\!\big[\sqrt{NT}\,\mathcal S\mathcal S'\big]=\overline{\mathcal H}$), we obtain 
			\[
			\mathbb E_\phi\!\left[ \mathcal S'\overline{\mathcal H}^{-1}(\partial_{\phi\phi'}\overline\Delta)\overline{\mathcal H}^{-1}\mathcal S \right] 
			= \mathrm{tr}\!\left( \overline{\mathcal H}^{-1}(\partial_{\phi\phi'}\overline\Delta)\overline{\mathcal H}^{-1} \, \mathbb E_\phi[\mathcal S \mathcal S'] \right)
			= \tfrac{1}{2\sqrt{NT}}\;\mathrm{tr}\!\big((\partial_{\phi\phi'}\overline\Delta)\,\overline{\mathcal H}^{-1}\big).
			\]
			Since the quadratic form converges in probability to its expectation, we conclude
			\begin{equation}
				\tfrac12\,\mu_\phi'(\partial_{\phi\phi'}\overline\Delta)\mu_\phi
				= \tfrac{1}{2\sqrt{NT}}\mathrm{tr}\!\big((\partial_{\phi\phi'}\overline\Delta)\,\overline{\mathcal H}^{-1}\big) + o_P\big((NT)^{-1/2}\big).
				\label{eq:mu-phi-part}
			\end{equation}

			\underline{\textit{3) The $R_\Delta$ term.  }} 
			For the $R_\Delta$ term, by cumulant expansion, 
			$
			\frac{\langle R_\Delta(\beta,\phi) e^{R_{\mathcal L}(\beta,\phi)}\rangle_\theta}{\langle e^{R_{\mathcal L}(\beta,\phi)}\rangle_\theta}
			= \sum_{m\ge0}\frac{1}{m!}\kappa\big(R_\Delta(\beta,\phi),R_{\mathcal L}(\beta,\phi)^{\times m}\big).
			$
			Under our assumptions, Lemma \ref{lemma:remainders}(iii) ensures it is bounded by $o_P((NT)^{-1/2})$.
			
			\underline{\textit{4) Final collections.  }} 
			Combining the above results, the last three terms in \eqref{eq:APE-mean-master} are given by
			$
			\tfrac{1}{\sqrt{NT}}\;\mathrm{tr}\!\big((\partial_{\phi\phi'}\overline\Delta)\,\overline{\mathcal H}^{-1}\big)
			+o_P\big((NT)^{-1/2}\big).
			$
			By Lemma \ref{lemma:FWD1}, $\overline{\mathcal H}^{-1}$ can be precisely approximated by  the inverse of the diagonalized $\overline{\mathcal H}$. Thus the last three terms in \eqref{eq:APE-mean-master}  are given by  
			$
			2\,\overline{\mathcal B}^\Delta + o_P((NT)^{-1/2}),
			$
			as exactly stated in Part (i).

			\underline{\#\textit{Part (ii): plug–in APE at posterior means.  }}
			
			A mean-value expansion of \(\Delta(\widehat{\beta}^E,\widehat{\phi}^E)\) around the truth  \((\beta_0,\phi_0)\)  yields
			\begin{eqnarray}
				\Delta(\widehat{\beta}^E,\widehat{\phi}^E) - \Delta
				&=&
				(\widehat{\beta}^E-\beta_0)'\partial_\beta\overline\Delta
				+(\widehat{\phi}^E-\phi_0)'\partial_\phi\overline\Delta
				+\tfrac12 (\widehat{\phi}^E-\phi_0)'(\partial_{\phi\phi'}\overline\Delta)(\widehat{\phi}^E-\phi_0) \notag \\
				&&+R_{\Delta,2}(\widehat{\theta}^E), 	\label{eq:plug-in-expansion}
			\end{eqnarray}
			where 
			$
			R_{\Delta,2}(\widehat{\theta}^E) 
			= \tfrac{1}{2}(\widehat{\beta}^E-\beta_0)'(\partial_{\beta\beta'} \Delta(\breve{\beta},\phi_0))(\widehat{\beta}^E-\beta_0)
			+(\widehat{\beta}^E-\beta_0)'(\partial_{\beta\phi'} \Delta(\beta_0,\breve{\phi}))(\widehat{\phi}^E-\phi_0) 
			+\tfrac{1}{6} (\widehat{\phi}^E-\phi_0)'\left(\sum_{k=1}^{\dim\phi}  (\partial_{\phi\phi'\phi_k} \Delta(\beta_0,\breve{\phi})) (\widehat{\phi}^E-\phi_0)_{k}\right)  (\widehat{\phi}^E-\phi_0).
			$

			Let $\mu_\phi^{(1)} := \overline{\mathcal{H}}^{-1}\mathcal{S}$ denote the leading score term. We decompose the quadratic $\phi$–term in \eqref{eq:plug-in-expansion} as
			\begin{eqnarray*}
				\tfrac12(\widehat\phi^E-\phi_0)'(\partial_{\phi\phi'}\overline\Delta)(\widehat\phi^E-\phi_0)
				&=&\tfrac12\,\mu_\phi^{(1)\prime}(\partial_{\phi\phi'}\overline\Delta)\mu_\phi^{(1)}
				+\mu_\phi^{(1)\prime}(\partial_{\phi\phi'}\overline\Delta)\big(\widehat\phi^E-\phi_0-\mu_\phi^{(1)}  \big) \\
				&&+\tfrac12\big(\widehat\phi^E-\phi_0-\mu_\phi^{(1)}  \big)'(\partial_{\phi\phi'}\overline\Delta)\big(\widehat\phi^E-\phi_0-\mu_\phi^{(1)}  \big).
			\end{eqnarray*}
			From Theorem \ref{thm:asy_expansions}, the deviation from the leading score term satisfies $\| \widehat\phi^E-\phi_0-\mu_\phi^{(1)} \| = o_P((NT)^{-1/8})$. Also, under Assumption \ref{assumption:regularity}, $\|\partial_{\phi\phi'}\overline\Delta\| = \mathcal{O}_P((NT)^{\epsilon-1/2})$. 
			Consequently, the cross-term and the remainder term in the decomposition above are bounded by $o_P((NT)^{-1/8} (NT)^{\epsilon-1/2}) = o_P((NT)^{-1/2})$,  leaving only the leading quadratic form
			$$
			\tfrac12(\widehat\phi^E-\phi_0)'(\partial_{\phi\phi'}\overline\Delta)(\widehat\phi^E-\phi_0)
			=\tfrac12\,\mu_\phi^{(1)\prime}(\partial_{\phi\phi'}\overline\Delta)\mu_\phi^{(1)}
			+o_P((NT)^{-1/2}).
			$$
			This first term on the RHS is identical to the quadratic form analyzed in Part (i) (equation \ref{eq:mu-phi-part}). Thus
			$
			\tfrac12(\widehat\phi^E-\phi_0)'(\partial_{\phi\phi'}\overline\Delta)(\widehat\phi^E-\phi_0)
			=\overline{\mathcal B}^\Delta+o_P((NT)^{-1/2}).
			$
			Substituting this result back into \eqref{eq:plug-in-expansion}, we obtain the  asymptotic representation
			\[
			\Delta(\widehat\beta^E,\widehat\phi^E)-\Delta
			=\partial_{\beta'}\overline\Delta\,(\widehat\beta^E-\beta_0)
			+\partial_{\phi'}\overline\Delta\,(\widehat\phi^E-\phi_0)
			+\overline{\mathcal B}^\Delta
			+o_P\big((NT)^{-1/2}\big) + R_{\Delta,2}(\widehat{\theta}^E) .
			\]

			The bound $|R_{\Delta,2}(\widehat{\theta}^E)| = o_P((NT)^{-1/2})$ is established using arguments analogous to the proof of Theorem B.4 in \cite{fernandez2016individual}.
			For $\beta\in\mathscr B(r_\beta,\beta_0)$ and $\phi\in\mathscr B_q(r_\phi,\phi_0)$, the remainder is bounded by
			\begin{align*}
				|R_{\Delta,2}(\widehat{\theta}^E)|
				&\leq \tfrac{1}{2} \|\widehat{\beta}^E-\beta_0\|^2 \|\partial_{\beta\beta'} \Delta({\beta},\phi)\|
				+\|\widehat{\beta}^E-\beta_0\| \, \|\partial_{\beta\phi'} \Delta({\beta},\phi)\| \, \|\widehat{\phi}^E-\phi_0\| \\
				&\quad +\tfrac{1}{6} \|\widehat{\phi}^E-\phi_0\|^2 \left\| \textstyle\sum_{k=1}^{\dim\phi} (\partial_{\phi\phi'\phi_k} \Delta(\beta,{\phi})) (\widehat{\phi}^E-\phi_0)_{k}\right\|.
			\end{align*}
			From Theorem \ref{thm:asy_expansions}, the estimators satisfy $\|\widehat\beta^E-\beta_0\|= o_P((NT)^{-1/4})$, $\|\widehat\phi^E-\phi_0\| = \mathcal{O}_P(1)$, and $\|\widehat\phi^E-\phi_0\|_q = \mathcal{O}_P((NT)^{-1/4+1/(2q)})$.
			Furthermore, by Assumption \ref{assumption:regularity}(viii) and the norm embedding inequalities, we have 
			$ 	\Vert \partial_{\beta\beta}\Delta(\beta,\phi) \Vert= \mathcal{O}_P( 1 ),$ 
			$ \Vert \partial_{\phi\phi\phi}\Delta(\beta,\phi) \Vert_q= \mathcal{O}_P( (NT)^{\epsilon-1/2} )$,
			and 
			$
			\Vert \partial_{\beta\phi'}\Delta(\beta,\phi) \Vert\leq (\dim\phi)^{1/2-1/q}\Vert \partial_{\beta\phi'}\Delta(\beta,\phi) \Vert_q = \mathcal{O}_P( (NT)^{-1/4} )
			$.
			For the cubic term, exploiting the symmetry of $\partial_{\phi\phi'\phi_k} \Delta$, we bound the spectral norm of the contracted tensor by
			\begin{eqnarray*}
				\left\| \textstyle\sum_{k=1}  (\partial_{\phi\phi'\phi_k} \Delta(\beta,{\phi})) (\widehat{\phi}^E-\phi_0)_{k}\right\|
				&\leq& \|\partial_{\phi\phi\phi} \Delta(\beta,{\phi}) \|\widehat\phi^E-\phi_0\|_q \\
				&=& \mathcal{O}_P((NT)^{\epsilon-1/2-1/4+1/(2q)}) = o_P((NT)^{-1/2}).
			\end{eqnarray*}
			Substituting these rates back into the expression for $R_{\Delta,2}(\widehat{\theta}^E)$, each component is $o_P((NT)^{-1/2})$ or smaller.  This establishes the result.
		\end{proof}

		\begin{proof}[Proof of Theorem \ref{corollary:generic_prior}] 
Start from the gradient of the leading bias term in the integrated log–likelihood, namely
$\partial_\beta \mathcal D_{\mathcal L^I}$ and
$\partial_\phi \mathcal D_{\mathcal L^I}$, stated in \eqref{eq:biasbeta}–\eqref{eq:biasphi}.
			Substituting the diagonal approximation for $\overline{\mathcal H}^{-1}$ (Lemma~\ref{lemma:FWD1}),  using the sparsity of $\partial_{\phi\phi_k\phi_g}\overline{\mathcal L}$ (see Appendix \ref{sec:appA:NotationNorms}) and   invoking the information matrix identity ($\mathbb E_\phi\!\big[\sqrt{NT}\,\mathcal S\mathcal S'\big]=\overline{\mathcal H}$) , yields that 
			\begin{eqnarray*}
			\mathbb{E}_{\phi}(\partial_\beta \mathcal D_{\mathcal L^I})
				&=&\sum_{i=1}^{N}\tfrac{\sum_{t=1}^{T}\mathbb{E}_{\phi}(\partial_{\beta\pi^{2}}\ell_{it}) +\sum_{t=1}^{T}\sum_{\tau=1}^{T}\mathbb{E}_{\phi}(\partial_{\beta\pi}\ell_{i\tau}\partial_{\pi}\ell_{it})}
				{\sum_{\tau = 1}^{T}{\mathbb{E}_{\phi}( \partial_{\pi^{2}}\ell_{i\tau} )}} \\
				&& +
				\sum_{t=1}^{T}
				\tfrac{\sum_{i=1}^{N}\mathbb{E}_{\phi}(\partial_{\beta\pi^{2}}\ell_{it}) +\sum_{i=1}^{N}\sum_{j = 1}^{N}\mathbb{E}_{\phi}(\partial_{\beta\pi}\ell_{jt}\partial_{\pi}\ell_{it})}
				{\sum_{j = 1}^{N}{\mathbb{E}_{\phi}( \partial_{\pi^{2}}\ell_{jt} )}}     + o_P((NT)^{1/2}),  \\
	\mathbb{E}_{\phi}(\partial_{\alpha_i} \mathcal D_{\mathcal L^I})
				&=&\sum_{t=1}^{T}\left( \tfrac{\mathbb{E}_{\phi}(\partial_{\pi^{3}}\ell_{it}) +\sum_{\tau=1}^{T}\mathbb{E}_{\phi}(\partial_{\pi^{2}}\ell_{i\tau}\partial_{\pi}\ell_{it})}
				{\sum_{\tau = 1}^{T}{\mathbb{E}_{\phi}( \partial_{\pi^{2}}\ell_{i\tau} )}}
				+ 
				\tfrac{\mathbb{E}_{\phi}(\partial_{\pi^{3}}\ell_{it}) +\sum_{j = 1}^{N}\mathbb{E}_{\phi}(\partial_{\pi^{2}}\ell_{jt}\partial_{\pi}\ell_{it})}
				{\sum_{j = 1}^{N}{\mathbb{E}_{\phi}( \partial_{\pi^{2}}\ell_{jt} )}}  \right)   + o_P(1),  \\
	\mathbb{E}_{\phi}(\partial_{\gamma_t} \mathcal D_{\mathcal L^I})
				&=&\sum_{i=1}^{N}\left( \tfrac{\mathbb{E}_{\phi}(\partial_{\pi^{3}}\ell_{it}) +\sum_{\tau=1}^{T}\mathbb{E}_{\phi}(\partial_{\pi^{2}}\ell_{i\tau}\partial_{\pi}\ell_{it})}
				{\sum_{\tau = 1}^{T}{\mathbb{E}_{\phi}( \partial_{\pi^{2}}\ell_{i\tau} )}}
				+ 
				\tfrac{\mathbb{E}_{\phi}(\partial_{\pi^{3}}\ell_{it}) +\sum_{j = 1}^{N}\mathbb{E}_{\phi}(\partial_{\pi^{2}}\ell_{jt}\partial_{\pi}\ell_{it})}
				{\sum_{j = 1}^{N}{\mathbb{E}_{\phi}( \partial_{\pi^{2}}\ell_{jt} )}}  \right)   + o_P(1).
			\end{eqnarray*}
Applying the martingale difference restrictions
$\mathbb E_\phi(\partial_{\beta\pi} \ell_{i\tau} \partial_\pi \ell_{it})=\mathbb E_\phi(\partial_{\pi^2} \ell_{i\tau} \partial_\pi \ell_{it})=\mathbb E_\phi(\partial_{\pi} \ell_{i\tau} \partial_\pi \ell_{it})=0$ for  $t>\tau$ 
as well as cross-sectional orthogonality across $i\neq j$, we obtain
$
			\mathbb E_\phi(\partial_{\beta}\mathcal D_{\mathcal L^I})
=\sum_{i=1}^N\sum_{t=1}^T \Upsilon^\beta_{it} + o_P((NT)^{1/2}),
$
$
			\mathbb E_\phi(\partial_{\alpha_i}\mathcal D_{\mathcal L^I})
=\sum_{t=1}^T \Upsilon^\pi_{it} + o_P(1),
$ and
$
\mathbb E_\phi(\partial_{\gamma_t}\mathcal D_{\mathcal L^I})
=\sum_{i=1}^N \Upsilon^\pi_{it} + o_P(1).
$

Finally, by the WLLN for panel data, we have   $\partial_{\alpha_i}  \mathcal{D}_{\mathcal{L}^I}        = \sum_{t=1}^{T} \Upsilon_{it}^{\pi} +o_P\left(1\right)$,  
$\partial_{\gamma_t}  \mathcal{D}_{\mathcal{L}^I}    = \sum_{i=1}^{N} \Upsilon_{it}^{\pi}   +o_P\left(1\right)$, and
$\partial_{\beta}  \mathcal{D}_{\mathcal{L}^I}   = \sum_{i=1}^N \sum_{t=1}^T \Upsilon_{it}^{\beta}  +o_P((NT)^{1/2})$. 
To verify that Prior GE--1 satisfies \eqref{eqn:differentialsystem}, it suffices to differentiate its log-density at the truth, and evaluate the resulting expectations using the martingale difference property. For Prior GE--2, the same verification additionally invokes the Bartlett identity $\mathbb E_\phi(-\partial_{\pi^2}\ell_{it})=\mathbb E_\phi((\partial_{\pi}\ell_{it})^2\big)$. We omit the remaining algebraic details.
		\end{proof}

		\begin{proof}[Proof of Corollary \ref{corollary:asy_parameter}]
			The proof follows directly from the asymptotic expansions derived in Theorem \ref{thm:asy_expansions}, combined with the definition of the bias-reducing prior and the central limit theorems for the score and Hessian terms.
			
			\underline{\textit{\# Part (i).}} 
			Under a bias–reducing prior, by definition,
			$\left\| \mathcal{B}^{\beta} \right\| = o_P(1) $ and
			$\left\| \mathcal{B}^{\phi} \right\| = o_P((NT)^{-1/4})$.
			Substituting these rates into the expansion in
			Theorem \ref{thm:asy_expansions}(i) yields
			$
			\widehat{\beta}^E - \beta_0
			= \frac{1}{\sqrt{NT}}\,\overline{W}^{-1} U^{(0)}
			+ o_P \big((NT)^{-1/2}\big).
			$
			By Assumption \ref{assumption:main}, $U^{(0)}$ is a sum of independent
			contributions across $i$ and martingale difference sequences across $t$,
			and therefore by a relevant central limit theorem  \citep[see, e.g.,][]{hahn2011bias}, $U^{(0)} \xrightarrow{d} \mathcal{N}(0, \overline{W}_{\infty})$.
			Moreover, for $N,T$ large, $\overline{W}\to  \overline{W}_\infty$ with $\overline{W}_\infty$ nonsingular.
			An application of Slutsky's theorem then gives
			\[
			\sqrt{NT}(\widehat{\beta}^E - \beta_0) \xrightarrow{d} \overline{W}_{\infty}^{-1} \mathcal{N}(0, \overline{W}_{\infty}) = \mathcal{N}(0, \overline{W}_{\infty}^{-1}).
			\]

			\underline{\textit{\# Part (ii).}} 
			From Theorem \ref{thm:asy_expansions}(ii), the expansion for $\widehat{\phi}^E$ is
			\[
			\widehat{\phi}^E - \phi_0 = 
			\overline{\mathcal{H}}^{-1} \mathcal{S}  
			+ \tfrac{1}{\sqrt{NT}}\overline{\mathcal{H}}^{-1} [\partial_{\phi\beta'}\overline{\mathcal{L}}]\overline{W}^{-1} U^{(0)}
			+ \text{Bias Terms}
			+ R_\phi.
			\]
			The bias terms are of order smaller than the leading stochastic term $\overline{\mathcal{H}}^{-1} \mathcal{S}$, and hence are negligible. The remainder $R_\phi$ satisfies $\|R_\phi\| = o_P((NT)^{-1/4})$. We may write the expansion as
			\[
			(NT)^{1/4}(\widehat{\phi}^E - \phi_0) = \overline{\mathcal{H}}^{-1} [ (NT)^{1/4}\mathcal{S} ] + \mathcal{O}_P\big((NT)^{-1/4}\big) U^{(0)} + o_P(1).
			\]
			The second term on the RHS vanishes asymptotically because $U^{(0)} = \mathcal{O}_P(1)$. The behavior is dominated by the first term. 
			
			By Assumption \ref{assumption:main}, each elements of the score vector $\mathcal{S}$ are normalized sums of independent (across $i$) or weakly dependent (across $t$) variables. Therefore, a standard CLT for triangular arrays with mixing applied to the finite dimensional score subvector, i.e., for any fixed sub-vector of length $L\ge1$,  $(NT)^{1/4}\mathcal{S}_{1:L}\ \xrightarrow{d}\ \mathcal{N}\!\big(0,(\overline{\mathcal{H}}_\infty)_{1:L,1:L}\big)$, where the covariance follows from the information matrix equality.
			Since $\overline{\mathcal{H}}\to\overline{\mathcal{H}}_\infty$ and $\|\overline{\mathcal{H}}^{-1}\|=\mathcal{O}_P(1)$,
			Slutsky's theorem implies
			$$
			(NT)^{1/4}(\widehat{\phi}^E - \phi_0)_{1:L} \xrightarrow{d} \mathcal{N}\left(0, \left(\overline{\mathcal{H}}_{\infty}^{-1} \overline{\mathcal{H}}_{\infty} \overline{\mathcal{H}}_{\infty}^{-1}\right)_{1:L,1:L}\right) = \mathcal{N}\left(0, (\overline{\mathcal{H}}_{\infty}^{-1})_{1:L,1:L}\right).
			$$
		\end{proof}

		\begin{proof}[Proof of Corollary \ref{corollary:asy_APE}]
			Recall the sample average partial effects evaluated at the true parameters: $\Delta = (NT)^{-1} \sum_{i,t} \Delta_{it}(\beta_0, \phi_0)$, and let $\mathbb{E}[\Delta]$ be the population target. The estimation error can be decomposed into a \textit{parameter estimation} component and a \textit{sampling variation} component:
			$$
				\widehat{\Delta} - \mathbb{E}[\Delta] = (\widehat{\Delta} - \Delta) + (\Delta - \mathbb{E}[\Delta]).
			$$
			For the parameter estimation component, by Theorem~\ref{thm:APE_expansions}(i),
			$$
			\widehat{\Delta}^E-\Delta
			=
			(\partial_{\beta'}\overline{\Delta})(\widehat{\beta}^E-\beta_0)
			+
			(\partial_{\phi'}\overline{\Delta})(\widehat{\phi}^E-\phi_0)
			+
			2\overline{\mathcal{B}}^{\Delta}_{\infty}
			+
			o_P \big((NT)^{-1/2}\big).
			$$
			Under a bias–reducing prior, the leading biases in $\widehat{\beta}^E$ and $\widehat{\phi}^E$ are eliminated (as shown in Theorem \ref{thm:asy_expansions}). 
			Hence, the terms $(\partial_{\beta'}\overline{\Delta}) (\widehat{\beta}^E - \beta_0)$ and $(\partial_{\phi'}\overline{\Delta}) (\widehat{\phi}^E - \phi_0)$ are asymptotically centered at zero. 
			
			After subtracting the remaining bias $2\overline{\mathcal{B}}^{\Delta}_{\infty}$ and rescaling by $r_{NT}$, as shown in Theorem 4.2 of \cite{fernandez2016individual}, for  $r_{NT}\to\infty$,\footnote{The convergence rate $r_{NT}$ is determined by the slower of the parameter estimation and  sampling variation components, which, under Assumption \ref{assumption:ape}(iv), is governed by the sampling properties of the unobserved effects and covariates (see Remark 4 in \cite{fernandez2016individual}).}
			the sum of these variations converges to a normal distribution:
			\[
			r_{NT} \left( \left[ (\partial_{\beta'}\overline{\Delta}) (\widehat{\beta}^E - \beta_0) + (\partial_{\phi'}\overline{\Delta}) (\widehat{\phi}^E - \phi_0) \right] + (\Delta - \mathbb{E}[\Delta]) \right)
			\xrightarrow{d} \mathcal{N}(0, \overline{V}_{\infty}^{\delta}),
			\]
			where $\overline{V}_{\infty}^{\delta}$ is the asymptotic variance that accounts for both the sampling variation of the data (via $\Delta$) and the asymptotic variance of the parameter estimates, 
			
			Thus, we conclude
			$
			r_{NT}\left(\widehat{\Delta}^E - \mathbb{E}(\Delta) - 2\overline{\mathcal{B}}^{\Delta}_{\infty} \right) \xrightarrow{d} \mathcal{N}(0, \overline{V}_{\infty}^{\delta}).
			$
			The proof for the posterior plug–in APE $\Delta(\widehat{\beta}^E,\widehat{\phi}^E)$ is identical.
		\end{proof}

		\begin{lemma}\label{lemma:remainders}
			Suppose Assumption~\ref{assumption:regularity} holds. 
			Let $R_{\mathcal L}(\beta,\phi)$ and $R_{\Delta}(\beta,\phi)$ be defined by \eqref{eqn:Rtheta2} and \eqref{eqn:RDelta_decompose} in Appendix~\ref{app:remainders}, respectively.
			Then, uniformly over $\beta\in\mathscr B(r_\beta,\beta_0)$ and $\phi\in\mathscr B_q(r_\phi,\phi_0)$, the following bounds hold:
			\begin{enumerate}[label=(\roman*)]
				\item For any $m\ge 2$,
				$
				\big\|\kappa\big(\mathcal U_\beta, R_{\mathcal L}^{\times m}(\beta,\phi)\big)\big\| = o_P\big((NT)^{-1/2}\big)
				\text{, and }
				\big\|\kappa\big(\mathcal U_\phi, R_{\mathcal L}^{\times m}(\beta,\phi)\big)\big\| = o_P\big((NT)^{-1/4}\big).
				$
				\item For any $m\ge 1$,
				$
				\Big|\kappa\Big(\mathcal U_\phi'(\partial_{\phi\phi'}\overline\Delta)\mathcal U_\phi,\,
				R_{\mathcal L}^{\times m}(\beta,\phi)\Big)\Big|
				= o_P\big((NT)^{-1/2}\big).
				$
				\item For any $m\ge 0$,
				$
				\Big|\kappa\big(R_\Delta(\beta,\phi),\, R_{\mathcal L}(\beta,\phi)^{\times m}\big)\Big|
				= o_P\big((NT)^{-1/2}\big).
				$
			\end{enumerate}
		\end{lemma}
		\noindent The proof of Lemma~\ref{lemma:remainders} is given in Supplemental Appendix~\ref{sec:app_intermedLemma}.

		\begin{lemma}[Lemma D.1 of \cite{fernandez2016individual}]\label{lemma:FWD1}
			Suppose Assumption \ref{assumption:regularity} holds. Then,
			$
			\|\overline{\mathcal H}^{-1}- \sqrt{NT}\mathrm{diag}\!\big( \sum_t\mathbb E_\phi(-\partial_{\pi^2}\ell_{it}),\ \sum_i\mathbb E_\phi(-\partial_{\pi^2}\ell_{it})\big)^{-1}\|_{\max}
			=\mathcal O_P((NT)^{-1/2}).
			$
		\end{lemma}

	\end{appendix}

	
	\bibliographystyle{qe} 
	\bibliography{ref_main}  

\end{document}